\documentclass[a4paper,hidelinks, 12pt]{article}
\usepackage[utf8]{inputenc}
\usepackage{a4wide}
\usepackage{amsmath}
\usepackage{amsfonts}
\usepackage{amssymb}
\usepackage{amsthm}
\usepackage{mathabx}
\usepackage{array}
\usepackage{forest}
\usepackage{rotating}
\usepackage{thmtools}
\usepackage{thm-restate}

\usepackage{subfigure}
\usepackage{cite}
\usepackage{xcolor}
\usepackage{soul}
\usepackage{cancel}
\usepackage[toc,page]{appendix}
\numberwithin{equation}{section}
\usepackage{mathtools}
\usepackage{hyperref}
\usepackage{url}
\usepackage{enumitem}
\usepackage{algorithm}
\usepackage{float}
\usepackage{lipsum}
\usepackage{MnSymbol}
\allowdisplaybreaks

\newtheorem{proposition}{Proposition}
\newtheorem{corollary}{Corollary}

\theoremstyle{plain}

\newcommand{\La}{\mathfrak}
\newcommand{\Lg}{\mathsf}

\newcommand{\N}[1]{\mathbb{#1}}

\newcommand{\nn}{\nonumber}
\newcommand{\mc}{\mathcal}
\newcommand{\on}{\operatorname}
\newcommand{\spa}{\operatorname{span}}
\newcommand{\pr}{\operatorname{pr}}

\newcommand{\rank}{\operatorname{rank}}
\begin{document}
\begin{titlepage}
\begin{center}

{\large \bf {Scalar Lie point symmetries of the Standard Model with one or two real gauge singlets}}

\vskip 1cm

M.~Aa.~Solberg\footnote{E-mail: marius.solberg@ntnu.no} 

\vspace{1.0cm}

Department of Structural Engineering, \\ Norwegian University of Science and Technology, \\
Trondheim, Norway\\

\end{center}

\vskip 3cm

\begin{abstract}
We present a classification of all scalar Lie point symmetries of the Standard Model with
one or two real gauge-singlet scalars (SM+S and SM+2S). By analyzing the associated field equations, we identify all realizable and inequivalent Lie point
symmetry algebras of these models, distinguishing strict variational, variational (including
divergence symmetries), and Euler--Lagrange cases. In addition, we devise efficient algorithms that, for any given numerical instance of the models,
determine the Lie point symmetry algebra in each of the three categories by a parameter-based
decision procedure using affine reparametrizations and simple parameter tests, thereby avoiding
explicit symmetry analysis and the need to derive and solve the determining equations.
Finally, we prove several relevant general results, including a characterization of the three disjoint types of Lie point symmetry generators---strict variational, divergence, and non-variational---for a broad class of Lagrangians with potentials, including the SM+S and SM+2S.

\end{abstract}

\end{titlepage}

\setcounter{footnote}{0}

\tableofcontents

\section{Introduction}
\label{sect:intro}
Real scalar singlet extensions of the Standard Model (SM) are among the simplest Higgs-portal scenarios, in which new physics couples to the SM only through Higgs doublets \cite{SILVEIRA1985136,McDonald1994GaugeSS,Burgess_2001,ARCADI20201}. They can simultaneously offer a viable dark matter candidate and induce a strongly first-order electroweak phase transition, as required for successful electroweak baryogenesis and for explaining the observed matter--antimatter asymmetry of the Universe \cite{Ahriche_2007,Espinosa:2011ax,Ellis:2022lft}.

Adding one real scalar singlet to the SM (often denoted SM+S) can provide a dark matter candidate if a stabilizing symmetry is imposed, most commonly an unbroken $\mathbb{Z}_2$ symmetry $s\to -s$. In that case the renormalizable Higgs-portal interaction is $\lambda_{s\phi}\,s^2\,\Phi^\dagger\Phi$, where $s$ is the singlet and $\Phi$ the Higgs doublet. Although phenomenologically restricted, the $\mathbb{Z}_2$-symmetric SM+S remains a viable dark matter model in limited regions of parameter space \cite{yu2025newconstraintssingletscalar}. The SM+S without imposing a $\mathbb{Z}_2$ symmetry has also been considered for its impact on the electroweak phase transition and its relevance for electroweak baryogenesis, see e.g.~\cite{Lewis_2017,Chen_2017,Ellis:2022lft}.

Extensions with two real singlets (SM+2S) allow for a richer phenomenology and a larger set of realizable symmetry groups. The added freedom can furthermore strengthen a first-order electroweak phase transition and accommodate a dark matter candidate, while remaining consistent with experimental constraints over a wider region of parameter space \cite{Tofighi_2015,Shajiee_2019}.
 Beyond discrete stabilizing symmetries, continuous global symmetries acting on the singlets may also lead to phenomenologically viable scenarios. For instance, in the SM+2S an $\Lg{SO}(2)$ rotation symmetry is equivalent to describing the two real singlets as a single complex singlet with a global $\Lg{U}(1)$ symmetry, and such scenarios can yield pseudo-Goldstone boson dark matter when the symmetry is softly broken \cite{PhysRevD.79.015018,Kannike:2019wsn,Arina_2020}. More generally, symmetries constrain model building by reducing the number of independent parameters in the scalar potential, and continuous symmetries may imply conserved currents.

A systematic classification of the discrete and continuous variational symmetries admitted by a
model provides a catalogue of distinct model classes. Variational (Noether) symmetries are
symmetries of the action integral and are therefore, in the absence of quantum anomalies,
inherited by the quantized theory. In a suitable field basis, the resulting classes are
characterized by symmetry-imposed relations among the parameters, which are stable under
renormalization-group evolution and can be explored both theoretically and phenomenologically. Non-variational (also called dynamical) symmetries are also important in quantum field
theory \cite{Skinner:QFT2:Chap6}. A celebrated example is electromagnetic duality, which is a
symmetry of the Euler--Lagrange equations (i.e.\ the equations of motion) but not of the action;
see, e.g., \cite{AlvarezGaumeZamora:1997}.

Among the most important symmetries are Lie point symmetries, i.e.\ continuous transformations
connected to the identity that act on the spacetime and field variables and, when prolonged to
derivatives, map solutions of the field equations to solutions.

Lie point symmetries may be divided into strict variational (SVS), divergence (DS), and
non-variational symmetries (NVS), where SVS and DS are variational symmetries. Infinitesimally,
these three symmetry types correspond to three nested Lie algebras; namely, the strict variational
symmetry algebra (spanned by the SVS generators), the variational symmetry algebra (spanned by
the SVS and DS generators), and the Euler--Lagrange symmetry algebra (spanned by the SVS, DS,
and NVS generators). The latter algebra contains all Lie point symmetry generators of the field equations.

The aim of this work is to classify all realizable Lie point symmetry algebras of the three types
mentioned above in the generic SM+S and SM+2S models, and to provide practical algorithms that,
for any given point in the parameter space of either model, determine the resulting symmetry
algebras without resorting to a full explicit symmetry analysis. We achieve this by applying Lie symmetry analysis of systems of partial differential equations (PDEs) \cite{Lie1893} to the models' Euler--Lagrange equations, i.e.\ their field equations.
This approach was previously applied to the two-Higgs-doublet model (2HDM) in \cite{Solberg:2025ybf}, which also provides pedagogical examples and further discussion of Lie symmetry analysis in the context of multi-scalar and particle physics models.

The results of the present work may be summarized as follows. The classification
shows that, apart from the ubiquitous hypercharge generator \(X_Y\), all
additional scalar generators are affine transformations of the singlet fields,
namely shifts and linear transformations of the singlet field space. This is in
line with the natural expectation that, in models with polynomial singlet
potentials of degree at most four and canonical kinetic terms, continuous scalar
point symmetries should be affine field-space transformations rather than
nonlinear transformations of the singlet fields. More precisely, the
classification determines, for the SM+S and SM+2S, which Lie algebras spanned by
these affine generators occur as maximal symmetry algebras in the strict
variational, variational, and Euler--Lagrange categories.

\paragraph{Outline of the article}
\label{sec:OutlineOfTheArticle}
In Sections \ref{sec:PointSymmetriesOfSystemsOfPDEs}--\ref{sec:SymmetriesOfTheActionMcS} we review standard Lie point symmetry theory for systems of PDEs, including the connection to symmetries of the action. In Section \ref{sec:GeneralResultsforSymmetryClassification} we state several general results for theories with scalar potentials, culminating in a theorem that characterizes the three categories of Lie point symmetries for a wide class of models; see Corollary \ref{C:charVarScalSyms}. In Section \ref{sec:SMS} we carry out the Lie point symmetry classification of the SM+S, while Section \ref{sec:SM2S} classifies the Lie point symmetries of the SM+2S. An efficient algorithm for determining the symmetries of a given SM+2S potential is presented in Section \ref{sec:ASimpleAlgorithmForDecidingSM2SSymmetry}. Finally, Section \ref{sec:SummaryAndOutlook} contains a summary and outlook, Appendix~\ref{app:TechnicalProofs} contains the proofs of the technical results from Section~\ref{sec:GeneralResultsforSymmetryClassification}, Appendix~\ref{app:Leaves13to31} collects symmetry calculations for SM+2S parameter cases that do not lead to additional symmetries, and Appendix~\ref{sec:SymmetryAnalysisWithLgSO2AdaptedSM2SPotential} presents an alternative SO(2)-adapted formulation of the SM+2S potential.

\paragraph{Conventions and notation}
\label{sec:ConventionsAndNotation}
In this article we adopt the mathematicians' convention for Lie algebras, in which the generators are chosen so that $[X,Y]\in \La{g}$ (rather than $[X,Y]= iZ$ with $Z\in\La{g}$).
 Moreover, 
$d_\mu\equiv d/dx^\mu$ denotes the total derivative, while $D_\mu$ is reserved for the covariant derivative \eqref{E:covarDerDef}, unless the index is a multi-index $J$, see section \ref{sec:InfinitesimalGenerators}. 
Since a given abstract Lie algebra may act on the scalar fields in multiple
inequivalent ways, we distinguish abstract Lie-algebra isomorphism, denoted
\(\La{g}\cong\La{h}\), from the stronger equivalence
\(\La{g}\eqsim_O\La{h}\). The latter means that the algebras are abstractly
isomorphic and that, after an affine reparametrization with orthogonal linear
part, possibly the identity, one can choose bases of generators for the two
algebras such that corresponding generators act identically on the scalar
fields. The notion of equivalence encoded by \(\eqsim_O\) is made precise in
Section~\ref{sec:AffineReparametrizations}. This distinction will be relevant
below, where abstractly isomorphic but inequivalent field actions occur; see,
e.g.~\eqref{E:LieAlgX1FreeMassiveSMS}--\eqref{E:LieAlgX3FreeMassiveSMS}
in Section~\ref{sec:DeterminingEquationsSMS}.
Moreover, symbols such as $\La{a}(n)$ will be used to denote specific realizations of algebras by vector fields on the scalar field space, rather than the corresponding abstract Lie algebras.
Finally, repeated indices are implicitly summed over (Einstein's summation convention). 

\section{Lie symmetry analysis of PDEs}
\label{sec:LieSymmetryAnalysisOfPDEs}
In this section we review the aspects of Lie symmetry theory for systems of PDEs that are relevant for our analysis.
The presentation in Sections~\ref{sec:PointSymmetriesOfSystemsOfPDEs}--\ref{sec:SymmetriesOfTheActionMcS} partly overlaps with our earlier review in \cite{Solberg:2025ybf}, but the treatment here is condensed, with a few points elaborated.
The underlying material is standard; see, e.g., \cite{olver1998applications,OlverLecturesLGaDE,hydon2000symmetry,bluman2010applications,cantwell2002introduction}.

\subsection{Point symmetries of systems of PDEs} 
\label{sec:PointSymmetriesOfSystemsOfPDEs}

Consider an $n$th–order system of PDEs
\begin{align}\label{E:nthOrdPDEs}
  \Delta_i(x,y,y^{(1)},\ldots, y^{(n)}) = 0, \qquad i = 1,\ldots,m,
\end{align}
with $d$ independent variables $x=(x^0,\ldots, x^{d-1})$ and
$q$ dependent variables $y=(y^1,\ldots,y^q)$. Here $y^{(k)}$ collectively denotes all $k$th–order partial derivatives of the $y^j$ with respect to the $x^\mu$. We write the system compactly as $\Delta=0$.

A \emph{point symmetry} $S$ of \eqref{E:nthOrdPDEs} is a diffeomorphism of the space of independent and dependent variables (a point transformation)
that maps solutions to solutions:
\begin{align}
  S : U \subset \N{R}^{d+q} \to  \N{R}^{d+q}, \qquad S(x,y) = (\hat x,\hat y),
\end{align}
for some open set $U$, where its action is prolonged to the derivatives so that
\begin{align}
  S\!\left(\frac{\partial^k y^i}{\partial x^{\mu_1}\cdots \partial x^{\mu_k}}\right)
  = \frac{\partial^k \hat y^i}{\partial \hat x^{\mu_1}\cdots \partial \hat x^{\mu_k}},
\end{align}
and the transformed system
\begin{align}\label{E:nthOrdPDEsTransf}
  \Delta_i(\hat x,\hat y,\hat y^{(1)},\ldots,\hat y^{(n)}) = 0, \qquad i=1,\ldots,m,
\end{align}
holds whenever \eqref{E:nthOrdPDEs} holds. In compact notation,
\begin{align}\label{E:invarCond}
  \Delta = 0 \;\Rightarrow\; \hat{\Delta} = 0,
\end{align}
where $\hat{\Delta} \equiv \Delta(\hat x,\hat y,\hat y^{(1)},\ldots,\hat y^{(n)})$.

\subsection{Prolongations of infinitesimal generators} 
\label{sec:InfinitesimalGenerators}

An infinitesimal generator of a point transformation is a vector field
\begin{align}\label{E:infGenPoint}
  X = \xi^\mu(x,y)\,\partial_{\mu} + \eta^i(x,y)\,\partial_{y^i},
\end{align}
with $\partial_{\mu} \equiv \partial_{x^\mu}$, $\mu = 0,\ldots,d-1$ and $i=1,\ldots,q$. The one-parameter group $S_\epsilon=\exp(\epsilon X)$ generated by $X$ acts on $z=(x,y)$ by
\begin{align}
  S_\epsilon(z)=\hat z = z + \epsilon X(z) + \mathcal{O}(\epsilon^2),
\end{align}
for $|\epsilon|<\epsilon_0$, where $\epsilon_0>0$ may be finite or infinite.

The $k$th prolongation of $X$ extends the action of $X$ to derivatives of $y$ up to order $k$,
\begin{align}\label{E:kProlong}
  \operatorname{pr}^{(k)}X
  = X + \sum_{1 \le |J| \le k} \eta^i_J\,\partial_{y^i_J},
\end{align}
where $J=(j_0,\ldots,j_{d-1})$ is a multi–index of length
$|J|$ and $y^i_J$ is the derivative with respect to $J$;
\begin{align}
  |J| &= j_0 + \cdots + j_{d-1}, \label{E:|J|} \\
  y^i_J &\equiv 
  \frac{\partial^{|J|} y^i}{(\partial x^0)^{j_0}\cdots(\partial x^{d-1})^{j_{d-1}}}. \label{E:y^i_J}
\end{align}
The prolonged coefficients $\eta^i_J$ are given by
\begin{align}
  \eta^i_J &= D_J(Q^i) + \xi^\mu y^i_{J,\mu}, \\
  D_J &= \Big(\frac{d}{dx^0}\Big)^{j_0}\cdots\Big(\frac{d}{dx^{d-1}}\Big)^{j_{d-1}},\label{E:itTotder} \\
  Q^i &= \eta^i - \xi^\mu y^i_{,\mu}, \label{E:characteristicDef}
\end{align}
where $d/dx^\mu$ denotes the total derivative and $Q^i$ is known as the \emph{characteristic} of $X$. At this stage $X$ is a candidate point symmetry generator; it becomes a symmetry generator only after imposing the symmetry condition \eqref{E:invarCond}.

If $\xi^\mu=0$ for all $\mu$, the vector field $X$ is in \emph{evolutionary form},
\begin{align}\label{E:evolRepr}
  X=X_Q \equiv Q^i \partial_{y^i}.
\end{align}
The vector field $X_Q$ is called the \emph{evolutionary representative} of $X$. If $Q$ depends on derivatives, then $X_Q$ is a generalized (non-point) vector field \cite{olver1998applications}.

Throughout this work we restrict to \emph{scalar Lie point transformations}, i.e.\ point transformations that leave the independent variables fixed. Hence $\xi=0$ and
\begin{align}
  Q^i=\eta^i=\eta^i(x,y),
\end{align}
so that $X$ is a point vector field in evolutionary form. For vector fields in evolutionary form, the prolongation simplifies to
\begin{align}
  \operatorname{pr}X=\operatorname{pr}X_Q
  = \sum_{i,J} (D_J Q^i)\,\partial_{y^i_J}
  = \sum_{i,J} (D_J \eta^i)\,\partial_{y^i_J}.
\end{align}

We moreover define the (formal) infinite prolongation 
\begin{align}
  \operatorname{pr}X \equiv \operatorname{pr}^{(\infty)}X
\end{align}
by \eqref{E:kProlong} with $k=\infty$.

Let now $\Lg{G}$ be a connected Lie group with Lie algebra $\La{g}$, acting at least locally (i.e., possibly only defined in a neighbourhood of the identity) on the variables $z=(x,y)$ of a fully regular (locally solvable with non-vanishing Jacobian) system of $m$ PDEs $\Delta=0$. Then it can be shown that $\Lg{G}$ is a (local) symmetry group of $\Delta=0$ if and only if
\begin{align}\label{E:prDelta=0PDEsymCond}
  \bigl(\operatorname{pr}X(\Delta_i)\bigr)\big|_{\Delta = 0} = 0,
  \qquad \forall i\in \{1,\ldots,m \},
\end{align}
for all infinitesimal generators $X \in \La{g}$. In practice, only the prolongation up to the system order $n$ is needed, i.e.~$\operatorname{pr}^{(n)}X$.

The condition \eqref{E:prDelta=0PDEsymCond} is the \emph{linearized symmetry condition}, and the resulting equations are the \emph{determining equations} for the symmetry algebra $\La{g}$. Combining \eqref{E:kProlong} and \eqref{E:prDelta=0PDEsymCond} yields an overdetermined linear system of PDEs for the coefficients $\xi^\mu(x,y)$ and $\eta^i(x,y)$, which can typically be solved explicitly to obtain the infinitesimal symmetries. In this work we use the \texttt{Mathematica} package \texttt{SYM} \cite{dimas2005sym} to compute the determining equations.
    
\subsection{Variational symmetries}
\label{sec:SymmetriesOfTheActionMcS}

Consider the action
\begin{align}
  \mc S[y]
  = \int_{\Omega} \mc L(x,y,\ldots,y^{(n)})\,dx^0\cdots dx^{d-1}.
\end{align}
A one-parameter Lie point transformation group $S_\epsilon=\exp(\epsilon X)$ is a \emph{variational}
(Noether) symmetry group if it leaves the action $\mc S$ invariant up to a boundary term.
Infinitesimally (i.e.\ to first order in $\epsilon$), this holds if and only if\footnote{For strict variational symmetries ($\beta^\mu\equiv 0$) the equivalence also holds for
finite $\epsilon$, at least locally for $\epsilon\in\langle-\epsilon_0,\epsilon_0\rangle$ where $S_\epsilon$ is
defined, whereas the situation is in general more complicated for divergence symmetries; see
Theorem~4.12 and the discussion following Def.~4.33 in \cite{olver1998applications}, respectively.
Nevertheless, whenever \eqref{E:condVarSym} holds, $S_\epsilon$ yields a (local) symmetry group of
the associated Euler--Lagrange equations; cf.~Theorem~4.34 in \cite{olver1998applications}.}
\begin{align}\label{E:condVarSym}
  \on{pr}X(\mc L) + \mc L\,d_\mu \xi^\mu = d_\mu \beta^\mu,
\end{align}
where $\beta^\mu$ is a local function (i.e.\ a function on an appropriate jet space; see e.g.\
\cite{Solberg:2025ybf}) and $d_\mu$ denotes the total derivative with respect to $x^\mu$,
\begin{align}\label{E:dmuexpanded}
  d_\mu
  = \partial_\mu
    + y^i_{,\mu}\,\partial_{y^i}
    + y^i_{,\mu\nu}\,\partial_{y^i_{,\nu}}
    + \cdots\,.
\end{align}

If \eqref{E:condVarSym} holds with $\beta^\mu=0$, we call $X$ a \emph{strict variational symmetry},
\begin{align}\label{E:condStrictVarSym}
  \on{pr}X(\mc L) + \mc L\,d_\mu \xi^\mu = 0,
\end{align}
and, if \eqref{E:condVarSym} holds but \eqref{E:condStrictVarSym} does not, we call $X$ a \emph{divergence symmetry}. In the latter case, the action changes by a boundary term at first order in $\epsilon$,
\begin{align}
  \hat{\mc S}
  = \mc S + \epsilon \int_{\partial\Omega} \beta^\mu\,dF_\mu
  + \mc O(\epsilon^2),
\end{align}
with $dF_\mu$ the outward normal surface element on the boundary $\partial\Omega$.
A first-order boundary term in $\epsilon$ is sufficient to guarantee an on-shell conserved Noether current, cf.~\eqref{E:NoetherCurrent} below, 
and if the boundary term vanishes (e.g.\ under suitable fall-off or periodicity conditions), then the action is invariant to first order in $\epsilon$.


Adding a total divergence does not change the Euler--Lagrange equations: $\mc L'=\mc L+d_\mu \beta^\mu$ yields the same field equations as $\mc L$.
Define the Euler operator $E=(E_1,\ldots,E_q)$ by
\begin{align}\label{E:EulerOpComps}
  E_i
  = \sum_{J}(-1)^{|J|}D_J\frac{\partial}{\partial y^i_J}
  = \frac{\partial}{\partial y^i}
    - d_\mu\frac{\partial}{\partial y^i_{,\mu}} + \cdots,
  \qquad i=1,\ldots,q,
\end{align}
so that the Euler--Lagrange equations read
\begin{align}\label{E:ELeqs}
  E(\mc L)=0.
\end{align}
Moreover, a differential function $f$ is a total divergence if and only if \cite{olver1998applications}
\begin{align}\label{E:EtotDiv}
  E(f)=0.
\end{align}
Hence, if $\on{pr}X(\mc L)+\mc L\,d_\mu\xi^\mu=f$, then $X$ is a divergence symmetry precisely when \eqref{E:EtotDiv} holds, in which case $f=d_\mu\beta^\mu$ for some local $\beta^\mu$.

For a fixed allowed parameter choice, let
\(\La g_{\text{svar}}\), \(\La g_{\text{var}}\), and
\(\La g_{\text{EL}}\) denote the corresponding full Lie algebras of strict
variational, variational, and Euler--Lagrange symmetry generators,
respectively, with the variational generators satisfying
\eqref{E:condVarSym}. Then
\begin{align}\label{E:symAlgIncl}
  \La g_{\text{svar}}\subseteq \La g_{\text{var}}\subseteq \La g_{\text{EL}}.
\end{align}
Elements of \(\La g_{\text{EL}}\setminus \La g_{\text{var}}\) are called
\emph{non-variational} symmetry generators. We will call a symmetry algebra \emph{realizable} if, for at least one allowed
parameter choice, it is maximal within one of these three symmetry categories.
Subalgebras are therefore not regarded as realizable merely because they are
contained in a realizable algebra; they are realizable only if, for at least
one allowed parameter choice, they are themselves maximal in one of the three
symmetry categories.

The maximality requirement is included because the aim is to classify potentials
by their full symmetry content. For a fixed potential and a fixed symmetry category, a proper subalgebra of the
maximal symmetry algebra in that category is also a symmetry algebra of the same
potential, but it is not the maximal one and therefore does not characterize the
potential's full symmetry content in that category.

By Noether's theorem, variational symmetries yield conserved currents. For a first-order Lagrangian $\mc L(x,y,y^{(1)})$, the current associated to $X=\xi^\mu\partial_{\mu}+\eta^i\partial_{y^i}$ can be written as
\begin{align}\label{E:NoetherCurrent}
  j^\mu
  = (\eta^i-\xi^\nu y^i_{,\nu})\frac{\partial\mc L}{\partial y^i_{,\mu}}
    + \xi^\mu \mc L - \beta^\mu,
  \qquad d_\mu j^\mu=0.
\end{align}
At the quantum level, variational symmetries typically lift to the path integral provided the measure is invariant (i.e.\ in the absence of anomalies).

\subsection{Characterization and preservation of symmetry types}
\label{sec:GeneralResultsforSymmetryClassification}
In this section we characterize the different symmetry types (strict variational, divergence, and
non-variational) and show that the symmetry type is preserved under affine reparametrizations of
the fields.

Let
\begin{align}
  \mc{L}=T-V
\end{align}
be the Lagrangian of a theory, where $T$ denotes the kinetic part (the sum of the kinetic terms),
while $V(\varphi_1,\ldots,\varphi_m)$ is a potential, i.e.\ a real polynomial whose variables
$\varphi_j$ form a subset of the dependent variables of the theory,
\begin{align}\label{E:varphiSubsety}
  \varphi=\{\varphi_1,\ldots,\varphi_m\}\subset y=\{y^1,\ldots,y^q\},
\end{align}
where, for simplicity, we (re-)order the $\varphi$'s first in $y$,
\begin{align}
  \varphi_i\equiv y^i, \quad \forall i\in \{1,2,\ldots, m\}.
\end{align}
Note that $\varphi$ may consist of any of the fields in $y$, although the intention is to take
$V(\varphi)=V(\phi)$, i.e.\ let $V$ be a scalar potential. Also note that the ``kinetic'' part $T$
may in fact consist of any terms complementary to $V$.

Moreover, we say that \(E(\mc{L})=0\) has a \emph{polynomial consequence} if it implies a relation of the form
\begin{align}\label{E:polCons}
  p(\varphi)=0,
\end{align}
where \(p\) is a non-zero polynomial in the same fields as the potential \(V\).

Furthermore, we define the complement 
\begin{align}\label{E:varphiDef}
	\varphi^c= y \setminus \varphi =\{y^i \}_{i=m+1}^q \equiv \{\varphi^c_j \}_{j=1}^{q-m},
\end{align}
 i.e.~all elements of $y$ that are not in $\varphi$.
We also let the set of all spacetime derivatives of all fields be given by
\begin{align}
	\on{der}(y)= \{\, y^i_J \mid 1\le i\le q,\ |J|\ge 1 \,\},
\end{align}
and we will assume that $T$ is a (finite) polynomial in the fields
and their spacetime derivatives of arbitrary order,
\begin{align}
	T\in \N{R}\{y\}\equiv \N{R}[y,y^{(1)}, y^{(2)},\ldots],
\end{align}
possibly with a constant term.
Finally, $V$ is assumed to be a real polynomial in the field variables $\varphi\subset y$ and to contain no derivatives. Thus
\begin{align}
  V = V(\varphi)\in \N{R}[\varphi].
\end{align}
Likewise, the notation $\eta^i\in \N{R}[y]$ means that $\eta^i$ is a real polynomial in the variables $y^j$.

We now present a substantially modified version of a theorem proved in \cite{Solberg:2025ybf}, which will be useful for characterizing strict variational, divergence, and non-variational Lie point symmetries.
 The assumptions are chosen to fit models of the type NHDM+KS (i.e.\ $N$-Higgs-doublet models with $K$ singlets): In such theories the kinetic part $T$ is either at least linear in derivatives or at least quadratic in $\varphi^c$ (e.g.\ the NHDM gauge boson terms, which are at least quadratic in the gauge fields when no derivatives are present). Nevertheless, the theorem applies to a much broader class of models with potentials.
\begin{restatable}{theorem}{PotentialSymmetryTheorem}\label{T:prXannTannV}
Let $\mc{L}=T-V$ be a Lagrangian with $T\in \N{R}\{y\}$ and $V(\varphi_1,\ldots,\varphi_m) \in \N{R}[\varphi]$ with $\varphi \subset y$,
and let the infinitesimal generator 
\[
X=\eta^i(y^1,\ldots,y^{q})\partial_{y^i}, 
\]
where $\eta^i \in \N{R}[y]$ for all $i$, be a symmetry of $E(\mc{L})=0$,
with constant terms
  $a_i=\eta^i(0)$. Moreover, let all terms in $T$ be either at least quadratic in elements of the set $\varphi^c$ or at least linear in the elements of $\on{der}(y)$. 
 Finally, assume that the system $E(\mc{L})=0$ has no polynomial consequences and that $V$ may contain linear terms $\alpha_i \varphi_i$.
Then the following hold: 
\begin{enumerate}[label=(\roman*)]
\item If $\on{pr}X(T)=d_\mu \beta^\mu$ for some (possibly vanishing) divergence $d_\mu \beta^\mu$, then $\on{pr}X(V)=a_i\alpha_i\in \N{R}$. 
 \item If $\on{pr}X(T)=0$ and $a_i\alpha_i=0$, then $X$ is a strict variational symmetry generator.
  \item If $\on{pr}X(T)=d_\mu \beta^\mu$ for some non-vanishing divergence $d_\mu \beta^\mu$, or if the 
	divergence is vanishing but $a_i\alpha_i\ne 0$, 
        then $X$ is a divergence symmetry generator.
\end{enumerate}
\end{restatable}
\begin{proof}
The proof is given in Appendix~\ref{app:TechnicalProofs}.
\end{proof}

The following proposition is, to a large extent, a converse to 
Theorem \ref{T:prXannTannV}:
\begin{restatable}{proposition}{PotentialSymmetryConverse}\label{P:prXannTannVconv}
Let $\mc{L}=T-V$ be a Lagrangian, where $V\in \N{R}[\varphi]$ with $\varphi\subset y$, $T\in \N{R}\{y\}$, and let
 \[
X=\eta^i(\varphi) \partial_{y^i}, 
\]
 be a symmetry of $E(\mc{L})=0$, where $\eta^i \in \N{R}[\varphi]$ for all $i$. 
Denote $a_i=\eta^i(0)$ and let $\alpha_i \varphi_i$ be the (possibly vanishing) linear terms of $V$.
Moreover, let all terms in $T$ be either at least quadratic in elements of the set $\varphi^c$ or at least linear in the elements of $\on{der}(y)$. 
Then the following hold:
\begin{enumerate}[label=(\roman*)]
 \item  If $X$ is a strict variational symmetry generator, then $\on{pr}X(T)=0$ and $\on{pr}X(V)=a_i\alpha_i=0$.
  \item  If $X$ is a divergence symmetry generator, then
	$\on{pr}X(T)$ is a total divergence, that is, $\on{pr}X(T)=d_\mu \beta^\mu$,
        while $\on{pr}X(V)=a_i\alpha_i \in \N{R}$. Here $d_\mu \beta^\mu$ does not contain a constant term, and $d_\mu \beta^\mu$ and $a_i\alpha_i$ cannot be simultaneously vanishing.
\end{enumerate}
\end{restatable}
\begin{proof}
The proof is given in Appendix~\ref{app:TechnicalProofs}.
\end{proof}

We now state and prove a corollary that characterizes the three types of scalar Lie point symmetries, again under assumptions tailored to models of the type NHDM+KS. However, the result applies more generally to evolutionary symmetries beyond the scalar case.
\begin{corollary}
\label{C:charVarScalSyms}
Let $\mc{L}=T-V$ be a Lagrangian, where $V\in \N{R}[\varphi]$ with $\varphi\subset y$, $T\in \N{R}\{y\}$, and let
 \[
X=\eta^i(\varphi) \partial_{y^i}, 
\]
 be a symmetry of $E(\mc{L})=0$, where $\eta^i \in \N{R}[\varphi]$ for all $i$. 
Denote $a_i=\eta^i(0)$ and let $\alpha_i \varphi_i$ be the (possibly vanishing) linear terms of $V$.
Assume moreover that $T$ is such that every term is either at least quadratic in elements of the set $\varphi^c$ or at least linear in elements of $\on{der}(y)$, and that the system $E(\mc{L})=0$ has no polynomial consequences.
Then the following hold:
\begin{enumerate}[label=(\roman*)]
 \item  $X$ is a strict variational symmetry generator if and only if $\,\on{pr}X(T)=0$ and $a_i\alpha_i=0$.
\item $X$ is a divergence symmetry generator if and only if
$\,\on{pr}X(T)=d_\mu \beta^\mu$ for some $\beta^\mu$, and at least one of the conditions
$\on{pr}X(T)=0$ and $a_i\alpha_i=0$ fails, i.e.\ $(d_\mu \beta^\mu\neq 0)\ \text{or}\ (a_i\alpha_i\neq 0)$.
	\item	If $X$ is a variational symmetry generator, then
	      $\on{pr}X(V)=a_i\alpha_i$.
		\item $X$ is a non-variational symmetry generator if and only if \\ $E (\on{pr}X(T))$ $\ne 0$.		
\end{enumerate}
\end{corollary}
\begin{proof}
(i) and (ii): The rightward implications follow from Proposition \ref{P:prXannTannVconv}, while the leftward implications follow from Theorem \ref{T:prXannTannV}. (iii) follows from Proposition \ref{P:prXannTannVconv}. For (iv), we by (i) and (ii) know that
the symmetry $X$ has to be non-variational if and only if
 $\on{pr}X(T)$ cannot be written as a total divergence, which is equivalent to $E(\on{pr}X(T))\ne 0$, cf.~\eqref{E:EtotDiv}.
\end{proof}
In Corollary \ref{C:charVarScalSyms} we did not include $\on{pr}X(V)=a_i\alpha_i$ on the right-hand side of the equivalences in (i) and (ii), since $\on{pr}X(V)=a_i\alpha_i$ would then appear as a condition to be proved every time we apply the leftward implications of (i) and (ii). This is not the case, as $\on{pr}X(V)=a_i\alpha_i$ follows from both sides of the two equivalences (i) and (ii); cf. Proposition \ref{P:prXannTannVconv} and Theorem \ref{T:prXannTannV}.

\subsubsection{Affine reparametrizations}
\label{sec:AffineReparametrizations}
In the next proposition, we show that affine field reparametrizations preserve the symmetry algebra and, moreover, preserve the symmetry type: strict variational (SVS), divergence (DS), and non-variational symmetries (NVS).
\begin{restatable}[Affine reparametrizations]{proposition}{AffineReparametrizationProp}\label{P:AffRepar}
Let
\begin{align}\label{E:GenRepar}
  \tilde y = A y + \gamma
\end{align}
be an affine reparametrization of the fields, where $A$ is a real, invertible $q\times q$ matrix and $\gamma\in\N{R}^q$.
Define the reparametrized Lagrangian $\tilde{\mc L}$ by
\begin{align}\label{E:Ly=Ltyt}
  \tilde{\mc L}\big(x,\tilde y,\tilde y^{(1)},\ldots,\tilde y^{(n)}\big)
  =\mc L\big(x,y,y^{(1)},\ldots,y^{(n)}\big),
\end{align}
where $y=A^{-1}(\tilde y-\gamma)$.
Then
\begin{align}
  X=\xi^\mu\partial_\mu+\eta^i\partial_{y^i}
\end{align}
is a point symmetry of $E(\mc L)=0$ if and only if
\begin{align}\label{E:Xtieti}
  \tilde X=\xi^\mu\partial_\mu+\tilde\eta^i\partial_{\tilde y^i},
  \qquad \tilde\eta^i = A_{ij}\eta^j,
\end{align}
is a point symmetry of $\tilde{E}(\tilde{\mc L})=0$. Moreover, $X$ and $\tilde X$ are of the same symmetry type (SVS/DS/NVS), and any symmetry algebra
$\La g=\spa(X_1,\ldots,X_n)$ is mapped to an isomorphic algebra
$\tilde{\La g}=\spa(\tilde X_1,\ldots,\tilde X_n)$, where $\theta(X)=\tilde X$ defines a Lie algebra isomorphism.
\end{restatable}
\begin{proof}
The proof is given in Appendix~\ref{app:TechnicalProofs}.
\end{proof}

For practical reasons, the matrix $A$ is usually chosen to be orthogonal, so that the kinetic terms retain their canonical form.
This is also the case for unitary reparametrizations of the doublets in NHDMs: A unitary transformation acting on the doublets induces an orthogonal transformation of their real component fields \cite{Olaussen:2010aq}.
For Higgs doublets, constant shifts spoil the canonical form of the kinetic/gauge sector and are therefore usually avoided, except when expanding about a vacuum expectation value (VEV) after spontaneous symmetry breaking (SSB).

Consider a set of fields $\phi\subset y$ with identical gauge quantum numbers. One may then perform field reparametrizations
\begin{align}
  \tilde\phi = A\phi + \gamma,
\end{align}
where $A$ is taken (when possible) to preserve the kinetic terms, and $\phi^c = y\setminus\phi$ denotes the complementary set of fields, cf.~\eqref{E:varphiDef}.
If $A$ preserves the kinetic form (e.g.\ $A\in \Lg{O}(n)$ for real scalars), then the kinetic sector is unchanged in the new basis,
\begin{align}
  T[\phi,\phi^c] = T[\tilde\phi,\phi^c].
\end{align}
The potential is typically a polynomial in $\phi$, and can be rewritten as a function of the new fields,
\begin{align}
  V(\phi) \equiv \tilde V(\tilde\phi),
\end{align}
so that it describes the same theory although its explicit functional form generally changes, $\tilde V\neq V$.
Such models exhibit what is commonly referred to as \textit{reparametrization freedom}. Models with extended scalar sectors provide typical examples.
We will exploit this freedom to choose convenient bases and to avoid double counting of equivalent symmetry realizations.


To decide when two symmetries are equivalent under a field reparametrization, we must express the transformed coefficients $\tilde\eta^i$ of Proposition~\ref{P:AffRepar} in terms of the transformed field variables $\tilde y$. All scalar Lie point symmetries of the models considered in subsequent sections are of the form
\begin{align}\label{E:Xeqepy}
  X=\eta^i(y)\,\partial_{y^i},
\end{align}
and will be affine, i.e.
\begin{align}\label{E:etaAff}
  \eta(y)=By+a,
\end{align}
for some constant matrix $B$ and constant vector $a$. Under the affine reparametrization
\begin{align}\label{E:GenRepar2}
  \tilde y = Ay+\gamma,
\end{align}
the coefficients transform as $\tilde\eta=A\eta$, and hence
\begin{align}
  \tilde\eta
  =A(By+a)=ABy+Aa.
\end{align}
Substituting $y=A^{-1}(\tilde y-\gamma)$ from \eqref{E:GenRepar2} yields
\begin{align}\label{E:etaTyT}
  \tilde\eta(\tilde y)
  =ABA^{-1}\tilde y - ABA^{-1}\gamma + Aa.
\end{align}

Given two affine scalar symmetries $X_1$ and $X_2$ with coefficients  $(B_1,a_1)$ and $(B_2,a_2)$, respectively, equation \eqref{E:etaTyT} shows that $X_1$ and $X_2$ are equivalent under an affine reparametrization \eqref{E:GenRepar2} if there exist $A$ and $\gamma$ such that
\begin{align}\label{E:affEquiv}
  B_1 = AB_2A^{-1}, \qquad
  a_1 = -AB_2A^{-1}\gamma + Aa_2.
\end{align}
We call a reparametrization induced by an orthogonal matrix $A=O$ an \emph{orthogonal reparametrization}. For singlets we will only identify symmetries as equivalent when they are related by an orthogonal reparametrization, since such transformations preserve the canonical form of the kinetic terms.\footnote{$\Lg{U}(N)$ reparametrizations of $N$ Higgs doublets correspond to orthogonal transformations on the real component fields, but orthogonal transformations beyond the realification of $\Lg{U}(N)$ do not preserve the canonical kinetic form of NHDMs \cite{Olaussen:2010aq}.}

Moreover, proportional characteristics $\eta$ and $k\eta$ with $k\in\N{R}\setminus\{0\}$ generate the same one-parameter symmetry group up to a rescaling of the group parameter. Hence, two affine scalar symmetries are \emph{orthogonally equivalent} if and only if there exist $k\in\N{R}\setminus\{0\}$, an orthogonal matrix $O$, and $\gamma$ such that
\begin{align}\label{E:orthEquivSym}
  kB_1 = OB_2O^{T},\qquad
  ka_1 = -OB_2O^{T}\gamma + Oa_2.
\end{align}

Finally, two point symmetry algebras $\La g$ and $\La h$ of $\Delta=0$ will be called orthogonally equivalent or just \emph{equivalent}, denoted 
\begin{align}\label{E:OrthEquiv}
	\La g\eqsim_O \La h,
\end{align}
if one can be mapped to the other by an orthogonal reparametrization, so that they describe the same symmetry realization in different field bases. Concretely, let $\La h=\spa(X_1,\ldots,X_n)$ and define
\begin{align}
  \tilde{\La h}=\spa(\tilde X_1,\ldots,\tilde X_n),
\end{align}
where each $\tilde X_i$ is obtained from $X_i$ by the reparametrization \eqref{E:GenRepar2} with $A=O$ (so $\tilde{\La h}$ depends on the choice of $O$ and $\gamma$), and is written in the transformed variables,
\begin{align}
  \tilde X_i(\tilde y)=\tilde\eta_i^{j}(\tilde y)\,\partial_{\tilde y^j}.
\end{align}
We then say that $\La g$ and $\La h$ are equivalent if there exists an orthogonal matrix $O$ and a shift vector $\gamma$ such that
\begin{align}\label{E:orthEqAlgebras}
  \La g=\tilde{\La h}\big|_{\tilde y\to y}\,.
\end{align}

\section{SM+S}
\label{sec:SMS}
In this section we determine all Lie point symmetries of the Standard Model augmented by a real scalar gauge singlet. Since the components of the SM Higgs doublet and the singlet field transform differently under $\Lg{SU}(2)_L$, no non-trivial mixing of these scalars is expected, and the scalar Lie point symmetries are therefore comparatively constrained. Moreover, aside from constant singlet shifts and rephasings of the Higgs doublet, the model admits no further reparametrization freedom in the scalar sector, so the symmetry analysis is less involved than e.g.~in the 2HDM \cite{Solberg:2025ybf}.

\subsection{Lagrangian}
\label{sec:Lagrangian}
The (scalar) SM+S Lagrangian can be written as
\begin{align}\label{E:SMSlag}
	\mc{L}_\text{SMS}=-\tfrac{1}{4}W^a_{\mu\nu}W^{a\mu\nu}
  -\tfrac{1}{4}B_{\mu\nu}B^{\mu\nu}
  +  (D_\mu\Phi)^\dagger (D^\mu\Phi)
 +\tfrac{1}{2}\partial_\mu s \partial^\mu s - V(\Phi,s),
\end{align}
where the covariant derivative and the gauge field strength tensors are given by
\begin{align}
D_\mu \Phi
&= \Big(\partial_\mu + i g \frac{\sigma^a}{2} W^a_\mu + i g' \frac{1}{2} B_\mu\Big)\Phi,
\label{E:covarDerDef} \\
W^a_{\mu\nu}
&= \partial_\mu W^a_\nu - \partial_\nu W^a_\mu + g \epsilon^{abc}W^b_\mu W^c_\nu,
\label{E:Wfst}\\
B_{\mu\nu}
&= \partial_\mu B_\nu - \partial_\nu B_\mu,
\label{E:Bfst}
\end{align}
and $\sigma^a$ are the Pauli matrices ($a=1,2,3$).

The most general, renormalizable SM+S potential can then be defined as
\begin{align}\label{E:SMSpot}
V(\Phi,s)&=	-\mu^2 \Phi^\dagger \Phi + 
\lambda (\Phi^\dagger \Phi)^2 + \alpha s -\mu_s s^2+ \kappa s^3 + \lambda_s s^4 \nn \\
&+ \kappa_{s\phi} \Phi^\dagger \Phi s + \lambda_{s\phi} \Phi^\dagger \Phi s^2. 
 \end{align}
Here, $s$ is the real, scalar gauge singlet, while $\Phi$ is the SM Higgs doublet, 
 \begin{align}\label{E:HiggsDoubletSM}
	 \Phi = \frac{1}{\sqrt{2}}
\begin{pmatrix}
\phi_{1} + i \phi_{2} \\
\phi_{3} + i \phi_{4}
\end{pmatrix}, 
\end{align}
where the $\phi$'s are real, scalar fields.

\subsubsection{Linear terms}
\label{sec:LinearTermsSMS}
The parameter $\alpha$ can in many cases be eliminated by a constant field shift, 
\begin{align}\label{E:shiftsSMS}
	s \;\to\; s - \beta ,
\end{align}
with real $\beta$. 
 Under \eqref{E:shiftsSMS}, the tadpole coefficient transforms as
\begin{align}
	\alpha \;\to\; \tilde{\alpha}=\alpha + 2\mu_s \beta + 3\kappa \beta^2 - 4\lambda_s \beta^3 .
\end{align}
Yet, eliminating $\alpha$ is not always possible; for instance, if $\mu_s=\kappa=\lambda_s=0$, or if $\lambda_s=0$ with $3\kappa\alpha>\mu_s^2$, $\alpha$ cannot be absorbed and the potential remains linear in $s$.  We therefore retain the linear term $\alpha s$ in \eqref{E:SMSpot}.

One might object that such a tree-level potential is not ``physical'' because it lacks a global minimum.
However, a tree-level potential \(V\) does not, in general, have to possess a
global minimum: It may be unbounded from below while still admitting a local,
metastable minimum, or it may even lack any stationary point, whereas the
corresponding full effective potential \(V_{\text{eff}}\) may develop a stable
minimum
\cite{ColemanWeinberg1973,GildenerWeinberg1976,CamargoMolina2013Vevacious,Staub2018Reopen,KraussOpferkuchStaub2018UVLandscape,IvanovVazao2020YetAnotherLesson}.
More concretely, the one-loop correction to the renormalized effective
potential may be written \cite{Martin:2001vx}
\begin{align}
 \frac{1}{16\pi^2}V^{(1)}
 =
 \frac{1}{64\pi^2}
 \sum_n
 (-1)^{2s_n}(2s_n+1)(m_n^2)^2
 \left(
 \log\frac{m_n^2}{Q^2}-k_n
 \right),
 \label{eq:one-loop-effective-potential}
\end{align}
where the sum is over field-dependent squared-mass eigenvalues, with
multiplicities not accounted for by the spin factor \((2s_n+1)\) included in
the sum, \(s_n\) is the spin, \(Q\) is the renormalization scale, and \(k_n\)
are renormalization-scheme dependent constants.

With a non-vanishing quartic portal term
\(\lambda_{s\phi}\Phi^\dagger\Phi\,s^2\), the Higgs fluctuation obtains an
\(s_0\)-dependent contribution to its field-dependent squared mass when the
scalars are expanded about a spacetime-independent background,
\begin{align}
 h(x)=h_0+\eta_h(x), \qquad s(x)=s_0+\eta_s(x).
 \label{eq:scalar-background-expansion}
\end{align}
For \(\Phi=(0,h/\sqrt2)^T\), the portal term becomes
\begin{align}
 \lambda_{s\phi}\Phi^\dagger\Phi\,s^2
 =
 \frac{\lambda_{s\phi}}{2}h^2s^2
 =
 \frac{\lambda_{s\phi}}{2}(h_0+\eta_h)^2(s_0+\eta_s)^2 .
 \label{eq:quartic-portal-background}
\end{align}
Its quadratic part in the \(h\)- and \(s\)-fluctuations may be written as
\begin{align}
 V_{\rm portal}^{\rm quad}
 =
 \frac12
 \begin{pmatrix} \eta_h & \eta_s \end{pmatrix}
 \Delta \mc{M}_{\rm portal}^2
 \begin{pmatrix} \eta_h \\ \eta_s \end{pmatrix},
 \qquad
 \Delta \mc{M}_{\rm portal}^2
 =
 \lambda_{s\phi}
 \begin{pmatrix}
 s_0^2 & 2h_0s_0 \\
 2h_0s_0 & h_0^2
 \end{pmatrix}.
 \label{eq:portal-field-dependent-mass-squared-matrix}
\end{align}
For a general background \((h_0,s_0)\), this portal contribution is part of
the full field-dependent \(h\)-\(s\) squared-mass matrix, whose two eigenvalues
are the corresponding scalar entries in
\eqref{eq:one-loop-effective-potential}. Along the singlet background
direction \(h_0=0\), however, the \(h\)- and \(s\)-fluctuations do not mix, and
the Higgs-doublet squared-mass eigenvalue is simply
\begin{align}
 m_h^2(0,s_0)
 =
 -\mu^2+\kappa_{s\phi}s_0+\lambda_{s\phi}s_0^2 .
 \label{eq:field-dependent-higgs-mass}
\end{align}
Analogous portal-induced
field-dependent masses are displayed explicitly, for example, in
\cite{Chao:2014ina}.
The last term in \eqref{eq:field-dependent-higgs-mass} is the contribution
from the quartic portal term. Hence, for large \(|s_0|\),
\(m_h^2(0,s_0)\sim \lambda_{s\phi}s_0^2 \).
  Inserting \eqref{eq:field-dependent-higgs-mass} into
\eqref{eq:one-loop-effective-potential} shows that terms scaling as
\begin{equation}
	s_0^4\log(s_0^2)
\end{equation}
are generated radiatively. Such terms may modify the large-field behaviour of the effective potential
when the pure tree-level singlet potential contains no \(s^4\) term, i.e.
when \(\lambda_s=0\).  

If fermions are included, the top quark contributes in directions with a
non-vanishing Higgs background. Its field-dependent squared mass is
\begin{align}
 m_t^2(h_0)=\frac12 y_t^2 h_0^2 ,
\end{align}
so that, at fixed renormalization scale \(Q\), its leading one-loop
logarithmic contribution behaves as
\begin{align}
 -12\left(m_t^2(h_0)\right)^2
 \log\frac{m_t^2(h_0)}{Q^2}
 \sim
 -3y_t^4 h_0^4\log h_0^2 ,
\end{align}
up to the common prefactor in \eqref{eq:one-loop-effective-potential}.
Along the Higgs direction \(s_0=0\), the portal term gives the singlet
fluctuation a field-dependent squared mass scaling as
\begin{align}
 m_s^2(h_0)\sim \lambda_{s\phi}h_0^2 ,
\end{align}
and hence a scalar one-loop contribution scaling as
\begin{align}
 \lambda_{s\phi}^2 h_0^4\log h_0^2 .
\end{align}
Hence, if only these two contributions are compared, their contribution to the
coefficient of \(h_0^4\log h_0^2\) is proportional to
\(\lambda_{s\phi}^2-3y_t^4\), and can be positive for perturbative values of
\(\lambda_{s\phi}\). This is only a partial one-loop comparison along the Higgs direction, not a
criterion for boundedness from below of the full effective potential.

Moreover, the absence of linear terms when expanding \(V_{\text{eff}}\) about
one of its stationary points (e.g.\ the true vacuum) does not justify
discarding them in a symmetry analysis of the tree-level potential~\(V\),
since, as mentioned above, \(V\) need not possess a stationary point even when
\(V_{\text{eff}}\) does. In any case, \(V\) itself must be expanded about a
stationary point for its linear terms to vanish. Such an expansion, if
possible, also alters the form of the gauge sector through the
covariant-derivative terms; see Section~\ref{sec:LinearTerms} for further
discussion.

Thus, since we a priori do not know which tree-level potentials \(V\) without
stationary points correspond to distinct symmetry algebras, nor which of them
would be excluded or allowed after a separate effective-potential analysis,
we retain the linear tadpole terms in order to make our scalar symmetry
classifications of the SM+S and SM+2S models (without fermions) as general as
possible. The variational symmetries found for the tree-level
potential \(V\), including divergence symmetries, are, in the absence of anomalies, inherited by the quantum
theory, and hence by \(V_{\text{eff}}\). The non-variational symmetries remain relevant for the corresponding classical
field theory, namely the theory defined by the tree-level action, since they
map solutions of the Euler--Lagrange equations to solutions. The completeness of the classification also has additional
symmetry-theoretic value: It provides a reference point for identifying which
types of symmetries may be lost if future analyses of more complicated
singlet models are restricted to tree-level potentials with a stationary
point.

\subsection{Solving the determining equations}
\label{sec:DeterminingEquationsSMS}
We now proceed by applying \texttt{SYM} \cite{dimas2005sym} to find the determining 
equations \eqref{E:prDelta=0PDEsymCond} of the Euler--Lagrange equations
\begin{align}\label{E:ELeqsSMS}
	E(\mc{L}_\text{SMS})=0
\end{align}
of \eqref{E:SMSlag}, cf.~\eqref{E:EulerOpComps}.
There are 21 fields (16 gauge fields and 5 scalar fields) present in this model, and hence the Euler--Lagrange equations \eqref{E:ELeqsSMS} form a system of 21 equations.

Restricting to point symmetries with no spacetime transformations and requiring the gauge fields to be left invariant, we set
\begin{align}
  \xi^\mu&=0,\quad \text{for all}\quad 0\leq \mu \leq 3, \nn \\
  \eta^i&=0,\quad \text{for all}\quad 6\leq i \leq 21.
\end{align}
Moreover, we restrict to symmetries with no explicit spacetime dependence, i.e.\ $\eta^i=\eta^i(\phi_1,\ldots,\phi_5)$.
Then the infinitesimal generator takes the form
\begin{align}
  X=\sum_{i=1}^5 \eta^i(\phi_1,\ldots,\phi_5)\,\partial_{\phi_i},
\end{align}
where $\phi_1,\ldots,\phi_4$ denote the real doublet components and 
\begin{align}\label{E:phi5eqs}
	\phi_5\equiv s,
\end{align}
 so that we are considering (purely) scalar symmetries.

The simplest of the determining equations now read
\begin{align}
	\partial_{\phi_i} \partial_{\phi_j}\eta^k&=0,\quad \text{for all} \quad 1\leq i,j,k \leq 5 \\
	\partial_{y^i} \eta^j &=0,\quad \text{for all}\quad i> 5 \wedge 1\leq j \leq 5, 
\end{align}
 where $y^i$ for $i>5$ denote the gauge field variables, while $y^i\equiv \phi_i$ for $i\le 5$.
This means that the non-zero $\eta$'s are affine in the scalar fields and independent of the gauge field variables $y^i$ for $i>5$. 
Thus,
\begin{align}\label{E:SMSansatz}
	\eta^i = a_i + b_{ij}\phi_j,\qquad 1\le i\le 5,
\end{align}
   with an implicit sum over $j$ from 1 to 5.

Now, define the following Lie algebra generators,
\begin{align}
  X_1 &= \partial_s, \label{E:X1SMS} \\
  X_2 &= s\,\partial_s, \label{E:X2SMS} \\
  X_3 &= \Big(-\frac{\alpha}{\mu^2_s}+2s\Big)\partial_s, \qquad \text{for}\;\: \mu^2_s\ne 0 \label{E:X3SMS}\\
  X_Y &= -\phi_2 \partial_{\phi_1}+\phi_1 \partial_{\phi_2}-\phi_4\partial_{\phi_3}+\phi_3\partial_{\phi_4}, \label{E:XYSMS}
\end{align}
where \(X_Y\) generates the \(\Lg{U}(1)_Y\) hypercharge rephasing of the SM Higgs doublet,
\(\Phi \mapsto e^{i\epsilon}\Phi\) (or equivalently \(\Phi \mapsto e^{i\epsilon Y}\Phi\), with hypercharge \(Y=\tfrac12\) for the Higgs doublet, after rescaling \(\epsilon\)).
Then, substituting \eqref{E:SMSansatz} into the determining equations and requiring that the coefficient of each distinct monomial in each equation vanish, we obtain five solutions and corresponding Lie algebras (one of which is equivalent to another under a constant shift of $s$ and is therefore redundant in the reduced classification).

In the first solution, all $a_i=b_{ij}=0$, except for
\begin{gather}
  a_5,\, b_{55} \in \N{R}, \nn \\
  b_{12} = -b_{43} = -b_{21} = b_{34},
\end{gather}
which means there are three free parameters, and hence the Lie algebra is 3-dimensional, namely
\begin{align}\label{E:SMSsym1}
  \La{a}(1)\oplus \La{u}(1)_Y= \on{span}(X_1,X_2,X_Y),
\end{align}
where $\La{a}(1)=\on{span}(X_1,X_2)$ is the 2-dimensional, non-abelian affine Lie algebra. The parameter conditions
for the symmetry \eqref{E:SMSsym1} are
\begin{align}\label{E:SMSsym1conds}
  \alpha=\mu_s^2=\kappa_{s\phi}=\kappa=\lambda_{s\phi}=\lambda_s=0,
\end{align}
which means the scalar singlet is a massless, free field.

For the remaining solutions, the Lie symmetry algebras realized, together with the corresponding parameter conditions, are
\begin{align}
  \La{sh}\oplus \La{u}(1)_Y
    &= \on{span}(X_1,X_Y), \quad
    &&\alpha\ne 0,\: \mu_s^2=\kappa_{s\phi}=\kappa=\lambda_{s\phi}=\lambda_s=0,
    \label{E:LieAlgX1FreeMassiveSMS} \\
  \La{sc}\oplus \La{u}(1)_Y
    &= \on{span}(X_2,X_Y), \quad
    &&\mu_s^2\ne 0,\: \alpha=\kappa_{s\phi}=\kappa=\lambda_{s\phi}=\lambda_s=0,
    \label{E:LieAlgFreeMassiveSMS} \\
  \La{sc}'\oplus \La{u}(1)_Y
    &= \on{span}(X_3,X_Y), \quad
    &&\alpha, \mu_s^2\ne 0,\: \kappa_{s\phi}=\kappa=\lambda_{s\phi}=\lambda_s=0,
    \label{E:LieAlgX3FreeMassiveSMS} \\
  \La{u}(1)_Y
    &= \on{span}(X_Y), \quad
    &&\text{for all other parameter values.}
    \label{E:LieAlgUYFreeMassiveSMS}
\end{align}
Here $\La{sh}$, $\La{sc}$, and $\La{sc}'$ are one-dimensional Lie algebras (each isomorphic to $\N{R}$), distinguished by their action on the singlet field ($s$-shifts versus $s$-scalings).
Moreover,
\begin{align}
	\La{sc}\eqsim_O \La{sc}'.
\end{align}
(in the sense of \eqref{E:OrthEquiv} and Section~\ref{sec:AffineReparametrizations}), as explained below.

For a potential satisfying \eqref{E:LieAlgX3FreeMassiveSMS}, we may perform a shift \eqref{E:shiftsSMS} with
$\beta=\alpha/(2\mu^2_s)$ and hence eliminate $\alpha$, without introducing new parameters (a new constant term in the potential is irrelevant as it does not survive in the Euler--Lagrange equations, and therefore does not affect the symmetries).\footnote{In the affine form $\eta(s)=Bs+a$ (cf.\ \eqref{E:etaAff}), the generators $X_3$ and $X_2$ correspond to $(B_1,a_1)=(2,-\alpha/\mu_s^2)$ and $(B_2,a_2)=(1,0)$, respectively. Equation~\eqref{E:orthEquivSym} is satisfied with $O=1$, $k=1/2$, and $\gamma=\alpha/(2\mu_s^2)$, showing that $X_3$ and $X_2$ are orthogonally equivalent (i.e.\ related by an orthogonal affine reparametrization and a rescaling of the group parameter).}
 Hence, $\La{sc}\eqsim_O\La{sc}'$, i.e.\ the symmetries \eqref{E:LieAlgFreeMassiveSMS} and \eqref{E:LieAlgX3FreeMassiveSMS} are equivalent through a shift of the field $s$.
 Generators depending on potential parameters, like $X_3$ in \eqref{E:X3SMS}, are typical signs of redundant symmetries due to the potential not being maximally reduced.

The remaining symmetry algebras $\La{sh}\oplus \La{u}(1)_Y$ and $\La{sc}\oplus \La{u}(1)_Y$ are inequivalent since \eqref{E:orthEquivSym} would require $kB_1=OB_2O^T$, i.e.\ $0=1$, which is impossible. 
Therefore, the realizable, inequivalent scalar Lie point symmetry algebras in the SM+S may be written
\begin{align}
  \La{a}(1)\oplus \La{u}(1)_Y,\; \La{sh}\oplus \La{u}(1)_Y,\;  \La{sc}\oplus \La{u}(1)_Y,\;  \La{u}(1)_Y.
\end{align}


\subsection{Nature of the symmetries}
\label{sec:NatureOfTheSymmetries}
We start by calculating the effect of a linear combination of the generators
\eqref{E:X1SMS}, \eqref{E:X2SMS}, and \eqref{E:XYSMS} on the kinetic sector $T$,
\begin{align}\label{E:pr1ofallsymonTSMS}
	\on{pr}(k_1X_1+k_2X_2+k_3X_Y)(T)= k_2 \partial_\mu s \partial^\mu s,
\end{align}
where $T=\mc{L}_\text{SMS}+V$, cf.~\eqref{E:SMSlag}.
By Theorem~\ref{T:prXannTannV} (or Corollary~\ref{C:charVarScalSyms}), we can conclude from
\eqref{E:pr1ofallsymonTSMS} that $X_1$ generates a strict variational symmetry when $\alpha=0$ and
a divergence symmetry when $\alpha\neq 0$, cf.~\eqref{E:LieAlgX1FreeMassiveSMS}.
Here $\alpha\equiv \alpha_5$ in the general linear term $\alpha_i\phi_i$ used e.g.~in Theorem~\ref{T:prXannTannV}, where $i=1,\ldots,5$
and $\phi_5=s$ for the SM+S cf.~\eqref{E:HiggsDoubletSM} and \eqref{E:phi5eqs}; moreover $a_5=\eta^5(0)=1$ for $X_1$, cf.~\eqref{E:X1SMS}.
On the other hand, $X_Y$ always generates a strict variational symmetry, since for $X_Y$ we have $a_i=0$ for all $i$,
cf.~\eqref{E:XYSMS} (and furthermore $\alpha_i=0$ for $i\leq 4$).

It should not be surprising that $X_1$ is strictly variational in the absence of a linear term, since it generates
a constant shift $s\to s+c$ with $c\in\N{R}$, which leaves the Lagrangian invariant for a free, massless scalar singlet $s$,
cf.~\eqref{E:SMSsym1} and \eqref{E:SMSsym1conds}.
Thus, for the symmetry algebra $\La{a}(1)\oplus \La{u}(1)_Y$, cf.~\eqref{E:SMSsym1}, $X_1$ is strictly variational, since
the conditions \eqref{E:SMSsym1conds} require $\alpha=0$, whereas for $\La{sh}\oplus \La{u}(1)_Y$,
$X_1$ is a divergence symmetry since $\alpha\neq 0$, cf.~\eqref{E:LieAlgX1FreeMassiveSMS}.

We continue by checking the nature of $X_2$: Since
\begin{align}
	E_5(\on{pr} X_2 (T))= E_5(\partial_\mu s \partial^\mu s)=-2\partial_\mu \partial^\mu s,
\end{align}
where $E_5$ is the Euler operator corresponding to the field $s\equiv \phi_5$,
we can apply Corollary~\ref{C:charVarScalSyms} (iv) and conclude that $X_2$ is a non-variational symmetry.

As a consistency check, we calculate $E(\on{pr}X_2(\mc{L}_2))$, where $\mc{L}_2$ equals $\mc{L}_\text{SMS}$ with parameters
given by \eqref{E:LieAlgFreeMassiveSMS}, and find a non-vanishing result $4\mu_s^2 s -2 \partial_\mu \partial^\mu s$ in the
$s$-component of the Euler operator. Moreover, we find
\begin{align}
	\on{pr} X_2 \big(E(\mc{L}_2)\big)\big|_{E(\mc{L}_2)=0}=0,
\end{align}
which confirms that $X_2$ is a symmetry of the Euler--Lagrange equations of the SM+S with a free massive singlet $s$.
Here $\on{pr} X_2$ is effectively the second prolongation of $X_2$.

We conclude that there are four possible Lie point symmetry algebras in the SM+S, given by \eqref{E:SMSsym1},
\eqref{E:LieAlgX1FreeMassiveSMS}, \eqref{E:LieAlgFreeMassiveSMS}, and \eqref{E:LieAlgUYFreeMassiveSMS}.
For the first algebra, only the subalgebra $\La{sh}\oplus \La{u}(1)_Y=\on{span}(X_1,X_Y)$ is variational (it is strictly variational),
so for the parameter case \eqref{E:SMSsym1conds} we have $\La{g}_\text{EL}=\La{a}(1)\oplus \La{u}(1)_Y$, while
$\La{g}_\text{var}=\La{g}_\text{svar}=\La{sh}\oplus \La{u}(1)_Y$.
The second algebra, $\La{sh}\oplus \La{u}(1)_Y$, is variational, where one component is a divergence symmetry and the other is a strict variational symmetry;
hence $\La{g}_\text{EL}=\La{g}_\text{var}=\La{sh}\oplus \La{u}(1)_Y$ while $\La{g}_\text{svar}=\La{u}(1)_Y$, for the parameter case
\eqref{E:LieAlgX1FreeMassiveSMS}.
The third algebra, $\La{sc}\oplus \La{u}(1)_Y$, is non-variational, although $\La{u}(1)_Y$ is variational, which means
$\La{g}_\text{EL}=\La{sc}\oplus \La{u}(1)_Y$ while $\La{g}_\text{var}=\La{g}_\text{svar}=\La{u}(1)_Y$ for the parameter case
\eqref{E:LieAlgFreeMassiveSMS}.
The same conclusions hold for the parameter case \eqref{E:LieAlgX3FreeMassiveSMS}, since the two equivalent symmetry algebras are of the same symmetry type,
cf.~Proposition~\ref{P:AffRepar}.
For other parameters, $\La{g}_\text{EL}=\La{g}_\text{var}=\La{g}_\text{svar}=\La{u}(1)_Y$, cf.~\eqref{E:LieAlgUYFreeMassiveSMS}.

Finally, if we consider $X_Y$ to be trivial since it is always present, the only non-trivial variational symmetry of the SM+S is the scalar shift symmetry
$X_1=\partial_s$. It is a strict variational symmetry if $\alpha=0$, and a divergence symmetry otherwise.

\subsection{Algorithm for determining SM+S symmetry algebras}
\label{sec:ASimpleAlgorithmForDecidingSMSSymmetry}
The parameter conditions associated with the realizable Euler--Lagrange Lie point symmetry algebras
$\La{g}_\text{EL}$ in the SM+S, cf.~\eqref{E:SMSsym1}--\eqref{E:LieAlgUYFreeMassiveSMS}, yield a direct inspection-based algorithm for identifying the symmetry content of a given SM+S potential:
\begin{enumerate}
\item Read off the potential parameters and identify which parameter constraints in \eqref{E:SMSsym1}--\eqref{E:LieAlgUYFreeMassiveSMS} are satisfied.
\item  Match the satisfied constraints to the corresponding maximal algebra $\La{g}_\text{EL}$.
\item Read off the variational subalgebras $\La{g}_\text{var}$ and $\La{g}_\text{svar}$ from the case-by-case conclusions given at the end of Section~\ref{sec:NatureOfTheSymmetries}.
\end{enumerate}
This yields an algorithmic identification of the Lie point symmetry algebra of any numerical SM+S potential, without explicit symmetry calculations.

\section{SM+2S}
\label{sec:SM2S}
In this section we determine all scalar Lie point symmetries of the Standard Model extended by two real scalar singlets.
As in the 2HDM, the scalar sector admits a non-trivial reparametrization freedom: Since the two singlets carry identical quantum numbers, one may perform an \(\Lg{O}(2)\) basis rotation in singlet space, in addition to possible constant singlet shifts.
We will exploit this reparametrization freedom to keep the potential in a suitably reduced form, which both ensures solvability of the determining equations and prevents repeated occurrences of the same symmetry in different singlet bases.

\subsection{A general Lagrangian for the SM+KS}
\label{sec:AGeneralLagrangianForSMKS}
We write the SM+2S Lagrangian in a form applicable to the SM+KS for arbitrary $K\in \N{N}$ real scalar gauge singlets, 
similar to the treatment in \cite{Robens_2020}. 
Let
\begin{align}\label{E:SMKSlag}
	\mc{L}_\text{SMKS}=&-\tfrac{1}{4}W^a_{\mu\nu}W^{a\mu\nu}
  -\tfrac{1}{4}B_{\mu\nu}B^{\mu\nu}
  +  (D_\mu\Phi)^\dagger (D^\mu\Phi)
 +\tfrac{1}{2}\partial_\mu s_i \partial^\mu s_i - V(\Phi,s),
\end{align}
with covariant derivative $D_\mu$ and gauge field strength tensors given by \eqref{E:covarDerDef}, \eqref{E:Wfst}, and \eqref{E:Bfst}, 
and with the SM Higgs doublet defined in \eqref{E:HiggsDoubletSM}. 
The most general renormalizable SM+KS potential can be written as
\begin{align}\label{E:SMKSpot}
V(\Phi,s)=&-\mu^2 \Phi^\dagger \Phi + 
\lambda (\Phi^\dagger \Phi)^2 \nn \\
&+ \alpha_i s_i - \sum_{i\leq j} m_{ij} s_i s_j
+ \sum_{i\leq j\leq k}\kappa_{ijk} s_i s_j s_k
+ \sum_{i\leq j\leq k\leq l} \lambda_{ijkl} s_i s_j s_k s_l \nn \\
&+ \kappa_{i} s_i \Phi^\dagger \Phi  
+ \sum_{i\leq j} \lambda_{ij} s_i s_j \Phi^\dagger \Phi, 
\end{align}
where all sums run from $1$ to $K$.
We also note that the singlet kinetic term
\begin{align}
T_s = \tfrac{1}{2}(\partial_\mu s)^T (\partial^\mu s),
\end{align}
is invariant under an orthogonal basis change $O\in \Lg{O}(K)$, i.e.\ a reparametrization $\tilde{s}=Os$ with $s=(s_1,\ldots,s_K)^T$, since
\begin{align}
	T_s\to \tfrac{1}{2}(\partial_\mu \tilde{s})^T (\partial^\mu \tilde{s})
	= \tfrac{1}{2}(\partial_\mu s)^T O^T O (\partial^\mu s) =T_s.
\end{align}
The kinetic sector is thus preserved, and all singlets $s_i$ carry identical quantum numbers; hence the Lagrangian written in the new basis $\tilde{s}=Os$
represents the same physics, although the parameters of the potential $V$ are, in general, not invariant---in direct analogy with the reparametrization freedom
of the scalar sector in the 2HDM \cite{Branco_2012}.

\subsection{Lagrangian and reparametrizations of the SM+2S}
\label{sec:LagrangianAndReparametrizationsOfTheSM2S}
Henceforth we set \(K=2\) in the summations in \eqref{E:SMKSpot}, corresponding to the SM+2S, thereby obtaining the explicit potential
\begin{align}\label{E:SM2Spot}
V(\Phi,s_1,s_2)=\;&
-\mu^2\,\Phi^\dagger\Phi + \lambda\,(\Phi^\dagger\Phi)^2
\nn \\[2pt]
&+ \alpha_1 s_1 + \alpha_2 s_2
\nn \\[2pt]
&- \big( m_{11} s_1^2 + m_{12} s_1 s_2 + m_{22} s_2^2 \big)
\nn \\[2pt]
&+ \kappa_{111} s_1^3 + \kappa_{112} s_1^2 s_2 + \kappa_{122} s_1 s_2^2 + \kappa_{222} s_2^3
\nn \\[2pt]
&+ \lambda_{1111} s_1^4 + \lambda_{1112} s_1^3 s_2 + \lambda_{1122} s_1^2 s_2^2
+ \lambda_{1222} s_1 s_2^3 + \lambda_{2222} s_2^4
\nn \\[2pt]
&+ (\kappa_1 s_1 + \kappa_2 s_2)\,\Phi^\dagger\Phi
\nn \\[2pt]
&+ \big( \lambda_{11} s_1^2 + \lambda_{12} s_1 s_2 + \lambda_{22} s_2^2 \big)\,\Phi^\dagger\Phi \,.
\end{align}

We retain the linear terms $\alpha_i s_i$, since they cannot, in general, be eliminated by shifts of the singlets, see Section \ref{sec:LinearTerms}.
Nevertheless, we will use singlet shifts to simplify the parameterization: Under the reparametrizations
\begin{align}\label{E:fieldShiftsSM2S}
	\tilde{s}_1 = s_1+\gamma_1, \quad \tilde{s}_2 = s_2+\gamma_2,
\end{align}
the reparametrized Lagrangian satisfies 
\begin{align}
\tilde{\mc{L}}[\tilde{s},s^c]=\mc{L}[{s},s^c]=\mc{L}[\tilde{s}-\gamma,s^c], 	
\end{align}
where $\gamma=(\gamma_1,\gamma_2)^T$.
Suppressing the tildes in the new basis, the reparametrized Lagrangian is therefore obtained by performing the substitutions
\begin{align}\label{E:shiftsss}
	{s}_1 \to s_1-\gamma_1, \quad {s}_2 \to s_2-\gamma_2,
\end{align}
in the original Lagrangian $\mc{L}$.
The parameters of the potential then transform as follows:
\begin{align}
	\alpha_1 &\to \alpha _1+2 \gamma _1 m_{11}+\gamma _2 m_{12}+3 \gamma _1^2 \kappa
   _{111}+2 \gamma _2 \gamma _1 \kappa _{112}+\gamma _2^2 \kappa
   _{122}\nn \\&\quad \:-4 \gamma _1^3 \lambda _{1111}-3 \gamma _2 \gamma _1^2
   \lambda _{1112}-2 \gamma _2^2 \gamma _1 \lambda _{1122}-\gamma
   _2^3 \lambda _{1222},\label{E:shifta1} \\[2pt]
	\alpha_2 &\to \alpha _2+\gamma _1 m_{12}+2 \gamma _2 m_{22}+\gamma _1^2 \kappa
   _{112}+2 \gamma _2 \gamma _1 \kappa _{122}+3 \gamma _2^2 \kappa
   _{222}\nn \\&\quad \:-\gamma _1^3 \lambda _{1112}-2 \gamma _2
   \gamma _1^2 \lambda _{1122}-3 \gamma _2^2 \gamma _1 \lambda
   _{1222}-4 \gamma _2^3 \lambda _{2222}, \label{E:shifta2}  \\[2pt]
	\mu ^2	&\to \mu ^2+\gamma _1 \kappa _1+\gamma _2 \kappa _2-\gamma _1^2 \lambda
   _{11}-\gamma _1 \gamma _2 \lambda _{12}-\gamma _2^2 \lambda
   _{22}, \\[2pt]
	m_{11} &\to m_{11}+3 \gamma _1 \kappa _{111}+\gamma _2 \kappa _{112}-6 \gamma
   _1^2 \lambda _{1111}-3 \gamma _2 \gamma _1 \lambda
   _{1112}-\gamma _2^2 \lambda _{1122}, \label{E:shiftm11} \\[2pt]
	m_{12} &\to m_{12}+2 \gamma _1 \kappa _{112}+2 \gamma _2 \kappa _{122}-3 \gamma
   _1^2 \lambda _{1112}-4 \gamma _2 \gamma _1 \lambda _{1122}-3
   \gamma _2^2 \lambda _{1222},\label{E:shiftm12} \\[2pt]
m_{22} &\to	m_{22}+\gamma _1 \kappa _{122}+3 \gamma _2 \kappa _{222}-\gamma
   _1^2 \lambda _{1122}-3 \gamma _2 \gamma _1 \lambda
   _{1222}-6 \gamma _2^2 \lambda _{2222},\label{E:shiftm22} \\[2pt]
	\kappa
   _{1} &\to \kappa _1-2 \gamma _1 \lambda _{11}-\gamma _2 \lambda _{12} \label{E:shiftk1}  \\[2pt]
	\kappa
   _{2} &\to \kappa _2 -\gamma _1 \lambda _{12}-2 \gamma _2 \lambda _{22} \label{E:shiftk2} \\[2pt]
	\kappa
   _{111} &\to \kappa
   _{111}-4 \gamma _1 \lambda _{1111}-\gamma _2 \lambda _{1112}, \label{E:shiftk111} \\[2pt]
	\kappa
   _{112} &\to \kappa
   _{112}-3 \gamma _1 \lambda _{1112}-2 \gamma _2 \lambda _{1122},\label{E:shiftk112}  \\[2pt]
	\kappa
   _{122} &\to \kappa
   _{122}-2 \gamma _1 \lambda _{1122}-3 \gamma _2 \lambda _{1222}, \label{E:shiftk122}  \\[2pt]
	\kappa
   _{222}&\to \kappa
   _{222}-\gamma _1 \lambda _{1222}-4 \gamma _2 \lambda _{2222}. \label{E:shiftk222} 
\end{align}
The quartic couplings are not listed here since they are invariant under singlet shifts.
We also note that any two cubic couplings can generically be eliminated by a suitable choice of $(\gamma_1,\gamma_2)$, unless the relevant quartic couplings vanish, or unless the two shift-induced variations $\Delta\kappa(\gamma_1,\gamma_2)$ are linearly dependent.

Next, consider an $\Lg{O}(2)$ basis transformation of the singlets,
\begin{align}
\tilde{s}\equiv
\begin{pmatrix}
\tilde{s}_1 \\[2pt]
\tilde{s}_2
\end{pmatrix} =\begin{pmatrix}
\cos \theta & -\delta \sin \theta \\[2pt]
\sin \theta & \delta \cos \theta
\end{pmatrix}
\begin{pmatrix}
{s}_1 \\[2pt]
{s}_2
\end{pmatrix}
\equiv Os,
\end{align}
where $\delta=\pm 1$, $\theta\in[0,2\pi)$, and $\delta=-1$ yields an orthogonal transformation with determinant $-1$.
The reparametrized Lagrangian is then
\begin{align}
\tilde{\mc{L}}[\tilde{s},s^c]=\mc{L}[{s},s^c]=\mc{L}[O^T\tilde{s},s^c], 	
\end{align}
and, suppressing the tildes in the new basis, it is obtained by substituting $s\to O^T s$ in the original Lagrangian $\mc{L}$, i.e.
\begin{align}
	{s}_1 \to \cos(\theta) s_1 + \sin(\theta) s_2, \quad {s}_2 \to -\delta \sin(\theta) s_1 + \delta\cos(\theta) s_2.
\end{align}
In this case the parameters of the potential transform as:
\begin{align}\label{E:trafoParametersSO2}
	\alpha_1 &\to \alpha _1 c_{\theta }-\alpha _2 \delta  s_{\theta },  \\[2pt]
	\alpha_2 &\to \alpha _2 \delta  c_{\theta }+\alpha _1 s_{\theta }, \\[2pt]
	 m_{11} &\to c_{\theta }^2 m_{11}-\delta  c_{\theta } m_{12} s_{\theta }+m_{22}
   s_{\theta }^2, \\[2pt]
 m_{12} &\to	\delta  c_{2 \theta } m_{12}+m_{11} s_{2 \theta }-m_{22} s_{2
   \theta }, \\[2pt]
m_{22} &\to	c_{\theta }^2 m_{22}+\delta  c_{\theta } m_{12} s_{\theta }+m_{11}
   s_{\theta }^2, \\[2pt]
	\kappa_1 &\to c_{\theta } \kappa _1- \delta  \kappa _2
   s_{\theta }, \\[2pt]
	\kappa_2 &\to  \delta  c_{\theta } \kappa _2+s_{\theta
   }\kappa _1, \\[2pt]
	\kappa _{111} &\to c_{\theta }^3 \kappa _{111}-\delta  c_{\theta }^2 \kappa _{112}
   s_{\theta }+c_{\theta } \kappa _{122} s_{\theta }^2-\delta 
   \kappa _{222} s_{\theta }^3, \\[2pt]
	\kappa _{112} &\to\tfrac{1}{4} \delta  \left(c_{\theta }+3 c_{3 \theta }\right) \kappa
   _{112}+3 \delta  c_{\theta } \kappa _{222} s_{\theta }^2+3
   c_{\theta }^2 \kappa _{111} s_{\theta }+\tfrac{1}{4} \kappa
   _{122} \left(s_{\theta }-3 s_{3 \theta }\right), \\[2pt]
	\kappa _{122} &\to  \tfrac{1}{4} \left(c_{\theta }+3 c_{3 \theta }\right) \kappa
   _{122}-3 \delta  c_{\theta }^2 \kappa _{222} s_{\theta }+3
   c_{\theta } \kappa _{111} s_{\theta }^2-\tfrac{1}{4} \delta 
   \kappa _{112} \left(s_{\theta }-3 s_{3 \theta }\right),  \\[2pt]
	\kappa _{222} &\to \delta  c_{\theta }^3 \kappa _{222}+\delta  c_{\theta } \kappa
   _{112} s_{\theta }^2+c_{\theta }^2 \kappa _{122} s_{\theta
   }+\kappa _{111} s_{\theta }^3, \\[2pt]
	\lambda _{11} &\to  c_{\theta }^2 \lambda _{11}- \delta 
   c_{\theta } \lambda _{12} s_{\theta }+ \lambda _{22}
   s_{\theta }^2, \\[2pt]
	\lambda _{12} &\to  \delta  c_{2 \theta } \lambda _{12}+ \lambda
   _{11} s_{2 \theta }- \lambda _{22} s_{2 \theta }, \label{E:O2trafol12} \\[2pt]
	\lambda _{22} &\to  c_{\theta }^2 \lambda _{22}+ \delta 
   c_{\theta } \lambda _{12} s_{\theta }+ \lambda _{11}
   s_{\theta }^2, \\[2pt]
\lambda _{1111} &\to	c_{\theta }^4 \lambda _{1111}-\delta  c_{\theta }^3 \lambda _{1112}
   s_{\theta }-\delta  c_{\theta } \lambda _{1222} s_{\theta
   }^3+c_{\theta }^2 \lambda _{1122} s_{\theta }^2+\lambda _{2222}
   s_{\theta }^4, \\[2pt]
\lambda
   _{1112} &\to	\tfrac{1}{2} \delta  \left(c_{2 \theta }+c_{4 \theta }\right)
   \lambda _{1112}+4 c_{\theta }^3 \lambda _{1111} s_{\theta }-4
   c_{\theta } \lambda _{2222} s_{\theta }^3+\delta  \lambda
   _{1222} s_{\theta } s_{3 \theta }-\tfrac{1}{2} \lambda _{1122}
   s_{4 \theta }, \label{E:O2trafol1112} \\[2pt]
	\lambda
   _{1122} &\to \tfrac{1}{4} \left(3 c_{4 \theta }+1\right) \lambda
   _{1122}+\tfrac{3}{4} \delta  \lambda _{1112} s_{4 \theta
   }-\tfrac{3}{4} \delta  \lambda _{1222} s_{4 \theta }+\tfrac{3}{2}
   \lambda _{1111} s_{2 \theta }^2+\tfrac{3}{2} \lambda _{2222} s_{2
   \theta }^2, \\[2pt]
	\lambda
   _{1222} &\to  \tfrac{1}{2} \delta  \left(c_{2 \theta }+c_{4
   \theta }\right) \lambda _{1222}+ \tfrac{1}{2} \delta  \left(c_{2 \theta }-c_{4 \theta }\right)
   \lambda _{1112}\nn \\[2pt] &\quad \:-4 c_{\theta }^3 \lambda _{2222}
   s_{\theta }+4 c_{\theta } \lambda _{1111} s_{\theta
   }^3+\tfrac{1}{2} \lambda _{1122} s_{4 \theta }, \\[2pt]
	\lambda
   _{2222} &\to c_{\theta }^4 \lambda _{2222}+\delta  c_{\theta }^3 \lambda _{1222}
   s_{\theta }+\delta  c_{\theta } \lambda _{1112} s_{\theta
   }^3+c_{\theta }^2 \lambda _{1122} s_{\theta }^2+\lambda _{1111}
   s_{\theta }^4, \label{E:trafoParametersSO2lambda2222}
\end{align}
whereas $\mu^2$ and $\lambda$ are invariant. Since \(\tilde{\alpha}=O\alpha\) and \(\tilde{s}=Os\), a rotational reparametrization leaves the linear terms unchanged, \(\tilde{\alpha}^{T}\tilde{s}=\alpha^{T}s\), as it should. Moreover, since \(\tilde{a}=\tilde{\eta}(0)=O\eta(0)=Oa\) for a scalar symmetry with characteristic
\(\eta\), it follows that
\begin{align}
  \tilde{a}^{T}\tilde{\alpha} = a^{T}\alpha,
\end{align}
and hence a divergence symmetry (with \(\pr X(T)=0\)) remains a divergence symmetry under an
orthogonal reparametrization, in accordance with Corollary~\ref{C:charVarScalSyms}(ii) and
Proposition~\ref{P:AffRepar}.

The quartic coupling $\lambda_{1111}$ cannot always be transformed to zero; a counter-example is $\lambda_{1111}=\lambda_{2222}=1$ with all other
$\lambda_{ijkl}=0$.
However, $\lambda_{1112}$ can always be set to zero: The transformed coupling $\tilde{\lambda}_{1112}$, given by the right-hand side of
\eqref{E:O2trafol1112}, satisfies
\begin{align}
	\int_0^{2\pi} \tilde{\lambda}_{1112}(\theta)\, d\theta=0,
\end{align}
and since $\tilde{\lambda}_{1112}(\theta)$ is continuous, there exists $\theta_0\in\langle 0,2\pi\rangle$ such that
\begin{align}\label{E:l1112to0}
\tilde{\lambda}_{1112}(\theta_0)=0.
\end{align}
The same argument applies to $\lambda_{1222}$.
Turning to the cubic singlet couplings, one likewise has
\begin{align}
	\forall i\, \forall j\, \forall k \ \exists\, \theta \in [0,2\pi\rangle:\;
\tilde{\kappa}_{ijk}(\theta) = 0,
\end{align}
but, in general, not for the same $\theta$ for different $(i,j,k)$.
We begin the reduction of the SM+2S potential by setting $\lambda_{1112}=0$. For an alternative, $\Lg{SO}(2)$-adapted formulation of the potential that allows more transparent reductions of the quartic and cubic parameters, see Appendix~\ref{sec:SymmetryAnalysisWithLgSO2AdaptedSM2SPotential}.

\subsubsection{Linear terms}
\label{sec:LinearTerms}
We perform the symmetry analysis on the tree-level potential \eqref{E:SM2Spot}. As for the SM+S in Section~\ref{sec:LinearTermsSMS}, the field shifts \eqref{E:fieldShiftsSM2S} cannot, in general, be employed to remove all linear terms for arbitrary parameter values; cf.~\eqref{E:shifta1} and \eqref{E:shifta2}. Indeed, eliminating the terms linear in the shifted singlets requires solving the coupled system $\tilde{\alpha}_1(\gamma_1,\gamma_2)=0$ and $\tilde{\alpha}_2(\gamma_1,\gamma_2)=0$, which may fail to admit any real solution in $(\gamma_1,\gamma_2)$ even for tree-level potentials with extrema; see the argument below.
 Symmetry analyses are carried out at tree level because renormalization, in the absence of quantum anomalies, preserves variational symmetries, so that the effective potential $V_{\text{eff}}$ has the same (variational) symmetries as $V$. A physical effective potential $V_{\text{eff}}$ should possess a minimum, and if one expands all fields about such a minimum, no linear terms are present. However, this argument cannot be used to remove the linear terms from the tree-level potential $V$, since it may happen that $V$ has no minimum while $V_{\text{eff}}$ does; cf.~the discussion of the linear term in the SM+S in Section~\ref{sec:LinearTermsSMS}. To remain completely general, we therefore keep the linear terms in our symmetry analysis of $V$, just as we did for the SM+S.

Even if one were to demand that the tree-level potential $V$ have a stationary point, expanding about a stationary point with nonzero components along the Higgs-doublet directions $\phi_i$, $i=1,\dots,4$, modifies the explicit form of the SM ``kinetic'' sector (including Higgs--gauge interactions), e.g.\ by generating gauge boson mass terms. Although the symmetries are not destroyed (they are only ``hidden''), this field shift (an affine reparametrization) replaces the original linear (tadpole) terms by new quadratic and higher-order terms in the shifted Lagrangian. While the resulting parameters are not independent, they would nevertheless increase the number of cases to be considered in the symmetry analysis.

Furthermore, expanding only in the singlet directions does not, in general, remove all terms that are linear in the singlets. Let
\begin{align}
\phi=(\phi_1,\ldots,\phi_4,s_1,s_2),
\end{align}
and suppose $V$ has a stationary point at $\phi=v$. Then
\begin{align}\label{E:singletStatConds}
0=\left.\frac{\partial V}{\partial \phi_i}\right|_{\phi=v}\qquad \forall i\in \{1,\ldots,6\}.
\end{align}
Expanding about the full stationary point, $\phi=\tilde{\phi}+v$, and defining
\begin{align}\label{E:Vtilde1}
\tilde{V}(\tilde{\phi})=V(\tilde{\phi}+v),
\end{align}
yields
\begin{align}\label{E:singletStatCondsTransf}
0=\left.\frac{\partial \tilde{V}}{\partial \tilde{\phi}_i}\right|_{\tilde{\phi}=0} \qquad \forall i\in \{1,\ldots,6\},
\end{align}
so in particular the terms linear in the shifted fields vanish, e.g.\ $\tilde{\alpha}_1=0$ and $\tilde{\alpha}_2=0$.

On the other hand, consider a singlet-only shift $w=(0,0,0,0,v_5,v_6)$, i.e.\ $\phi=\tilde{\phi}+w$, and (re-)define
\begin{align}\label{E:Vtilde2}
\tilde{V}(\tilde{\phi}) = V(\tilde{\phi}+w),
\end{align}
which is, in general, a different function than the $\tilde{V}$ defined in \eqref{E:Vtilde1}.
Then the stationary conditions \eqref{E:singletStatCondsTransf} are evaluated at
\begin{align}
	\tilde{\phi}=(v_1,v_2,v_3,v_4,0,0),
\end{align}
rather than at $\tilde{\phi}=0$. In this case the singlet stationary conditions read, using \eqref{E:SM2Spot},
\begin{align}\label{E:singletOnlyShiftTadpole}
0=\left.\frac{\partial \tilde{V}}{\partial \tilde{s}_1}\right|_{\substack{\tilde{s}_1=\tilde{s}_2=0\\ \tilde{\phi}_a=v_a}}
=\tilde{\alpha}_1 + \tilde{\kappa}_1\,(\Phi^\dagger\Phi)\big|_{\tilde{\phi}_a=v_a}
=\tilde{\alpha}_1 + \frac{1}{2}\tilde{\kappa}_1\,(v_1^2+v_2^2+v_3^2+v_4^2),
\end{align}
and analogously for $\tilde{s}_2$. Thus, unless $\tilde{\kappa}_i=0$ for the chosen shift, a singlet-only shift does not force $\tilde{\alpha}_1=\tilde{\alpha}_2=0$; rather, it fixes $\tilde{\alpha}_i$ in terms of the Higgs-doublet background through \eqref{E:singletOnlyShiftTadpole}.
 In addition, the portal terms $\tilde{\kappa}_i\,\tilde{s}_i\,\Phi^\dagger\Phi$ generically survive a singlet-only shift, and the purely linear singlet terms $\tilde{\alpha}_i\,\tilde{s}_i$ vanish in general only when expanding about a stationary point of the full scalar sector.

\subsection{Determining equations}
\label{sec:DeterminingEquationsSM2S}
As before, we apply \texttt{SYM} \cite{dimas2005sym} to calculate the determining equations \eqref{E:prDelta=0PDEsymCond} of the 22 Euler-Lagrange equations 
\begin{align}
E(\mc{L}_\text{SM2S})=0,	
\end{align}
one for each field.
We will only study scalar symmetries, and hence set
\begin{align}
	\xi=0, \quad \eta^i=0,\quad \forall i \geq 7,
\end{align}
 i.e.~we are considering the infinitesimal generator
\begin{align}\label{E:infGenSM2Sskalar}
	X= \sum_{i=1}^6 \eta^i \partial_{\phi_i},
\end{align}
 with the identifications
\begin{align}
	\phi_5\equiv s_1, \quad \phi_6 \equiv s_2,
\end{align}
 and where the first four $\phi_i$'s are the component fields of the SM Higgs doublet \eqref{E:HiggsDoubletSM}.
Again, the simplest determining equations include
\begin{align}\label{E:simplestDetEqsSM2S}
  \partial_{\phi_i}\partial_{\phi_j}\eta^k=0, \qquad 1\le i,j,k\le 6,
\end{align}
which implies that $\eta^k$ is affine in the scalar fields,
\begin{align}\label{E:etaAffineSM2S}
  \eta^k=a_k+b_{k\ell}\phi_\ell, \qquad 1\le k,\ell\le 6.
\end{align}
Similar to the SM+S, the determining equations include
\begin{align}
\partial_{y_j}\eta^i=0,\qquad i=1,\ldots,6,\ \ j\ge 7,
\end{align}
so $\eta^1,\ldots,\eta^6$ are independent of the gauge field variables $y_j$. Here $y_j$ denotes the full set of field variables; for $j\le 6$ we have $y_j\equiv \phi_j$, while $y_j$ with $j\ge 7$ are the gauge field variables.
In any case, we restrict attention to scalar symmetries only and therefore do not consider gauge-field-dependent generators. 

Then, substituting \eqref{E:etaAffineSM2S} into the determining equations and solving only the subset that is independent of the potential parameters (i.e.\ contains no parameters from $V$) yields
\begin{align}
a_i&=0, && i=1,\ldots,4, \\
b_{jk}&=0, && (j,k)\notin\{(1,2),(2,1),(3,4),(4,3),(5,5),(5,6),(6,5),(6,6)\},
\end{align}
where, at this stage,
\begin{gather}
a_5,a_6,b_{55},b_{56},b_{65},b_{66}\in\mathbb{R},\\
b_{12}=-b_{21}=-b_{43}=b_{34}\in\mathbb{R}.
\end{gather}
Substituting these results into the remaining determining equations leaves the undetermined constants
\begin{align}\label{E:B3}
B_3=\{ a_5, a_6, b_{55}, b_{56}, b_{65}, b_{66} \}
\end{align}
appearing in the equations. The values and mutual relations of these constants may be further restricted when the classifying equations (i.e.\ the determining equations containing parameters from $V$) are imposed; in particular, they may set some of them to zero or yield additional equalities.
One additional constant, which we take to be $b_{34}$, remains undetermined but does not occur in the remaining equations; it corresponds to the ever-present symmetry $X_Y$.

\subsection{Parameter cases and reductions of the SM+2S potential}
\label{sec:ParameterCasesAndReductionsOfTheSM2SPotential}
To avoid unnecessary repetitions of equivalent Lie symmetry algebras, including ``exotic'' cases (i.e.\ instances where the symmetry generators depend explicitly on the potential parameters) and to make the determining equations solvable, we will reduce the potential and divide our analysis into four branches.
The starting point for all branches is the elimination of $\lambda_{1112}$,
\begin{align}\label{E:l1112to0_2}
  \lambda_{1112}\to 0,
\end{align}
by an $\Lg{SO}(2)$ reparametrization of the potential, cf.~\eqref{E:l1112to0}.
The four branches are then determined by whether $\lambda_{1111}$ and $\lambda_{1122}$ vanish or not (four possible combinations) in the reparametrized potential where $\lambda_{1112}= 0$.

\subsubsection{Branch I: $\lambda_{1111}\ne 0,\; \lambda_{1122}\ne 0$}
\label{sec:BranchOne}

We now assume
\begin{align}\label{E:l1111l1122ne0}
	\lambda_{1111}\ne 0,\quad \lambda_{1122}\ne 0.
\end{align}
Then we can eliminate the parameters $\kappa_{111}$ and $\kappa_{112}$ by two shifts \eqref{E:shiftsss}, cf.~\eqref{E:shiftk111} and \eqref{E:shiftk112}, without altering any of the parameters $\lambda_{ijkl}$, since they are not affected by such shifts. Thus
\begin{align}\label{E:k111k112to0}
	\kappa_{111}\to 0,\quad \kappa_{112}\to 0.
\end{align}

\paragraph{Leaf~1}
\label{sec:Leaf1}

Hence, the reparametrization freedom is exhausted, and we 
have arrived at Leaf~1 in Fig.~\ref{fig:reduction-tree0}, highlighted in red.
By solving the determining equations with the assumptions
 \eqref{E:l1112to0_2}, \eqref{E:l1111l1122ne0} and \eqref{E:k111k112to0} by \texttt{Mathematica}'s built-in \texttt{Reduce} function, we obtain two solutions for the parameters $B_3$ \eqref{E:B3}. The first solution is
\begin{align}\label{E:asbsN1}
	a_5=a_6=b_{55}=b_{66}=0,\quad b_{56}=-b_{65},
\end{align}
with the parameter conditions \eqref{E:consB1N1a}. In addition, one parameter (taken to be $b_{34}$) is free, as mentioned below \eqref{E:B3}. Substituting \eqref{E:asbsN1}
into \eqref{E:etaAffineSM2S} and \eqref{E:infGenSM2Sskalar}
yields the following symmetry algebra:
\begin{gather}
	\La{so}(2)\oplus \La{u}(1)_Y=\on{span}(s_2\partial_{s_1}-s_1\partial_{s_2},X_Y), \\
	\text{for}\quad
\alpha_{1}=\alpha_{2}=m_{12}=\kappa_{1}=\kappa_{2}=\kappa_{122}=\kappa_{222}=\lambda_{12}=\lambda_{1222}=0, \nn \\
m_{11}=m_{22},\ \lambda_{11}=\lambda_{22},\ \lambda_{1111}=\lambda_{2222},\ \lambda_{1122}=2\lambda_{2222}. \label{E:consB1N1a}
\end{gather}

Here $X_Y$ was given in \eqref{E:XYSMS}, and $\La{so}(2)\cong \N{R}$ generates a rotation.
The conditions \eqref{E:consB1N1a} are in addition to
\eqref{E:l1112to0_2}, \eqref{E:l1111l1122ne0} and \eqref{E:k111k112to0}, and correspond to the potential
\begin{align}\label{E:b1so2pot}
	V=m_{11}s^Ts+\lambda_{11}s^Ts\Phi^\dag \Phi+\lambda_{1111}(s^Ts)^2 -\mu^2 \Phi^\dagger \Phi + 
\lambda (\Phi^\dagger \Phi)^2,
\end{align}
with $s^T=(s_1,s_2)$, which indeed is invariant under $\Lg{SO}(2)$ rotations.
Since the kinetic terms are also invariant under $\Lg{SO}(2)$, the full Lagrangian is invariant, and this is a strict variational symmetry.

The other solution yields only the ubiquitous algebra $\La{u}(1)_Y$, which is the maximal algebra in Leaf~1 for all parameter values not satisfying \eqref{E:consB1N1a}.
For instance, after eliminating $\lambda_{1112}$ by an $\Lg{SO}(2)$ rotation in singlet space, suppose that $\lambda_{1111}\neq 0$ and $\lambda_{1122}\neq 0$.
Then $\kappa_{111}$ and $\kappa_{112}$ can be eliminated by two shifts.
If we then find, e.g., that $m_{11}\neq m_{22}$ (contradicting \eqref{E:consB1N1a}), the maximal symmetry algebra is $\La{u}(1)_Y$.

\begin{figure}[htbp]
\centering
\scalebox{1.0}{
\begin{forest}
for tree={
  draw, rounded corners, thick,
  align=center, inner sep=4pt,
  l sep=10mm, s sep=8mm,
  edge={->, thick}
}
[{$\lambda_{1112}\xrightarrow{\ \mathrm{SO}(2)\ } 0$}, very thick,
  [{$\lambda_{1111}\neq 0$\\$\lambda_{1122}\neq 0$}, edge label={node[midway, sloped, inner sep=1pt, fill=white]{I}}
    [{$\kappa_{111}\to 0$\\$\kappa_{112}\to 0$}, draw=red, very thick, 
  label={[circle, draw, fill=yellow!30, inner sep=1pt, font=\scriptsize]
         north east:{1}}]
  ]
  [{$\lambda_{1111}=0$\\$\lambda_{1122}\neq 0$}, edge label={node[midway, sloped, inner sep=1pt, fill=white]{II}}
    [{$\kappa_{112}\to 0$\\$\kappa_{122}\to 0$}, draw=red, very thick, 
  label={[circle, draw, fill=yellow!30, inner sep=1pt, font=\scriptsize]
         north east:{2}}]
  ]]
\end{forest}
}
\caption{\textbf{Reduction tree, Branches I and II.}
Each path starting from the root at the top and terminating at a leaf (outlined in red) at the bottom corresponds to a reduced potential.
For each leaf, the determining equations for the corresponding reduced potential are solved, and the associated symmetries are derived.
Leaf~1 is discussed in Section~\ref{sec:BranchOne}, while Leaf~2 is treated in Section~\ref{sec:BranchTwo}.}
\label{fig:reduction-tree0}
\end{figure}

\subsubsection{Branch II: $\lambda_{1111}= 0,\; \lambda_{1122}\ne 0$}
\label{sec:BranchTwo}

Now consider the case
\begin{align}\label{E:l1111eq0l1122ne0}
	\lambda_{1111}=0,\quad \lambda_{1122}\ne 0.
\end{align}
We can eliminate $\kappa_{112}$ by an appropriate choice of $\gamma_2$, and subsequently set $\kappa_{122}$ to zero by an appropriate choice of $\gamma_1$, cf.~\eqref{E:shiftk112} and \eqref{E:shiftk122}, respectively. 

\paragraph{Leaf~2}
\label{sec:Leaf2}

This means
\begin{align}\label{E:k112k122to0}
	\kappa_{112}\to 0, \quad \kappa_{122}\to 0,
\end{align}
and the situation is given by Leaf~2 in Fig.~\ref{fig:reduction-tree0}.
Solving the determining equations under the assumptions of Leaf~2, i.e.\ \eqref{E:l1112to0_2}, \eqref{E:l1111eq0l1122ne0} and \eqref{E:k112k122to0}, we find that all parameters of $B_3$ vanish, and the symmetry algebra is the ever-present
\begin{align}
	\La{u}(1)_Y=\on{span}(X_Y).
\end{align}

\subsubsection{Branch III: $\lambda_{1111} \ne 0,\; \lambda_{1122}= 0$}
\label{sec:BranchThree}
\label{sec:Node1}
Now suppose
\begin{align}\label{E:l1111ne0l1122eq0}
	\lambda_{1111} \ne 0,\; \lambda_{1122}= 0,
\end{align}
and we are hence on Branch III in Fig.~\ref{fig:reduction-tree1}.


\begin{sidewaysfigure}
\centering
\resizebox{1.0\textheight}{!}{
\begin{forest}
for tree={
  draw, rounded corners, thick,
  align=center, inner sep=4pt,
  l sep=10mm, s sep=8mm,
  edge={->, thick}
}
[{$\lambda_{1112}\xrightarrow{\ \mathrm{SO}(2)\ } 0$}, very thick,
  [{$\lambda_{1111}\neq 0$\\$\lambda_{1122}=0$}, edge label={node[midway, sloped, inner sep=1pt, fill=white]{III}}
    [{$\kappa_{111}\to 0$\\($\gamma_1$ fixed)}
     [{$\lambda_{1222}=0$?}
        [{$\lambda_{2222}=0$?}, edge label={node[midway, sloped, inner sep=1pt, fill=white, font=\small]{yes}}
        [{$\kappa_{112}=0$?}, edge label={node[midway, sloped, inner sep=1pt, fill=white, font=\small]{yes}}, draw=red, very thick, 
  label={[circle, draw, fill=yellow!30, inner sep=1pt, font=\scriptsize]
         north east:{3}}
        [{$\kappa_{122}=0$?}, edge label={node[midway, sloped, inner sep=1pt, fill=white, font=\small]{yes}}
          [{$\kappa_{222}=0$?}, edge label={node[midway, sloped, inner sep=1pt, fill=white, font=\small]{yes}} [{$m_{12}= 0$?}, edge label={node[midway, sloped, inner sep=1pt, fill=white, font=\small]{yes}} [{$m_{22}= 0$?}, edge label={node[midway, sloped, inner sep=1pt, fill=white, font=\small]{yes}} [{$\varnothing$}, edge label={node[midway, sloped, inner sep=1pt, fill=white, font=\small]{yes}}] [{$\alpha_2\to 0$}, edge label={node[midway, sloped, inner sep=1pt, fill=white, font=\small]{no}}, draw=green, very thick]] [{$\alpha_1\to 0$}, edge label={node[midway, sloped, inner sep=1pt, fill=white, font=\small]{no}}]] [{$m_{22}\to 0$}, edge label={node[midway, sloped, inner sep=1pt, fill=white, font=\small]{no}}]]
          [{$m_{12}\to 0$}, edge label={node[midway, sloped, inner sep=1pt, fill=white, font=\small]{no}}]
        ]
        [{$m_{11}\to 0$}, edge label={node[midway, sloped, inner sep=1pt, fill=white, font=\small]{no}}]
      ] 
          [{$\kappa_{222}\to 0$}, edge label={node[midway, sloped, inner sep=1pt, fill=white, font=\small]{no}}, draw=red, very thick, 
  label={[circle, draw, fill=yellow!30, inner sep=1pt, font=\scriptsize]
         north east:{4}}]]
        [{$\kappa_{122}\to 0$}, edge label={node[midway, sloped, inner sep=1pt, fill=white, font=\small]{no}}, draw=red, very thick, 
  label={[circle, draw, fill=yellow!30, inner sep=1pt, font=\scriptsize]
         north east:{5}}]
      ]
    ]
  ]
[{$\lambda_{1111}=0$\\$\lambda_{1122}=0$}, edge label={node[midway, sloped, inner sep=1pt, fill=white]{IV}}
  [{$\lambda_{1222}=0$?}
    [{$\lambda_{2222}=0$?}, edge label={node[midway, sloped, inner sep=1pt, fill=white, font=\small]{yes}}
      [{$\lambda_{12}\to 0$}, edge label={node[midway, sloped, inner sep=1pt, fill=white, font=\small]{yes}}
        [{$\lambda_{11}=0$?}
          [{$\lambda_{22}=0$?}, edge label={node[midway, sloped, inner sep=1pt, fill=white, font=\small]{yes}}[{$\kappa_{111}\to 0$}, edge label={node[midway, sloped, inner sep=1pt, fill=white, font=\small]{yes}}[{$\kappa_{112}= 0$?} [{$\kappa_{122}=0$?}, edge label={node[midway, sloped, inner sep=1pt, fill=white, font=\small]{yes}} [{$\kappa_{222}=0$?\\Continued in Fig.~\ref{fig:reduction-tree2}}, edge label={node[midway, sloped, inner sep=1pt, fill=white, font=\small]{yes}},draw=blue, very thick][{$m_{12}\to 0\:$ ($\gamma_2$ fixed) \\$m_{22}\to 0\:$ ($\gamma_1$ fixed)}, draw=red, very thick, 
  label={[circle, draw, fill=yellow!30, inner sep=1pt, font=\scriptsize]
         north east:{15}}, edge label={node[midway, sloped, inner sep=1pt, fill=white, font=\small]{no}} ]][{$m_{11}\to 0 \:$ ($\gamma_2$ fixed)\\$m_{12}\to 0\:$ ($\gamma_1$ fixed)}, draw=red, very thick, 
  label={[circle, draw, fill=yellow!30, inner sep=1pt, font=\scriptsize]
         north east:{16}}, edge label={node[midway, sloped, inner sep=1pt, fill=white, font=\small]{no}}]] ]
          [{$\kappa_2\to 0$ \\ ($\gamma_2$ fixed)}, edge label={node[midway, sloped, inner sep=1pt, fill=white, font=\small]{no}} [{$m_{11}= 0$?} [{$\varnothing$}, draw=red, very thick, 
  label={[circle, draw, fill=yellow!30, inner sep=1pt, font=\scriptsize]
         south east:{17}}, edge label={node[midway, sloped, inner sep=1pt, fill=white, font=\small]{yes}}][{$\kappa_{111}= 0$?}, edge label={node[midway, sloped, inner sep=1pt, fill=white, font=\small]{no}}[{$\kappa_{112}= 0$?}, edge label={node[midway, sloped, inner sep=1pt, fill=white]{yes}}
    [{$\alpha_1\to 0$ 
		}, draw=red, very thick, 
  label={[circle, draw, fill=yellow!30, inner sep=1pt, font=\scriptsize]
         north west:{18}}, edge label={node[midway, sloped, inner sep=1pt, fill=white, font=\small]{yes}}][{$m_{12}\to 0$
				}, draw=red, very thick, 
  label={[circle, draw, fill=yellow!30, inner sep=1pt, font=\scriptsize]
         north east:{19}}, edge label={node[midway, sloped, inner sep=1pt, fill=white, font=\small]{no}}]
  ]
  [{$m_{11}\to 0$\\ ($\gamma_1$ fixed)}, draw=red, very thick, 
  label={[circle, draw, fill=yellow!30, inner sep=1pt, font=\scriptsize]
         north east:{20}}, edge label={node[midway, sloped, inner sep=1pt, fill=white, font=\small]{no}}
  ]]]]]
          [{$\kappa_1\to 0$\\($\gamma_1$ fixed)}, edge label={node[midway, sloped, inner sep=1pt, fill=white, font=\small]{no}} [{$\lambda_{22}=0$?} [{$m_{22}=0$?}, edge label={node[midway, sloped, inner sep=1pt, fill=white, font=\small]{yes}} [{$\varnothing$}, draw=red, very thick, 
  label={[circle, draw, fill=yellow!30, inner sep=1pt, font=\scriptsize]
         south east:{21}}, edge label={node[midway, sloped, inner sep=1pt, fill=white, font=\small]{yes}}][{$\kappa_{222}=0$?}, edge label={node[midway, sloped, inner sep=1pt, fill=white, font=\small]{no}}[{$\kappa_{122}= 0$?}, edge label={node[midway, sloped, inner sep=1pt, fill=white]{yes}}
    [{$\alpha_2\to 0$ 
		}, draw=red, very thick, 
  label={[circle, draw, fill=yellow!30, inner sep=1pt, font=\scriptsize]
         north east:{22}}, edge label={node[midway, sloped, inner sep=1pt, fill=white, font=\small]{yes}}][{$m_{12}\to 0$
				}, draw=red, very thick, 
  label={[circle, draw, fill=yellow!30, inner sep=1pt, font=\scriptsize]
         north east:{23}}, edge label={node[midway, sloped, inner sep=1pt, fill=white, font=\small]{no}}]
  ]
  [{$m_{22}\to 0$ \\ ($\gamma_2$ fixed)}, draw=red, very thick, 
  label={[circle, draw, fill=yellow!30, inner sep=1pt, font=\scriptsize]
         north east:{24}}, edge label={node[midway, sloped, inner sep=1pt, fill=white, font=\small]{no}}
  ]]][{$\kappa_2\to 0$}, draw=red, very thick, 
  label={[circle, draw, fill=yellow!30, inner sep=1pt, font=\scriptsize]
         north east:{25}}, edge label={node[midway, sloped, inner sep=1pt, fill=white, font=\small]{no}}]]]
        ]
      ]
      [{$\kappa_{222}\to 0$\\ ($\gamma_2$ fixed)}, edge label={node[midway, sloped, inner sep=1pt, fill=white, font=\small]{no}} [{$\kappa_{122}=0$?} [{$\kappa_{112}=0$?}, edge label={node[midway, sloped, inner sep=1pt, fill=white, font=\small]{yes}} [{$\kappa_{111}=0$?}, edge label={node[midway, sloped, inner sep=1pt, fill=white, font=\small]{yes}} [{$m_{11}=0$?}, edge label={node[midway, sloped, inner sep=1pt, fill=white, font=\small]{yes}} [{$\varnothing$}, draw=red, very thick, 
  label={[circle, draw, fill=yellow!30, inner sep=1pt, font=\scriptsize]
         north west:{26}}, edge label={node[midway, sloped, inner sep=1pt, fill=white, font=\small]{yes}}][{$\alpha_1\to 0$\\ $m_{11}\ne 0$}, draw=red, very thick, 
  label={[circle, draw, fill=yellow!30, inner sep=1pt, font=\scriptsize]
         north east:{27}}, edge label={node[midway, sloped, inner sep=1pt, fill=white, font=\small]{no}}]][{$m_{11}\to 0$}, draw=red, very thick, 
  label={[circle, draw, fill=yellow!30, inner sep=1pt, font=\scriptsize]
         north east:{28}}, edge label={node[midway, sloped, inner sep=1pt, fill=white, font=\small]{no}}]][{$m_{12}\to 0$}, draw=red, very thick, 
  label={[circle, draw, fill=yellow!30, inner sep=1pt, font=\scriptsize]
         north east:{29}}, edge label={node[midway, sloped, inner sep=1pt, fill=white, font=\small]{no}}]][{$m_{22}\to 0$}, draw=red, very thick, 
  label={[circle, draw, fill=yellow!30, inner sep=1pt, font=\scriptsize]
         north east:{30}}, edge label={node[midway, sloped, inner sep=1pt, fill=white, font=\small]{no}}]]]
    ]
    [{$\kappa_{122}\to 0$\\$\kappa_{222}\to 0$}, draw=red, very thick, 
  label={[circle, draw, fill=yellow!30, inner sep=1pt, font=\scriptsize]
         north east:{31}}, edge label={node[midway, sloped, inner sep=1pt, fill=white, font=\small]{no}}]
  ]
]
]
\end{forest}
}
\caption{\textbf{Reduction tree, Branches III and IV.}
Red boxes indicate key nodes, numbered sequentially (3--31).
All but one of these are leaves corresponding to reduced potentials, while Node~3 represents an intermediate reduction stage.
Some leaves may still allow further (possibly nonlinear) reductions.
For each key node, the determining equations are solved and the corresponding symmetries are derived.
The empty set $\varnothing$ denotes that no further reductions apply.
}
\label{fig:reduction-tree1}

\end{sidewaysfigure}


 Since $\lambda_{1111} \ne 0$ we may choose $\gamma_1$ such that 
\begin{align}\label{E:bIIIka111to0}
	\kappa_{111}\to 0,
\end{align}
cf.~\eqref{E:shiftk111}. 

\paragraph{Node 3}
\label{sec:Node3}
Moreover, if
\begin{align}\label{E:Node3l1222=l2222=0}
	\lambda_{1222}=\lambda_{2222}=0,
\end{align}
no further $\kappa_{ijk}$ can be eliminated by a shift in $\gamma_2$, but we can still try to derive the possible symmetries for Node~3, without putting any restrictions on the parameter $\kappa_{112}$ of Node 3, cf.~Fig.~\ref{fig:reduction-tree1}.
With these assumptions, the determining equations yield
four different solutions, where the first is
\begin{align}
	a_5=b_{55}=b_{56}=b_{65}=0,
\end{align}
corresponding to the additional parameter conditions
\begin{align}\label{E:parCondsB3N3sol1}
	\alpha_2= m_{12}=m_{22}=\kappa_2=\kappa_{112}=\kappa_{122}=\kappa_{222}=\lambda_{12}=\lambda_{22}=0,
\end{align}
with symmetry algebra
\begin{align}\label{E:LaB3N3sol1}
	\La{a}(1)_2\oplus \La{u}(1)_Y=\on{span}(\partial_{s_2},s_2\partial_{s_2},X_Y),
\end{align}
where the index $2$ of $\La{a}(1)_2$ corresponds to the field $s_2$.  

The second solution also includes $a_6=0$ and is valid under the same set of parameter conditions as in \eqref{E:parCondsB3N3sol1},
except that $m_{22}\ne 0$.
In this case the corresponding symmetry algebra is
\begin{align}\label{E:LaB3N3sol2}
	\La{sc}_2\oplus \La{u}(1)_Y=\on{span}(s_2\partial_{s_2},X_Y),
\end{align}
where $\La{sc}_2\cong \N{R}$ is the scaling algebra of $s_2$.

The third solution,
\begin{align}\label{E:Node3sol3}
  a_5=b_{55}=b_{56}=b_{65}=0, \quad a_6 m_{22}=-\frac{b_{66}\alpha_2}{2},
\end{align}
is typical of potentials that are not maximally reduced, since it yields a symmetry generator that
depends explicitly on the potential parameters:
\begin{align}\label{E:Node3thirdsol}
  \N{R}\oplus \La{u}(1)_Y \cong
  \on{span}\Big((2s_2-\frac{\alpha_2}{m_{22}})\partial_{s_2},\, X_Y\Big),
\end{align}
which indicates that the potential should be reduced further, as we have done for the other
numbered key nodes (the leaves). See also~\eqref{E:X3SMS} and the subsequent discussion. Another drawback of intermediate reduction stages such as Node~3 is that the determining equations
may be too complicated to solve.
The solution \eqref{E:Node3sol3} is valid under the parameter conditions
\begin{align}\label{E:parCondsB3N3sol3}
	m_{12}=\kappa_2=\kappa _{112}=\kappa _{122}=\kappa _{222}=\lambda
   _{12}=\lambda _{22}=0,\quad \alpha_2\ne 0.
\end{align}

If 
\begin{align}\label{E:Node3m22ne0}
	m_{22}\ne 0,
\end{align}
we may eliminate $\alpha_2$ by an appropriate choice of $\gamma_2$ in \eqref{E:shifta2}, i.e.
\begin{align}\label{E:parCondsB3N3sol3b}
	\alpha_2\to 0,
\end{align}
cf.~the green leaf on Branch~III in Fig.~\ref{fig:reduction-tree1}.
This does not reintroduce any of the vanishing parameters, cf.~\eqref{E:shiftm12}--\eqref{E:shiftk222}.
Thus the potential is equivalent to one with $\alpha_2=0$, and the parameter-dependent generator in \eqref{E:Node3thirdsol} disappears.
Re-solving the determining equations with \eqref{E:l1112to0_2}, \eqref{E:l1111ne0l1122eq0}, \eqref{E:Node3l1222=l2222=0}, \eqref{E:parCondsB3N3sol3}, \eqref{E:Node3m22ne0}, and \eqref{E:parCondsB3N3sol3b} then yields the same algebra as in solution two, i.e.\ \eqref{E:LaB3N3sol2}.

If instead  
\begin{align}\label{E:Node3m22=0}
	m_{22}=0,
\end{align}
an inspection of \eqref{E:Node3sol3} or a re-calculation of the determining equations with  \eqref{E:l1112to0_2}, \eqref{E:l1111ne0l1122eq0}, \eqref{E:Node3l1222=l2222=0}, \eqref{E:parCondsB3N3sol3} and \eqref{E:Node3m22=0} yields the symmetry algebra
\begin{align}\label{E:L3sh2+u1}
	\La{sh}_2 \oplus \La{u}(1)_Y=\on{span}(\partial_{s_2},X_Y),
\end{align}
where $\La{sh}_2 \cong \N{R}$ is a shift algebra.

\paragraph{Leaf 4}
\label{sec:Node4}
We now consider the case where 
\begin{align}
	\lambda_{1222}=0,\quad \lambda_{2222}\ne 0,
\end{align}
in addition to \eqref{E:l1111ne0l1122eq0} and \eqref{E:bIIIka111to0} (which holds for all leaves of Branch III).
Since $\lambda_{2222}\ne 0$ we may set
\begin{align}
	\kappa_{222} \to 0,
\end{align}
by adjusting $\gamma_2$, cf.~\eqref{E:shiftk222}. We thus arrive at Leaf~4 in Fig.~\ref{fig:reduction-tree1}.
Solving the determining equations for the reduced potential at Leaf~4 shows that the symmetry algebra is trivial, \(\La{u}(1)_Y\).

\paragraph{Leaf 5}
\label{sec:Node5}
For the final case of Branch III, we assume
\begin{align}
	\lambda_{1222}\ne 0,
\end{align}
and we may set
\begin{align}
	\kappa_{122}\to 0,
\end{align}
by adjusting $\gamma_2$. Then the only solution again yields the trivial algebra $\La{u}(1)_Y$.

\subsubsection{Branch IV: $\lambda_{1111}=\lambda_{1122}=0$}
\label{sec:BranchFour}
For the final branch of potential reductions, Branch IV,
we assume
\begin{align}
	\lambda_{1111}=\lambda_{1122}=0,
\end{align}
cf.~Figs.~\ref{fig:reduction-tree1} and~\ref{fig:reduction-tree2}.

\paragraph{Leaf 6}
\label{sec:Leaf6}
Furthermore, if we assume
\begin{align}
	\lambda_{1222}=\lambda_{2222}=0,
\end{align}
we may set
\begin{align}
	\lambda_{12}\to 0,
\end{align}
through a specific choice of $\theta$, cf.~\eqref{E:O2trafol12} and the fact that $\int_{0}^{2\pi}\big(A\cos(2\theta)+ B\sin(2\theta)\big)\,d\theta=0$. This will not reintroduce any $\lambda_{ijkl}$, because all of these are zero in the first place. If we also assume
\begin{align}
	\lambda_{11}=\lambda_{22}= 0,
\end{align}
we can set
\begin{align}
	\kappa_{111}\to 0,
\end{align}
but if 
\begin{align}
	\kappa_{112}=\kappa_{122}=\kappa_{222}=0,
\end{align}
we can make a new choice for $\theta$ without introducing 
parameters $\kappa_{ijk}$. Furthermore, if 
\begin{align}\label{E:k1=k2=0Leaf6}
	\kappa_1=\kappa_2=0,
\end{align}
we may set
\begin{align}
	m_{12}\to 0,
\end{align}
i.e.~diagonalize the mass-squared matrix of the singlets through 
adjusting $\theta$. However, if 
\begin{align}\label{E:m11=m22=0Leaf6}
	m_{11}=m_{22}=0,
\end{align}
we may set
\begin{align}
	\alpha_1\to 0,
\end{align}
by a new choice of $\theta$ (or, alternatively, we could have set $\alpha_2\to 0$), and none of the vanishing parameters above will reappear. We have then ended up at Leaf~6 in Fig.~\ref{fig:reduction-tree2}.


\begin{figure}[htbp]
\centering
\scalebox{1.0}{
\begin{forest}
for tree={
  draw, rounded corners, thick,
  align=center, inner sep=4pt,
  l sep=10mm, s sep=8mm,
  edge={->, thick}
}
[{$\kappa_{222}=0$? \\ ($\lambda_{ijkl}=0, \lambda_{ij}=0, \kappa_{1jk}=0$)}, draw=blue, very thick
 [{$\kappa_1=0 \wedge \kappa_2=0\;$?}, edge label={node[midway, sloped, inner sep=1pt, fill=white, font=\small]{yes}} [{$m_{12}\to 0$ \\ ($\theta$ fixed)}, edge label={node[midway, sloped, inner sep=1pt, fill=white, font=\small]{yes}} [{$m_{11}=0$?}[{$m_{22}=0$?}, edge label={node[midway, sloped, inner sep=1pt, fill=white, font=\small]{yes}} [{$\alpha_1\to 0$\\ ($\theta$ fixed)}, draw=red, very thick, 
  label={[circle, draw, fill=yellow!30, inner sep=1pt, font=\scriptsize]
   north east:{6}}, edge label={node[midway, sloped, inner sep=1pt, fill=white, font=\small]{yes}}][{$\alpha_2 \to 0$ \\  ($\gamma_2$ fixed)}, draw=red, very thick, 
  label={[circle, draw, fill=yellow!30, inner sep=1pt, font=\scriptsize]
         north east:{7}}, edge label={node[midway, sloped, inner sep=1pt, fill=white, font=\small]{no}}]][{$\alpha_1\to 0$\\($\gamma_1$ fixed)}, edge label={node[midway, sloped, inner sep=1pt, fill=white, font=\small]{no}} [{$m_{22}=0$?} [{$\varnothing$}, draw=red, very thick, 
  label={[circle, draw, fill=yellow!30, inner sep=1pt, font=\scriptsize]
         north west:{8}}, edge label={node[midway, sloped, inner sep=1pt, fill=white, font=\small]{yes}}][{$\alpha_2\to 0$\\ ($\gamma_2$ fixed)}, draw=red, very thick, 
  label={[circle, draw, fill=yellow!30, inner sep=1pt, font=\scriptsize]
         north east:{9}}, edge label={node[midway, sloped, inner sep=1pt, fill=white, font=\small]{no}}]]] ]  ][{$\kappa_1\to 0$,\:$\kappa_2\ne 0$\\ ($\theta$ fixed)}, edge label={node[midway, sloped, inner sep=1pt, fill=white, font=\small]{no}}  [{$m_{11}=0$?} [{$m_{12}=0$?}, edge label={node[midway, sloped, inner sep=1pt, fill=white, font=\small]{yes}} [{$\varnothing$}, draw=red, very thick, 
  label={[circle, draw, fill=yellow!30, inner sep=1pt, font=\scriptsize]
         north west:{10}}, edge label={node[midway, sloped, inner sep=1pt, fill=white, font=\small]{yes}}][{$\alpha_1 \to 0$\\ ($\gamma_2$ fixed)}, draw=red, very thick, 
  label={[circle, draw, fill=yellow!30, inner sep=1pt, font=\scriptsize]
         north east:{11}}, edge label={node[midway, sloped, inner sep=1pt, fill=white, font=\small]{no}}]][{$\alpha_1 \to 0$\\ ($\gamma_1$ fixed)}, draw=red, very thick, 
  label={[circle, draw, fill=yellow!30, inner sep=1pt, font=\scriptsize]
         south east:{12}}, edge label={node[midway, sloped, inner sep=1pt, fill=white, font=\small]{no}}]] ]] 
[{$m_{22}\to 0$\\($\gamma_2$ fixed)}, edge label={node[midway, sloped, inner sep=1pt, fill=white, font=\small]{no}} [{$m_{11}=0$?} [{$\varnothing$}, draw=red, very thick, 
  label={[circle, draw, fill=yellow!30, inner sep=1pt, font=\scriptsize]
         south east:{13}}, edge label={node[midway, sloped, inner sep=1pt, fill=white, font=\small]{yes}}][{$\alpha_1\to 0$\\ ($\gamma_1$ fixed)}, draw=red, very thick, 
  label={[circle, draw, fill=yellow!30, inner sep=1pt, font=\scriptsize]
         north east:{14}}, edge label={node[midway, sloped, inner sep=1pt, fill=white, font=\small]{no}}]]] ]
\end{forest}
}
\caption{\textbf{Reduction tree, Branch IV (continued).}
Each path starting from the root in Fig.~\ref{fig:reduction-tree1} and terminating at a leaf (numbered and outlined in red) corresponds to a reduced potential.
For each leaf, the determining equations for the reduced potential are solved and the associated symmetries are derived.
The empty set $\varnothing$ denotes that no further reductions apply.}
\label{fig:reduction-tree2}
\end{figure}


We can then solve the determining equations under the parameter conditions for Leaf~6 given above, which yields
two solutions. In the first, all six parameters of the set $B_3$
are free, which corresponds to the 7-dimensional symmetry algebra \cite{Olver_1995}
\begin{align}\label{E:symAlgKinTermsKeq2}
	\La{a}(2) \oplus \La{u}(1)_Y = \on{span}(X_1,X_2,X_3,X_4,X_5,X_6,X_Y),
\end{align}
where
{\allowdisplaybreaks[0]
\begin{align}
	X_1 &=\partial_{s_1}, \\
	X_2 &=s_1\partial_{s_1}, \\
	X_3 &=s_2 \partial_{s_1}, \\
	X_4 &=\partial_{s_2}, \\
	X_5 &=s_1 \partial_{s_2}, \\
	X_6 &=s_2 \partial_{s_2},
\end{align}
}
which holds when the additional condition
\begin{align}
	\alpha_2=0,
\end{align}
is satisfied. Thus, the only parameters that need not vanish are
\begin{align}
	\mu, \lambda \in \N{R},
\end{align}
but since they may vanish, \eqref{E:symAlgKinTermsKeq2} is also the
symmetry algebra of the kinetic terms.
The affine algebra of the plane, $\La{a}(2)$ (also known as $\La{aff}(2)$), is the Lie algebra of the reparametrization group $\Lg{A}(2)$ (or $\Lg{Aff}(2)$) of the SM+2S. This is analogous to
the fact that the symmetry algebra $\La{su}(2)$ of the kinetic terms of the two-Higgs-doublet model (2HDM) is the Lie algebra of its reparametrization group $\Lg{SU}(2)$, although in this case only reparametrizations that preserve the form of the kinetic terms correspond to symmetries, whereas $\Lg{A}(2)$ will generally not preserve the form of the kinetic terms of the SM+2S, but will still correspond to symmetries of the Euler--Lagrange equations.

We can check the nature of the symmetries of the kinetic terms
$T$ (or the Lagrangian $\mc{L}_0$ with no singlet fields in the potential) by solving
\begin{align}\label{E:prXTSM2S}
	\on{pr}(k_iX_i)(\mc{L}_0)=k_2\partial_\mu s_1 \partial^\mu s_1
	+k_6\partial_\mu s_2 \partial^\mu s_2 +(k_3+k_5)\partial_\mu s_1\partial^\mu s_2=0,
\end{align}
which means that $X_1$, $X_4$ and $X_3-X_5$ (corresponding to $k_5=-k_3$), in addition to $X_Y$,
are strict variational symmetries for $\mc{L}_0$. Moreover,
by computing $E(\on{pr}(k_iX_i)(\mc{L}_0))$ we find that the only non-zero elements are
\begin{align}
	E_5(\on{pr}(k_iX_i)(\mc{L}_0))&= -2k_2 \partial_\mu \partial^\mu s_1
	-(k_3+k_5)\partial_\mu \partial^\mu s_2, \label{E:E5prXTSM2S} \\
	E_6(\on{pr}(k_iX_i)(\mc{L}_0))&= -2k_6 \partial_\mu \partial^\mu s_2
	-(k_3+k_5)\partial_\mu \partial^\mu s_1, \label{E:E6prXTSM2S}
\end{align}
hence there are no divergence symmetries for $\mc{L}_0$, since we get no new values for the parameters $k_i$ that set these expressions to zero.
However, in theories where a potential with non-vanishing linear terms is present, the symmetries $X_1$ and $X_4$ may be degraded to divergence symmetries, cf.~Corollary~\ref{C:charVarScalSyms} or Proposition~\ref{P:symSM+2S}.
Considering the Lagrangian $\mc{L}_0$ with no singlet fields in the potential, we obtain
\begin{align}
	\La{g}_\text{svar} = \La{g}_\text{var}= \spa(X_1,X_4,X_3-X_5,X_Y)
	= \La{e}(2)\oplus \La{u}(1)_Y,
\end{align}
where $\La{e}(2)$ is the Euclidean Lie algebra in two dimensions,\footnote{The Euclidean Lie algebra $\La{e}(2)$ can also be written as
$\N{R}^2 \rtimes \La{so}(2)$, where $\N{R}^2$ is spanned by the
two commuting shift generators.} while
\begin{align}
	\La{g}_\text{EL}=\spa(X_1,\ldots,X_6,X_Y)= \La{a}(2) \oplus \La{u}(1)_Y,
\end{align}
cf.~\eqref{E:symAlgIncl}. An explicit calculation of $\pr\!\big(k_iX_i\big)\,\big(E(\mc{L}_0)\big)$ for this model yields expressions containing only
terms proportional to $E_i(\mc{L}_0)$ for $i=5,6$ (i.e.~corresponding to the singlets), confirming that $\La{a}(2)$ is a
symmetry algebra of the associated field equations $E(\mc{L}_0)=0$.

The other solution to the determining equations under the parameter conditions for Leaf~6, valid when 
\begin{align}
	\alpha_2\ne 0,
\end{align}
 yields the symmetry algebra
\begin{align}
 \La{d}_4 \oplus \La{u}(1)_Y = \on{span}(X_1,X_2,X_4,X_5,X_Y).
\end{align}
Here we use the notation of \cite{andrada2005product,ovando2006four}.\footnote{The solvable Lie algebra $\La{d}_4$ can also be written as
$\N{R}^2 \rtimes \La{a}(1)$, where $\N{R}^2$ again is generated by the two shift symmetries.} 
Explicit calculations of $\on{pr}X(\mc{L})$, $E(\on{pr}X(\mc{L}))$
 and $\on{pr}X(E(\mc{L}))$ for a potential where $\alpha_2$ is the only non-vanishing parameter confirm that $X_1$ is a strict variational symmetry, $X_4$ is a divergence symmetry while $X_2$
and $X_5$ are non-variational symmetries of the field equations.

\paragraph{Leaf 7}
\label{sec:Leaf7}
If the parameter conditions follow the same path in Figures
\ref{fig:reduction-tree1} and \ref{fig:reduction-tree2} as that of Leaf~6, except that
\begin{align}
	m_{22}\ne 0,
\end{align}
in the next-to-final node (and hence $\alpha_1$ is still undetermined), we cannot adjust $\theta$ without
changing the parameters $m_{ij}$. However, we may eliminate
$\alpha_2$, i.e.~set
\begin{align}
	\alpha_2\to 0,
\end{align}
by choosing a specific value for $\gamma_2$, without changing
any other parameter, and we end up at Leaf~7 in Fig.~\ref{fig:reduction-tree2}.
We then obtain two possible symmetry algebras,
\begin{align}
\La{a}(1)_1\oplus \La{sc}_2 \oplus \La{u}(1)_Y &=\on{span}(X_1,X_2,X_6,X_Y), \quad \alpha_1=0, \\
\La{sh}_1\oplus \La{sc}_2 \oplus \La{u}(1)_Y &=\on{span}(X_1,X_6,X_Y), \quad \alpha_1\ne 0.
\end{align}

\paragraph{Leaf 8}
\label{sec:Leaf8}
Assume that the situation now is the same as before
\eqref{E:m11=m22=0Leaf6} in the paragraph on Leaf~6, but that now
\begin{align}
  m_{11}\ne 0.
\end{align}
Then we can choose a $\gamma_1$ such that
\begin{align}
  \alpha_1\to 0,
\end{align}
cf.~\eqref{E:shifta1}. If we also assume
\begin{align}\label{E:m22=0Leaf8}
  m_{22}=0,
\end{align}
we end up in Leaf~8 in Fig.~\ref{fig:reduction-tree2}, and can find the corresponding symmetries
by solving the determining equations for the parameter conditions described above. The two solutions correspond to symmetry algebras
\begin{align}
  \La{sc}_1 \oplus \La{a}(1)_2 \oplus \La{u}(1)_Y &= \spa(X_2,X_4,X_6,X_Y), \quad \alpha_2=0, \\
  \La{sc}_1 \oplus \La{sh}_2 \oplus \La{u}(1)_Y &= \spa(X_2,X_4,X_Y), \quad \alpha_2\ne 0.
\end{align}

\paragraph{Leaf 9}
\label{sec:Leaf9}
Assume conditions identical as for Leaf~8, only with
\begin{align}
  m_{22}\ne 0,
\end{align}
cf.~\eqref{E:m22=0Leaf8}. We may then choose a $\gamma_2$ such that
\begin{align}
  \alpha_2\to 0,
\end{align}
as implied by Leaf~9 in Fig.~\ref{fig:reduction-tree2}. The corresponding symmetries are
\begin{align}
  \La{gl}(2,\N{R}) \oplus \La{u}(1)_Y&= \spa(X_2,X_3,X_5,X_6,X_Y), \quad m_{11}=m_{22}, \label{E:Leaf9_gl2}\\
  \La{sc}_1 \oplus \La{sc}_2 \oplus \La{u}(1)_Y &= \spa(X_2,X_6,X_Y), \quad m_{11}\ne m_{22}. \label{E:Leaf9_sc1sc2}
\end{align}
An explicit calculation of \(\pr X(\mc{L})\) and \(E(\pr X(\mc{L}))\) for the case \eqref{E:Leaf9_gl2}, with
\(X=\sum_{i\in\sigma} k_i X_i\) and \(\sigma=\{2,3,5,6\}\), shows that the only variational subalgebra of
\(\La{gl}(2,\N{R})\) is the one-dimensional algebra spanned by \((X_3-X_5)\) (and this generator is strictly variational). Moreover, evaluating $\pr X\!\left(E(\mc{L})\right)$ yields only terms proportional to $E_i(\mc{L})$ for $i$
corresponding to the singlets, confirming that $\La{gl}(2,\N{R})$ is a symmetry algebra of the field equations.

\paragraph{Leaf 10}
\label{sec:Leaf10}
Let the parameter conditions be as above \eqref{E:k1=k2=0Leaf6} for Leaf~6, that is,
\begin{align}
  \lambda_{ijkl}=\lambda_{mn}=\kappa_{opq}=0,
\end{align}
for all indices, but now assume that \eqref{E:k1=k2=0Leaf6} does not hold, i.e.~at least one $\kappa_i$ is non-zero.
Then we can choose a $\theta$ such that
\begin{align}
  \kappa_1 \to 0, \quad \kappa_2 \ne 0.
\end{align}
Subsequently, assume  
\begin{align}
  m_{11}=m_{12}=0,
\end{align}
and we arrive at Leaf~10 in Fig.~\ref{fig:reduction-tree2}. The corresponding symmetries are then
\begin{align}
  \La{a}(1)_1 \oplus \La{u}(1)_Y &= \spa(X_1,X_2,X_Y), \quad \alpha_1=0, \\
  \La{sh}_1 \oplus \La{u}(1)_Y &= \spa(X_1,X_Y), \quad \alpha_1\ne 0.
\end{align}

\paragraph{Leaf 11}
\label{sec:Leaf11}
Assume exactly the same parameter conditions as for Leaf~10, but with
\begin{align}
  m_{12}\ne 0.
\end{align}
We may then pick a $\gamma_2$ such that
\begin{align}
  \alpha_1\to 0,
\end{align}
which corresponds to the path of parameter assumptions leading to Leaf~11. The resulting symmetry algebra is then just $\La{u}(1)_Y$.

\paragraph{Leaf 12}
\label{sec:Leaf12}
Assume the same parameter conditions as for Leaf~10, only with
\begin{align}
  m_{11}\ne 0, \quad m_{12}\in \N{R},
\end{align}
which means there are no assumptions on $m_{12}$. Since $m_{11}\ne 0$ we may shift $s_1$, i.e.~choose a $\gamma_1$
such that
\begin{align}
  \alpha_1\to 0,
\end{align}
which corresponds to Leaf~12. Then the symmetries are given by
\begin{align}
  \La{sc}_1\oplus \La{u}(1)_Y &= \spa(X_2,X_Y), \quad m_{12}=0, \\
  \La{u}(1)_Y &= \spa(X_Y), \quad m_{12}\ne 0.
\end{align}

Leaves~13--31 of the reduction tree, shown in Figs.~\ref{fig:reduction-tree1}--\ref{fig:reduction-tree2},
do not yield any new symmetry algebras beyond those already appearing in Leaves~1--12, and are therefore
presented in Appendix~\ref{app:Leaves13to31}.

The following proposition determines the variational symmetry generators that emerge from the analysis of the reduction tree:
\begin{proposition}
\label{P:symSM+2S}
All symmetry algebras $\La{g}$ that are realizable in the SM+2S
are subalgebras of the symmetry algebra of the kinetic terms,
$\La{a}(2)\oplus \La{u}(1)_Y=\spa(X_1,\ldots,X_6,X_Y)$. Moreover, all elements of any realizable symmetry algebra $\La{g}$
in the SM+2S are non-variational, except for
\begin{enumerate}[label=(\roman*)]
  \item $X_1=\partial_{s_1}$, which is a strict variational symmetry if $\alpha_1=0$, and a divergence symmetry if $\alpha_1\ne 0$.
  \item $X_4=\partial_{s_2}$, which is a strict variational symmetry if $\alpha_2=0$, and a divergence symmetry if $\alpha_2\ne 0$.
  \item $X_3-X_5=s_2\partial_{s_1}-s_1\partial_{s_2}$, which is always a strict variational symmetry.
  \item $X_Y=-\phi_2 \partial_{\phi_1}+\phi_1 \partial_{\phi_2}-\phi_4\partial_{\phi_3}+\phi_3\partial_{\phi_4}$, which is always a strict variational symmetry.
\end{enumerate}
\end{proposition}
\begin{proof}
Inspecting the results from all leaves shows that all realized symmetry algebras are subalgebras of
$\La{a}(2)\oplus \La{u}(1)_Y$. For variational symmetries, this also follows {a priori} from
Proposition~\ref{P:prXannTannVconv}. All realized symmetry algebras may be written as the span of some of the
elements in the set
\begin{align}
	 \aleph=\{X_1,\ldots,X_6,X_Y,X_3-X_5\}.
\end{align}
The calculations in \eqref{E:prXTSM2S}--\eqref{E:E6prXTSM2S} showed that $X_1$, $X_4$, and $X_3-X_5$ were the only
strict variational symmetries of the kinetic terms of the SM+2S, while the other elements of $\aleph$ were
non-variational, except for $X_Y$. The symmetry $X_Y$ is strictly variational for any SM+2S, independent of the
exact potential $V$.

By Corollary~\ref{C:charVarScalSyms}(i) (or Theorem~\ref{T:prXannTannV}(i)), $X_3-X_5$ is strictly variational, since
$\on{pr}(X_3-X_5)(T)=0$ by \eqref{E:prXTSM2S} and $a_1=a_2=0$, as there are no shifts (i.e.~no $X_1$ or $X_4$)
involved in the generator. Moreover, because $X_1$ and $X_4$ are strict variational symmetries of the kinetic terms,
we have $\on{pr}X_i(T)=0$ for $i=1,4$. Since $a_1=1$ and $a_2=0$ for $X_1$, Corollary~\ref{C:charVarScalSyms}(i),(ii)
shows that $X_1$ is strictly variational if $\alpha_1=0$ and a divergence symmetry if $\alpha_1\ne 0$. For $X_4$,
$a=(0,1)^T$, and Corollary~\ref{C:charVarScalSyms}(i),(ii) yields, in the same manner, that $X_4$ is strictly
variational if $\alpha_2=0$ and a divergence symmetry if $\alpha_2\ne 0$.
\end{proof}

\subsection{Symmetry classification in SM+2S}
\label{sec:ClassificationSymmetriesSM2S}

Orthogonal affine reparametrizations \(\tilde{s}=Os+\gamma\), where \(O\) is orthogonal,
preserve the canonical form of the kinetic terms. We denote symmetry generators and
symmetry algebras that can be mapped into each other under such reparametrizations as
(orthogonally) equivalent; cf.~Section~\ref{sec:AffineReparametrizations}. The equivalence
(or non-equivalence) of the symmetry algebras found in
Section~\ref{sec:ParameterCasesAndReductionsOfTheSM2SPotential} can then be decided by the
following principles, in addition to Proposition~\ref{P:symSM+2S} above:
\begin{itemize}
  \item Equivalent symmetry algebras must be isomorphic (in particular, they must have the
  same dimension); cf.~Proposition~\ref{P:AffRepar}.
	 \item A symmetry generator \(X\) is mapped to a symmetry generator \(\tilde{X}\) of the
  same type (SVS/DS/NVS); cf.~Proposition~\ref{P:AffRepar}.
\item Consider an affine scalar symmetry of the form \(\eta = Bs + a\).
Under orthogonal reparametrizations \(\tilde{s}=Os+\gamma\), the corresponding pairs \((B,a)\)
transform according to \eqref{E:orthEquivSym}. For convenience, we restate \eqref{E:orthEquivSym}:
\begin{equation*}
  k B_1 = O B_2 O^{T}, \qquad
  k a_1 = - O B_2 O^{T}\gamma + O a_2 .
\end{equation*}
Hence orthogonally equivalent linear parts are related by orthogonal similarity up to the overall
factor \(k\), implying that \(\rank(B)\) is preserved, and that the eigenvalues agree up to the same
proportionality factor \(k\).
\end{itemize}

\subsubsection{Inequivalent symmetry algebras of the field equations}
\label{sec:SymmetryAlgebrasFieldEquations}

We now consider the inequivalent Euler--Lagrange symmetry algebras \(\La{g}_\text{EL}\) that are realized by some SM+2S model.
As we have seen, there are many models for which the trivial \(\La{u}(1)_Y\) algebra is the
maximal symmetry algebra. The realizable, non-trivial maximal symmetry algebras of the
SM+2S are as follows:

\paragraph{(1+1)d algebras}
\label{sec:11DAlgebras}

In the following, we suppress the ubiquitous \(\La{u}(1)_Y\) summand, which is present in all
symmetry algebras. There are then five possible \((1+1)\)d algebras, specified by their
non-trivial \(1\)d summands,
\begin{align}
  \La{sh}_1,\; \La{sh}_2,\; \La{sc}_1,\; \La{sc}_2,\; \La{so}(2).
\end{align}

For \(\La{sh}_1\) and \(\La{sh}_2\), the corresponding matrices \(B_i\) (\(i=1,2\)) both vanish
(and both symmetries are variational; cf.~Proposition~\ref{P:symSM+2S}). Moreover, for
\(\La{sh}_2\) we have \(\eta_2=a_2=(0,1)^T\). Since
\begin{align}\label{E:sh2tosh1}
  a_1=(1,0)^T = O a_2,
\end{align}
cf.~\eqref{E:orthEquivSym} with \(k=1\) (and $\gamma$ arbitrary), for
\begin{align}\label{E:m01-10}
  O = \begin{pmatrix}
    0 & 1 \\
    -1 & 0
  \end{pmatrix},
\end{align}
the two symmetries are orthogonally equivalent. That is, \(\La{sh}_2\) is mapped to
\(\La{sh}_1\) under the reparametrization defined by \(O\), and we write
\begin{align}
  \La{sh}_1 \eqsim_O \La{sh}_2 \eqsim_O \La{sh},
\end{align}
where \(\La{sh}\) is used as a common representative.

The algebras \(\La{sc}_1\) and \(\La{sc}_2\) are both non-variational; cf.~Proposition~\ref{P:symSM+2S}.
The corresponding matrices
\begin{align}\label{E:B1B2scs}
  B_1=\begin{pmatrix}
    1 & 0 \\
    0 & 0
  \end{pmatrix}, \quad
  B_2=\begin{pmatrix}
    0 & 0 \\
    0 & 1
  \end{pmatrix},
\end{align}
have the same spectrum \(\{0,1\}\) (and hence proportional spectra for \(k=1\)), which is a necessary condition for orthogonal equivalence.
In fact, they are orthogonally equivalent, since
\begin{align}\label{E:B1=OB2OT}
  B_1 = O B_2 O^{T},
\end{align}
with \(a=\gamma=0\) and \(O\) given in \eqref{E:m01-10}; cf.~\eqref{E:orthEquivSym} with \(k=1\).
Thus,
\begin{align}
  \La{sc}_1 \eqsim_O \La{sc}_2 \eqsim_O \La{sc}.
\end{align}
In contrast, \(\La{sc} \neqsim_O \La{sh}\): for \(\La{sh}\) one has \(B=0\), whereas for
\(\La{sc}\) the matrix \(B\) has a nonzero eigenvalue, so the spectra cannot be proportional.
Moreover, \(\La{sh}\) is variational while \(\La{sc}\) is non-variational.

The last \((1+1)\)d algebra has a non-trivial summand \(\La{so}(2)\), and cannot be equivalent to
any of the other algebras, since the corresponding matrix is
\begin{align}
  B_{\La{so}(2)}=\begin{pmatrix}
    0 & 1 \\
    -1 & 0
  \end{pmatrix},
\end{align}
with eigenvalues \(\lambda=\pm i\), which cannot be proportional (for any \(k\in \N{R}\)) to the
real eigenvalues occurring in the other \((1+1)\)d cases. Hence, the three possible inequivalent
\((1+1)\)d symmetry algebras of the SM+2S are
\begin{align}\label{E:1pl1ELAlgebras}
  \La{sh}\oplus \La{u}(1)_Y,\; \La{sc}\oplus \La{u}(1)_Y,\; \La{so}(2)\oplus \La{u}(1)_Y.
\end{align}
Note that all algebras in \eqref{E:1pl1ELAlgebras} are abstractly isomorphic to \(\N{R}^2\),
but they induce inequivalent actions on the fields (up to orthogonal affine reparametrizations),
and are therefore not equivalent.

\paragraph{(2+1)d algebras}
\label{sec:21DAlgebras}

Omitting the ubiquitous \(\La{u}(1)_Y\) summand, the possible \((2+1)\)d algebras are
\begin{align}
  \La{a}(1)_1,\; \La{a}(1)_2,\; \La{sh}_1\oplus \La{sc}_2,\;
  \La{sc}_1\oplus \La{sh}_2,\; \La{sc}_1\oplus \La{sc}_2.
\end{align}
The only candidates for equivalence are the pair \(\La{a}(1)_i\) (they are mutually isomorphic, but not
isomorphic to the remaining algebras), and the pair of algebras of the form
\(\La{sh}\oplus\La{sc}\), since they each contain one variational and one non-variational summand;
cf.~Proposition~\ref{P:symSM+2S}.

Now, the matrix \(O\) given in \eqref{E:m01-10}, with \(k=1\) and \(\gamma=a=0\), maps \(B_2\) to
\(B_1\), cf.~\eqref{E:B1=OB2OT}, and hence maps the scaling generator \(s_2\partial_{s_2}\) to
\(s_1\partial_{s_1}\). Moreover, the same choice of \(O\), \(k\), and \(\gamma\) maps the shift
generator \(\partial_{s_2}\) to \(\partial_{s_1}\) by \eqref{E:sh2tosh1}. Hence,
\begin{equation}\label{E:a11_a12_equiv}
  \La{a}(1)_1 \eqsim_O \La{a}(1)_2 .
\end{equation}
The two algebras of the form \(\La{sh}\oplus\La{sc}\) are equivalent in a similar manner: With the
same \(O\) and \(\gamma=0\) one has \(\partial_{s_2}\mapsto \partial_{s_1}\) and
\(s_1\partial_{s_1}\mapsto s_2\partial_{s_2}\), hence
\begin{align}\label{E:shsc_equiv}
  \La{sh}_1\oplus \La{sc}_2\eqsim_O \La{sc}_1\oplus \La{sh}_2.
\end{align}

Thus, there are three inequivalent realizable \((2+1)\)d symmetry algebras in the SM+2S,
\begin{align}
  \La{a}(1)_1\oplus \La{u}(1)_Y,\qquad
  \La{sh}_1\oplus \La{sc}_2\oplus \La{u}(1)_Y,\qquad
  \La{sc}_1\oplus \La{sc}_2\oplus \La{u}(1)_Y.
\end{align}
We here retain the indices to avoid confusion with other implementations of the same abstract Lie-algebra type,
e.g.~\(\La{a}(1)=\spa(s_1\partial_{s_1},\, s_1\partial_{s_2})\).

\paragraph{(3+1)d algebras}
\label{sec:31DAlgebras}
In our analysis, we found only one \((3+1)\)d algebra of Euler--Lagrange type, namely
\begin{align}
  \La{a}(1)_1\oplus \La{sc}_2 \oplus \La{u}(1)_Y \, .
\end{align}

\paragraph{(4+1)d algebras}
\label{sec:41DAlgebras}
We found two \((4+1)\)d algebras,
\begin{align}
  \La{d}_4\oplus \La{u}(1)_Y,\qquad \La{gl}(2,\N{R})\oplus \La{u}(1)_Y,
\end{align}
which are inequivalent since they are not isomorphic. Moreover, \(\La{d}_4\) contains two
variational generators, whereas \(\La{gl}(2,\N{R})\) contains none.

\paragraph{(6+1)d algebras}
\label{sec:61DAlgebras}
There are no \((5+1)\)d symmetry algebras in the SM+2S. However, there is a \((6+1)\)d algebra,
namely the symmetry algebra of the kinetic terms,
\begin{align}
  \La{a}(2)\oplus \La{u}(1)_Y .
\end{align}

\subsubsection{Inequivalent variational symmetry algebras}
\label{sec:InequivalentVariationalSymmetryAlgebras}

According to Proposition~\ref{P:symSM+2S}, the variational symmetry generators are
\(X_1\), \(X_4\), \(X_3-X_5\), and \(X_Y\). Hence, the realizable variational symmetry algebras
\(\La{g}_\text{var}\) are precisely the subalgebras generated by these elements (with
\(\La{u}(1)_Y=\spa(X_Y)\) present in all cases). Inspecting the realized symmetry algebras of the
Euler--Lagrange equations listed in Section~\ref{sec:SymmetryAlgebrasFieldEquations}, we conclude that
the only possible inequivalent variational symmetry algebras of the SM+2S are
\begin{align}\label{E:VarSymAlgsSM2S}
  \La{sh}\oplus \La{u}(1)_Y,\qquad
  \La{so}(2)\oplus \La{u}(1)_Y,\qquad
  \La{e}(2)\oplus \La{u}(1)_Y,\qquad
  \La{sh}_1\oplus \La{sh}_2\oplus \La{u}(1)_Y.
\end{align}
Here, the first two algebras are realized as maximal symmetry algebras for suitable parameter
choices in the SM+2S. In contrast, \(\La{e}(2)\oplus \La{u}(1)_Y\) and
\begin{align}\label{E:shsq}
  \La{sh}^2 \equiv \La{sh}_1\oplus \La{sh}_2 \cong \N{R}^2
\end{align}
only occur as maximal variational subalgebras of larger symmetry algebras of the Euler--Lagrange
equations, namely \(\La{a}(2)\) and \(\La{d}_4\) (Leaf~6), respectively.

\subsubsection{Inequivalent strict variational symmetry algebras}
\label{sec:InequivalentStrictlyVariationalSymmetryAlgebras}

We now determine which of the variational symmetry algebras realized in the SM+2S,
cf.~\eqref{E:VarSymAlgsSM2S}, can also be realized as strict variational symmetry algebras
\(\La{g}_\text{svar}\). The algebra \(\La{sh}\oplus \La{u}(1)_Y\) exists in strict variational
form; for instance, \(\alpha_2\) may vanish for \(\La{sh}_2\oplus \La{u}(1)_Y\) at Leaf~3,
cf.~\eqref{E:L3sh2+u1} and Proposition~\ref{P:symSM+2S}. Moreover,
\(\La{so}(2)\oplus \La{u}(1)_Y\) is always strictly variational, since both generators are strictly
variational; cf.~Proposition~\ref{P:symSM+2S}. The algebra \(\La{e}(2)\oplus \La{u}(1)_Y\) is the
largest strict variational subalgebra of the symmetry algebra of the kinetic terms, as noted at
Leaf~6.
Finally, \(\La{sh}^2\equiv \La{sh}_1\oplus \La{sh}_2\) in \eqref{E:shsq} occurs only as a maximal
variational subalgebra of \(\La{d}_4\) at Leaf~6, under the condition \(\alpha_2\neq 0\). In this
case, Proposition~\ref{P:symSM+2S} implies that \(X_4=\partial_{s_2}\) necessarily generates a divergence symmetry,
so the maximal strict variational subalgebra of \(\La{d}_4\) is just \(\La{sh}_1\). 

Thus, the
possible inequivalent maximal strict variational symmetry algebras realizable in an SM+2S model
are
\begin{align}\label{E:SvarSymAlgsSM2S}
  \La{sh}\oplus \La{u}(1)_Y,\qquad
  \La{so}(2)\oplus \La{u}(1)_Y,\qquad
  \La{e}(2)\oplus \La{u}(1)_Y.
\end{align}
Only the first two can coincide with the full symmetry algebra \(\La{g}_\text{EL}\) of a model,
since \(\La{e}(2)\oplus \La{u}(1)_Y\) occurs only as the largest strict variational subalgebra of
\(\La{g}_\text{EL}=\La{a}(2)\oplus \La{u}(1)_Y\), the symmetry algebra of the kinetic terms;
cf.~Leaf~6.

\subsection{Algorithm for determining SM+2S symmetry algebras}
\label{sec:ASimpleAlgorithmForDecidingSM2SSymmetry}

Figures~\ref{fig:reduction-tree0}--\ref{fig:reduction-tree2}, together with the key node/leaf
analyses in Section~\ref{sec:ParameterCasesAndReductionsOfTheSM2SPotential} and
Appendix~\ref{app:Leaves13to31}, yield the following efficient algorithm for determining, for any
given numerical instance of an SM+2S model (and, in some cases, also for potentials with a small
number of undetermined parameters), the maximal symmetry algebra of each type:
\begin{enumerate}
  \item Start at the root of Figures~\ref{fig:reduction-tree0}--\ref{fig:reduction-tree1}, and eliminate
  \(\lambda_{1112}\) by an \(\Lg{SO}(2)\) rotation.

  \item Recompute all potential parameters using
  eqs.~\eqref{E:trafoParametersSO2}--\eqref{E:trafoParametersSO2lambda2222} or
  eqs.~\eqref{E:shifta1}--\eqref{E:shiftk222}, employing the same rotation or shift as in the
  previous step.

\item Follow the appropriate path in Figures~\ref{fig:reduction-tree0}--\ref{fig:reduction-tree2}
until reaching a node whose content prescribes one or two reparametrizations of the form
\(p\to 0\) (for some parameter \(p\)).

\item If the node is non-key (i.e.\ neither outlined in red nor numbered), perform the parameter
elimination indicated at the node via an appropriate rotation or shift, and return to Step~2.
Otherwise, proceed.

  \item At a key node \(n\), use the corresponding key node analysis (Leaf~\(n\) for \(n\neq 3\), or
  Node~3) in Section~\ref{sec:ParameterCasesAndReductionsOfTheSM2SPotential} and
  Appendix~\ref{app:Leaves13to31} to determine the admissible maximal symmetry algebra(s)
  \(\La{g}_\text{EL}\): If the key node admits a unique algebra, conclude that algebra. If the key
  node is marked with \(\varnothing\), determine the algebra by comparing the parameters with the parameter conditions at the key node. Otherwise, perform the elimination(s)
  indicated at the key node, recompute the parameters, and determine the algebra from the
  recalculated potential parameters.

  \item Conclude that the maximal variational symmetry algebra \(\La{g}_\text{var}\) for the potential
  is generated by \(\La{g}_\text{EL}\cap \spa(W)\), where \(W=\{X_1,X_4,X_3-X_5,X_Y\}\).

  \item Conclude that the maximal strict variational symmetry algebra \(\La{g}_\text{svar}\) for the
  potential is generated by \(\La{g}_\text{EL}\cap \spa(W')\), where
  \(W'=\{X_1\delta_{0,{\alpha}_1},\,X_4\delta_{0,{\alpha}_2},\,X_3-X_5,\,X_Y\}\), for recalculated parameters ${\alpha}_i$.
\end{enumerate}
In the last step, \(X_1\) is excluded from $W'$ if the linear parameter \(\alpha_1\neq 0\),
cf.~Proposition~\ref{P:symSM+2S}(i), and similarly for \(X_4\). For this final step, the parameters
\(\alpha_i\) may sometimes have to be recomputed after performing the elimination prescribed at the
key node. This is unnecessary if \(\La{g}_\text{EL}\) does not contain \(X_1\) or \(X_4\).

In a discrete parameter scan one will generically not satisfy exact equalities such as \(p=0\),
unless \(p\) can be eliminated by a reparametrization. One can therefore introduce a
tolerance \(t>0\) and treat \(|p|<t\) as zero; see, e.g., \cite{Plantey:2024yfm}.

\section{Summary and outlook}
\label{sec:SummaryAndOutlook}
We have performed a classification of the scalar Lie point symmetries for the SM+S and SM+2S models.
For the SM+S, we found four inequivalent realizable scalar Lie point symmetry algebras
$\La{g}_\text{EL}$ of the Euler--Lagrange equations,
\begin{align}
  \La{u}(1)_Y,\; \La{sh}\oplus \La{u}(1)_Y,\; \La{sc}\oplus \La{u}(1)_Y,\;  \La{a}(1)\oplus \La{u}(1)_Y.
\end{align}
The inequivalent realizable variational scalar Lie point symmetry algebras $\La{g}_\text{var}$ and strict variational
scalar Lie point symmetry algebras $\La{g}_\text{svar}$ of the SM+S were
\begin{align}
  \La{u}(1)_Y  \quad \text{and}  \quad \La{sh}\oplus \La{u}(1)_Y.
\end{align}
Here the shift algebra $\La{sh}$ occurs as a strict variational symmetry algebra only when it is a
subalgebra of the larger non-variational symmetry algebra $\La{a}(1)$.
A simple algorithm for determining the symmetry algebras of a given SM+S model instance was provided in Section~\ref{sec:ASimpleAlgorithmForDecidingSMSSymmetry}.

We found 11 inequivalent realizable scalar Lie point symmetry algebras \(\La{g}_\text{EL}\) of the Euler--Lagrange equations of the SM+2S model. They can be written as\begin{gather}\label{E:summarySM2SEL}
\begin{aligned}
&\La{u}(1)_Y,\\
&\La{sh}\oplus \La{u}(1)_Y,\quad
 \La{sc}\oplus \La{u}(1)_Y,\quad
 \La{so}(2)\oplus \La{u}(1)_Y,\\
&\La{a}(1)\oplus \La{u}(1)_Y,\quad
 \La{sc}^2\oplus \La{u}(1)_Y,\quad
 \La{sh}\oplus \La{sc}\oplus \La{u}(1)_Y,\\
&\La{a}(1)\oplus \La{sc}\oplus \La{u}(1)_Y,\\
&\La{d}_4\oplus \La{u}(1)_Y,\quad
 \La{gl}(2,\N{R})\oplus \La{u}(1)_Y,\\
&\La{a}(2)\oplus \La{u}(1)_Y.
\end{aligned}
\end{gather}
Some of these algebras are isomorphic as abstract Lie algebras, yet their actions on the scalar fields are inequivalent; that is, their concrete vector-field realizations are not related by an orthogonal affine reparametrization of the scalar fields.

The inequivalent realizable variational scalar Lie point symmetry algebras $\La{g}_\text{var}$ were 
\begin{align}\label{E:summarySM2Svar}
 \La{u}(1)_Y,\; \La{sh}\oplus \La{u}(1)_Y,\; \La{so}(2)\oplus \La{u}(1)_Y,\;
  \La{e}(2)\oplus \La{u}(1)_Y,\; \La{sh}^2\oplus \La{u}(1)_Y,
\end{align}
where the last two occur only as variational subalgebras of larger symmetry algebras $\La{g}_\text{EL}$ of the field equations.

The inequivalent realizable strict variational scalar Lie point symmetry algebras $\La{g}_\text{svar}$ of the SM+2S were
\begin{align}\label{E:summarySM2Ssvar}
 \La{u}(1)_Y,\; \La{sh}\oplus \La{u}(1)_Y,\; \La{so}(2)\oplus \La{u}(1)_Y,\;
  \La{e}(2)\oplus \La{u}(1)_Y,
\end{align}
where the last algebra occurs only as a strict variational subalgebra of a larger symmetry algebra $\La{g}_\text{EL}$ of the field equations.
We note that we would have missed the variational symmetry algebra $\La{sh}^2\oplus \La{u}(1)_Y$ if we had restricted our analysis to strict variational symmetries. 

To distinguish between different types of symmetry generators and symmetry algebras, we proved in Corollary~\ref{C:charVarScalSyms} a characterization of strict variational (SVS), divergence (DS), and non-variational (NVS) generators, valid for a wide class of Lagrangians with potentials. In
Proposition~\ref{P:symSM+2S} we applied this characterization to the SM+2S to determine the
SVS/DS/NVS nature of its symmetry generators. Moreover, in Proposition~\ref{P:AffRepar} we showed
that this classification is preserved under affine reparametrizations.

Section~\ref{sec:ASimpleAlgorithmForDecidingSM2SSymmetry} furthermore presented an efficient
algorithm for identifying the scalar Lie point symmetry algebras of the SM+2S model for any given
numerical choice of the potential parameters (or, in some cases, with a few parameters left
undetermined), without explicitly solving the determining equations. This parameter-based identification can be useful for numerical parameter scans.

The automorphism groups of the Lie point symmetry algebras found in this work may be exploited to determine the realizable discrete symmetries of the SM+S and SM+2S models; see
\cite{hydon1999use,hydon2000construct}. Such an analysis lies outside the scope of the present work.

For classifications of Lie point symmetries in larger models such as
SM+\(K\)S with \(K>2\), more indirect methods may provide useful support for
the direct reduction-tree approach, given the increasing number of parameters
and the greater number of branches in the reduction tree
(cf.~Fig.~\ref{fig:reduction-tree1} for \(K=2\)).
One possible strategy (at least for variational symmetries) is to first determine the symmetry
algebra $\La{g}$ of the kinetic terms, and then identify the realizable Lie symmetry algebras of
the full Lagrangian by choosing one representative subalgebra $\La{h}\subset\La{g}$ from each class of subalgebras
equivalent under reparametrizations, and determining whether $\La{h}$ is
realizable for some choice of potential parameters. In general, not all such subalgebras are
realizable. For instance, in the SM+2S the algebra $\La{k}\oplus \La{u}(1)_Y$, with the boost
$\La{k}=\spa(s_2\partial_{s_1}+s_1\partial_{s_2})$, is not equivalent (under reparametrizations) to any of
the algebras \eqref{E:summarySM2SEL}--\eqref{E:summarySM2Ssvar} above (the associated linear part
$B$ of $\La{k}$ has eigenvalues $\pm 1$) and is therefore not realizable.
However, this route would not yield algorithms for deciding the symmetry algebras without invoking explicit symmetry calculations.

Another approach for models with singlets is to restrict attention to
tree-level potentials with a stationary point and to expand about such a
point, thereby eliminating the linear (tadpole) singlet terms. Nevertheless,
this would exclude tree-level potentials lacking a stationary point even
though the corresponding effective potential may have one, resulting in a less
general analysis. Moreover, expanding about such a stationary point generally leads to a more
complicated kinetic/gauge sector; see Sections~\ref{sec:LinearTermsSMS} and
\ref{sec:LinearTerms} for a discussion of both caveats.


\appendix

\section{Proofs of the technical results in Section~\ref{sec:GeneralResultsforSymmetryClassification}}
\label{app:TechnicalProofs}
This appendix collects the proofs of the technical results stated in
Section~\ref{sec:GeneralResultsforSymmetryClassification}. The proof of
Theorem~\ref{T:prXannTannV} uses the following standard definition of the
adjoint Fréchet derivative:

For an \(r\)-tuple of differential functions \(P=(P_1,\ldots,P_r)\), the
\emph{adjoint Fréchet derivative} \(D_P^\ast\) is the \(q\times r\) matrix
differential operator defined by \cite{olver1998applications}
\begin{align}\label{E:FrechetAdj}
  (D_P^\ast)_{ij}
  = \sum_{J} (-1)^{|J|} D_J \circ \frac{\partial P_j}{\partial y^i_J}\,,
\end{align}
where $D_J$ is given by \eqref{E:itTotder} and acts on the product of $\partial P_j/\partial y^i_J$ and the function to which $D_P^\ast$ is applied.
Equivalently, for any test function $\psi=(\psi_1,\ldots,\psi_r)$,
\begin{align}
  (D_P^\ast\psi)_i
  = \sum_{j,J} (-1)^{|J|} D_J\!\left(\frac{\partial P_j}{\partial y^i_J}\,\psi_j\right).
\end{align}
\PotentialSymmetryTheorem*
\begin{proof}
 As in the corresponding proof in \cite{Solberg:2025ybf}, we use the following identity, which holds in general \cite{olver1998applications}:
\begin{align}\label{E:commFrechet1}
	\on{pr}X_Q(E(\mc{L}))= E(\on{pr}X_Q(\mc{L}))-D_Q^\ast E(\mc{L}),
\end{align}
where $D_Q^\ast$ denotes the adjoint of the Fréchet derivative with respect to the characteristic $Q$ associated with $X$.
Since $X$ is in evolutionary form, \eqref{E:commFrechet1} yields
\begin{align}\label{E:commFrechet2}
	\on{pr}X(E(\mc{L}))= E(\on{pr}X(\mc{L}))-D_Q^\ast E(\mc{L}).
\end{align}
 Assuming $E(\mc{L})=0$,
the left-hand side of \eqref{E:commFrechet2} vanishes 
since $X$ is a symmetry of $E(\mc{L})=0$, and $E(\mc{L})=0$ implies the last term in \eqref{E:commFrechet2} vanishes as well, since total derivatives of $E(\mc{L})$ vanish when evaluated on any solution of the Euler--Lagrange equations.
Hence, under the assumption ${E(\mc{L})=0}$,
\begin{align}
	E\big(\on{pr}X(\mc{L})\big)=0.
\end{align}
Moreover, if we also assume 
\begin{align}\label{E:prXTisAdiv}
\on{pr}X(T)=d_\mu \beta^\mu, 	
\end{align}
for a (possibly vanishing) divergence $d_\mu \beta^\mu$,
then for all $i$,
\begin{align}\label{E:derPrPot0}
	E_i(\on{pr}X(V))=\frac{\partial}{\partial y^i}\on{pr}X(V)=0,
\end{align}
since $E(d_\mu \beta^\mu)=0$. Now, \eqref{E:derPrPot0} also 
must hold when $E(\mc{L})\ne 0$, since the system $E(\mc{L})= 0$
had no polynomial consequences. This means that $\on{pr}X(V)=C$ for a constant $C\in \N{R}$, and then
\begin{align}
	C=\on{pr}X(V)=X(V)=\eta^i\partial_{y^i}V,
\end{align}
and considering $y=0$ yields
\begin{align}
	C=\eta^i(0)(\partial_{y^i}V)|_{y=0} = a_i \alpha_i
\end{align}
where $\alpha_i=0$ for $i>m$. We have thus proved part (i) of the theorem. Moreover,
\begin{align}\label{E:prXLTV}
	\on{pr}X(\mc{L})=\on{pr}X(T)-a_i \alpha_i,
\end{align}
which proves part (ii), cf.~\eqref{E:condStrictVarSym} with $\xi=0$. 
Moreover, \eqref{E:prXLTV} can be written
\begin{align}
	\on{pr}X(\mc{L})= d_\mu (\beta^\mu-\frac{a_i \alpha_i x^\mu}{4})
	\equiv d_\mu B^\mu .
\end{align}
This proves part (iii) provided we can show that $\on{pr}X(T)=d_\mu\beta^\mu$ cannot be a non-zero real constant (and hence cannot cancel a non-zero $-\on{pr}X(V)$). Consequently, the strict variational condition $\on{pr}X(\mc L)=\on{pr}X(T)-\on{pr}X(V)=0$ can only be satisfied if $\on{pr}X(T)$ and $\on{pr}X(V)$ vanish separately, and thus cannot be fulfilled in the situations covered by (iii):
As $\xi=0$ the prolongation of $X$ applied to $T$ yields
\begin{align}\label{E:prXTxi=0}
	\on{pr}X(T) = \sum_{i,J} \Big(D_J \eta^i(y)\Big)\frac{\partial T}{\partial y^i_J}\in \N{R}\{y\},
\end{align}
 which means that each term of the polynomial $\on{pr}X(T)$ at least contains one factor from $\varphi^c \cup \on{der}(y)$ by assumption on $T$, that is, 
\begin{align}\label{E:prXTproptovarphicorder0}
	\on{pr}X(T) =\varphi^c_{\sigma(i)} q_i +
		y^{j}_{\tilde{J}(j)} r_j,
\end{align}
where $\sigma(i)\in \{1,2,\ldots,(q-m)\}$, the multi-index $|\tilde{J}(j)|\geq 1$
and $q_i, r_j \in \N{R} \{ y\}$ for all $i,j$. Indeed, if $|J|\geq 1$ in \eqref{E:prXTxi=0},
then $D_J\eta^i$ contains derivatives in each term, while for $|J|= 0$
each term of ${\partial T}/{\partial y^i_J}$ is at least linear in the fields $\varphi^c$.
Therefore, $\on{pr}X(T)$ can never equal a non-vanishing constant, 
and hence a strict variational symmetry requires $\on{pr}X(T)=0$ and $\on{pr}X(V)=0$ separately.
\end{proof}

\PotentialSymmetryConverse*
\begin{proof}
In both cases,
\begin{align}\label{E:prXVisaPolinVarphi}
	\on{pr}X(V) =X(V) = \eta^i(\varphi)\frac{\partial V}{\partial y^i}\in
	\N{R}[\varphi],
\end{align}
  is an ordinary, real polynomial in the variables $\varphi$.
On the other hand, since $\xi=0$ the prolongation of $X$ applied to $T$ yields
\begin{align}
	\on{pr}X(T) = \sum_{i,J}\Big(D_J \eta^i(\varphi)\Big)\frac{\partial T}{\partial y^i_J}\in \N{R}\{y\},
\end{align}
 which, as in the proof of Theorem \ref{T:prXannTannV}, means that each term of the polynomial $\on{pr}X(T)$ at least contains one factor from $\varphi^c \cup \on{der}(y)$, that is, 
\begin{align}\label{E:prXTproptovarphicorder}
	\on{pr}X(T) =\varphi^c_{\sigma(i)} q_i +
		y^{j}_{\tilde{J}(j)} r_j, 
\end{align}
where $\sigma(i)\in \{1,2,\ldots,(q-m)\}$, the multi-index $|\tilde{J}(j)|\geq 1$
and $q_i, r_j \in \N{R} \{ y\}$ for all $i,j$. Thus, comparing eqs.~\eqref{E:prXVisaPolinVarphi} and \eqref{E:prXTproptovarphicorder}, there can be no cancellations between the terms of $\on{pr}X(V)$ and $\on{pr}X(T)$. In case (i), we have $\on{pr}X(\mc{L})=0$, since $X$ is a strict variational symmetry, with $\xi=0$, cf.~\eqref{E:condStrictVarSym}. Then, in the absence of cancellations between $\on{pr}X(T)$ and $\on{pr}X(V)$, they must be annihilated separately, i.e.
\begin{align}\label{E:StrictVarSymAnnBothTandV}
	\on{pr}X(T)=\on{pr}X(V)=0.
\end{align}

For case (ii), assume $X$ is a divergence symmetry. Then the total divergence 
\begin{align}\label{E:prXLtotDiv}
	\on{pr}X(\mc{L})=d_\mu B^\mu \in \N{R}\{y\},
\end{align}
 will be a non-zero, polynomial function in $\N{R}\{y\}$. 
We will now show that since $\on{pr}X(\mc{L})$ equals a total divergence, the potential $V$ cannot contribute with more than a constant $C$ to $\on{pr}X(\mc{L})$, that
is,  $\on{pr}X(V)=C$: $\on{pr}X(\mc{L})$ is a total divergence if and only if 
\begin{align}\label{E:E(div)}
	E_i(\on{pr}X(\mc{L}))= \Big( \frac{\partial }{\partial y^i}
 - d_\mu\!\left(\frac{\partial }{\partial y^i_{,\mu}}\right)
 + d_\mu d_\nu\!\left(\frac{\partial }{\partial y^i_{,\mu\nu}}\right)
 - \cdots\Big) \on{pr}X(\mc{L})=0 \,,
\end{align}
   for all $i$, cf.~\eqref{E:EtotDiv}, where $d_\mu$ is given by \eqref{E:dmuexpanded}. 
Hence, 
\begin{align}\label{E:TermsFromPotAreZero}
	\frac{\partial \on{pr}X(V)}{\partial y^i}=0,
\end{align}
for all $i$, which is evident for $i>m=\on{dim}(\varphi)$ since $\on{pr}X(V)$ only depends on the fields $\varphi$. 
In case $i\leq m$ (and hence ${\partial }/{\partial y^i}={\partial }/{\partial \varphi^i}$) \eqref{E:TermsFromPotAreZero} also has to hold,
for if the expression was non-zero, it could not have been canceled by any other
term in \eqref{E:E(div)}, since all other terms then are proportional to variables $\varphi^c$ or derivatives, as spacetime variables do not occur explicitly in $\mc{L}$ and the total derivatives $d_\xi$ of \eqref{E:E(div)} hence will generate derivatives in each term.
 Therefore, 
\begin{align}\label{E:prXV=C}
	\on{pr}X(V)=C \in \N{R},
\end{align}
which implies  
\begin{align}\label{E:exprForTotDivT}
	\on{pr}X(T)= d_\mu \beta^\mu \equiv  d_\mu  (B^\mu+\frac{Cx^\mu}{4}), 
\end{align}
where $B^\mu$ was given in \eqref{E:prXLtotDiv}. 
As given by \eqref{E:prXTproptovarphicorder}, $\on{pr}X(T)$ cannot contain a constant term (hence, the
constant $C$ in \eqref{E:exprForTotDivT} must be canceled by another term in $B^\mu$).
Furthermore, $d_\mu \beta^\mu$
and $C$ cannot both be vanishing, since we then would have a strict variational symmetry, contrary to our assumption.
Finally, 
\eqref{E:prXVisaPolinVarphi},
\eqref{E:StrictVarSymAnnBothTandV} and \eqref{E:prXV=C} with all fields set to zero yields (for both cases)
\begin{align}
	\on{pr}X(V)=\eta^i(0)\partial_{y^i}(V)|_{y=0}=a_i\alpha_i,
\end{align}
which concludes the proof.
\end{proof}

\AffineReparametrizationProp*
\begin{proof}
Assume that $X$ is a point symmetry of $E(\mc L)=0$.
By the chain rule applied to \eqref{E:GenRepar},
\begin{align}
  \partial_{\tilde y^i}
  =\frac{\partial y^k}{\partial\tilde y^i}\,\partial_{y^k}
  =A^{-1}_{ki}\partial_{y^k}.
\end{align}
Hence, with $\tilde\eta^i=A_{ij}\eta^j$,
\begin{align}\label{E:x'eqX}
  \tilde X
  = \xi^\mu\partial_\mu + \tilde\eta^i\partial_{\tilde y^i}
  = \xi^\mu\partial_\mu + A_{ij}\eta^j A^{-1}_{ki}\partial_{y^k}
  = \xi^\mu\partial_\mu + \eta^k\partial_{y^k}
  = X.
\end{align}
Thus $X$ and $\tilde X$ represent the same vector field expressed in the $(x,y)$ and $(x,\tilde y)$ coordinates, respectively.

Since $A$ is constant, \eqref{E:GenRepar} and \eqref{E:Xtieti} extend to all derivatives:
\begin{align}\label{E:jetTrafoAff}
  \tilde y^i_J = A_{ij}y^j_J,
  \qquad
  \tilde\eta^i_J = A_{ij}\eta^j_J,
\end{align}
for all multi-indices $J$. Consequently,
\begin{align}\label{E:reparTildeyiJ}
  \partial_{\tilde y^i_J}=A^{-1}_{ki}\partial_{y^k_J},
\end{align}
and the prolongations coincide,
\begin{align}\label{E:prX'eqprX}
  \on{pr}\tilde X=\on{pr}X,
\end{align}
cf.~\eqref{E:kProlong}.

Let $\tilde E$ denote the Euler operator with respect to the fields $\tilde y$.
Then the Euler operators $E_i$ transform as
\begin{align}\label{E:EulerTrafoAff}
  \tilde E_i = A^{-1}_{ki}E_k
  \qquad\Longleftrightarrow\qquad
  \tilde E = A^{-T}E,
\end{align}
and hence, using \eqref{E:Ly=Ltyt},
\begin{align}\label{E:EtildeLtilde}
  \tilde E(\tilde{\mc L}) = A^{-T}E(\mc L).
\end{align}
Therefore,
\begin{align}
  \on{pr}\tilde X\!\big(\tilde E(\tilde{\mc L})\big)\Big|_{\tilde E(\tilde{\mc L})=0}
  &=\on{pr}X\!\big(A^{-T}E(\mc L)\big)\Big|_{A^{-T}E(\mc L)=0}\nn\\
  &=A^{-T}\on{pr}X\!\big(E(\mc L)\big)\Big|_{E(\mc L)=0}=0,
\end{align}
since $A^{-T}$ is invertible and $X$ is a symmetry of $E(\mc L)=0$.
Thus $\tilde X$ is a point symmetry of $\tilde{E}(\tilde{\mc L})=0$.
The converse direction follows by the same argument, since \eqref{E:GenRepar} is invertible.

To compare symmetry types, note that \eqref{E:prX'eqprX} and \eqref{E:Ly=Ltyt} imply
\begin{align}
  \on{pr}\tilde X(\tilde{\mc L})=\on{pr}X(\mc L).
\end{align}
Hence $\tilde X$ is a strict variational symmetry of $\tilde{\mc L}$ if and only if $X$ is a strict variational symmetry of $\mc L$ by \eqref{E:condStrictVarSym}. Likewise, $\tilde X$ is a divergence symmetry if and only if $X$ is, since
\begin{align}
  d_\mu\tilde\beta^\mu[\tilde y]=d_\mu\beta^\mu[y]
\end{align}
is a total divergence in either coordinate system, where $\tilde\beta^\mu$ is defined by
$\tilde\beta^\mu[\tilde y]=\beta^\mu[y]\equiv \beta^\mu\big(x,y,y^{(1)},\ldots\big)$ with $y=A^{-1}(\tilde y-\gamma)$.

Finally, define $\theta:\La g\to\tilde{\La g}$ by $\theta(X)=\tilde X$. By \eqref{E:x'eqX} we have $\tilde X=X$, so $\tilde{\La g}=\La g$ and $\theta$ is the identity map. Hence $\theta$ is a Lie algebra isomorphism.
\end{proof}
In Proposition \ref{P:AffRepar} we state only that $\La g$ and $\tilde{\La g}$ are isomorphic, rather than equal, although \eqref{E:x'eqX} identifies them. This is because they generally appear different when expressed in old and new coordinates, and because we will use the proposition to test whether, after reparametrizing one algebra and then suppressing the tildes, it may coincide with another fixed algebra; cf.~\eqref{E:orthEqAlgebras} and Sections \ref{sec:DeterminingEquationsSMS} and \ref{sec:SymmetryAlgebrasFieldEquations}. By suppressing the tildes, one loses the explicit identification and retains only the isomorphism.

\section{Completion of the reduction tree: Leaves 13--31}
\label{app:Leaves13to31}
This appendix completes the reduction-tree analysis by considering Leaves~13--31,
shown in Figs.~\ref{fig:reduction-tree1}--\ref{fig:reduction-tree2}. These leaves do not
yield any symmetry algebras beyond those already appearing in Leaves~1--12.

\paragraph{Leaf 13}
\label{sec:Leaf13}
Suppose, as for Leaf~6, that
\begin{align}
  \lambda_{ijkl}=0,\quad \lambda_{ij}=0,\quad \kappa_{1jk}=0,
\end{align}
for all indices, but that
\begin{align}
  \kappa_{222}\ne 0.
\end{align}
Then we can pick a $\gamma_2$ such that
\begin{align}
  m_{22}\to 0,
\end{align}
cf.~\eqref{E:shiftm22}. Moreover assume
\begin{align}
  m_{11}=0.
\end{align}
The symmetries corresponding to this case, given by Leaf~13, are
\begin{align}
  \La{a}(1)_1\oplus \La{u}(1)_Y &= \spa(X_1,X_2,X_Y), \quad \kappa_1=m_{12}=\alpha_1=0, \\
  \La{sh}_1\oplus \La{u}(1)_Y &= \spa(X_1,X_Y), \quad \kappa_1=m_{12}=0,\; \alpha_1\ne 0, \\
  \La{u}(1)_Y &= \spa(X_Y), \quad \text{for all other parameter values.}
\end{align}

\paragraph{Leaf 14}
\label{sec:Leaf14}
Assume the same parameter conditions as in Leaf~13, except that
\begin{align}
  m_{11}\ne 0,
\end{align}
and hence we can choose a $\gamma_1$ such that
\begin{align}
  \alpha_1\to 0,
\end{align}
as displayed in Leaf~14 in Fig.~\ref{fig:reduction-tree2}.
The corresponding symmetries are then
\begin{align}
  \La{sc}_1 \oplus \La{u}(1)_Y &= \spa(X_2,X_Y), \quad m_{12}=0, \\
  \La{u}(1)_Y &= \spa(X_Y), \quad m_{12}\ne 0.
\end{align}

\paragraph{Leaf 15}
\label{sec:Leaf15}
Now suppose
\begin{align}\label{E:parCondsLeaf15}
  \lambda_{ijkl}=0,\quad \lambda_{ij}=0,\quad \kappa_{111}=\kappa_{112}=0,
\end{align}
in the same manner as for Leaf~6, but with
\begin{align}
  \kappa_{122}\ne 0.
\end{align}
Then we may first pick a $\gamma_2$ that eliminates $m_{12}$, thereafter a $\gamma_1$ that eliminates $m_{22}$
without reintroducing $m_{12}$, cf.~\eqref{E:shiftm12} and \eqref{E:shiftm22}, that is,
\begin{align}
  m_{12} &\to 0, \\
  m_{22} &\to 0,
\end{align}
corresponding to Leaf~15 in Fig.~\ref{fig:reduction-tree1}. The symmetry in this case is just $\La{u}(1)_Y$.

\paragraph{Leaf 16}
\label{sec:Leaf16}
Assume the parameter conditions are the same as in 
\eqref{E:parCondsLeaf15}, with the only difference
\begin{align}
	\kappa_{112}\ne 0.
\end{align}
Then we may first choose a $\gamma_2$ that eliminates
$m_{11}$, then pick a $\gamma_1$ that removes $m_{12}$:
\begin{align}
	m_{11}\to 0, \\
	m_{12}\to 0,
\end{align}
as shown in Leaf 16 in Fig.~\ref{fig:reduction-tree1}.
The only possible symmetry algebra is then $\La{u}(1)_Y$.

\paragraph{Leaf 17}
\label{sec:Leaf17}
Suppose, in the same way as for Leaf~6, that
\begin{align}
  \lambda_{ijkl}=\lambda_{11}=\lambda_{12}=0,
\end{align}
but, in contrast to the case for Leaf~16 (and Leaf~6),
\begin{align}
  \lambda_{22}\ne 0.
\end{align}
Then a particular choice for $\gamma_2$ yields
\begin{align}
  \kappa_2\to 0,
\end{align}
cf.~\eqref{E:shiftk2}. Furthermore, assume
\begin{align}
  m_{11}=0,
\end{align}
which leads us to Leaf~17.
The symmetries in this case are then
\begin{align}
  \La{a}(1)_1\oplus \La{u}(1)_Y &= \spa(X_1,X_2,X_Y), \quad \kappa_1=\kappa_{111}=\kappa_{112}=\kappa_{122}=m_{12}=\alpha_1=0, \\
  \La{sh}_1 \oplus \La{u}(1)_Y &= \spa(X_1,X_Y), \quad \kappa_1=\kappa_{111}=\kappa_{112}=\kappa_{122}=m_{12}=0,\; \alpha_1\ne 0, \\
  \La{u}(1)_Y &= \spa(X_Y), \quad \text{for all other parameter values}.
\end{align}

\paragraph{Leaf 18}
\label{sec:Leaf18}
Now suppose the parameters are the same as for Leaf~17, only with
\begin{align}
  m_{11}\ne 0.
\end{align}
If we in addition assume
\begin{align}
  \kappa_{111}=\kappa_{112}=0,
\end{align}
we can eliminate $\alpha_1$ by a shift involving a special value of $\gamma_1$,
\begin{align}
  \alpha_1 \to 0,
\end{align}
cf.~\eqref{E:shifta1}, corresponding to Leaf~18 in Fig.~\ref{fig:reduction-tree1}. Then the possible symmetries are
\begin{align}
  \La{sc}_1 \oplus \La{u}(1)_Y &= \spa(X_2,X_Y), \quad \kappa_1=\kappa_{122}=m_{12}=0, \\
  \La{u}(1)_Y &= \spa(X_Y), \quad \text{for all other parameter values}.
\end{align}

\paragraph{Leaf 19}
\label{sec:Leaf19}
We now make the same assumptions as for Leaf~18, only with
\begin{align}
  \kappa_{112}\ne 0,
\end{align}
which makes it possible to eliminate $m_{12}$,
\begin{align}
  m_{12}\to 0,
\end{align}
by some choice of $\gamma_1$, as displayed in Leaf~19 in Fig.~\ref{fig:reduction-tree1}. Here $\alpha_1$ is
undetermined, i.e.~$\alpha_1\in \N{R}$. Then, at Leaf~19 the only admissible symmetry algebra is \(\La{u}(1)_Y\).

\paragraph{Leaf 20}
\label{sec:Leaf20}
Assume, as for the two previous leaves, that
\begin{align}
  \lambda_{ijkl}=\lambda_{11}=\lambda_{12}=\kappa_2=0,
\end{align}
whereas
\begin{align}
  \lambda_{22}, m_{11}, \kappa_{111} \ne 0.
\end{align}
Then we can set $m_{11}$ to zero by a choice of $\gamma_1$,
\begin{align}
  m_{11}\to 0,
\end{align}
corresponding to Leaf~20. The only symmetry is then $\La{u}(1)_Y$.

\paragraph{Leaf 21}
\label{sec:Leaf21}
Suppose, as for Leaf~17, that
\begin{align}
  \lambda_{ijkl}=\lambda_{12}=0,
\end{align}
for all indices, but that
\begin{align}
  \lambda_{11}\ne 0,
\end{align}
which allows us to set
\begin{align}
  \kappa_{1}\to 0,
\end{align}
by some choice of $\gamma_1$. Moreover, if
\begin{align}
  \lambda_{22}=m_{22}=0,
\end{align}
the situation is as given by Leaf~21 in Fig.~\ref{fig:reduction-tree1},
and we find the following symmetries:
\begin{align}
  \La{a}(1)_2 \oplus \La{u}(1)_Y &= \spa(X_4,X_6,X_Y), \quad \kappa_2=\kappa_{112}=\kappa_{122}=\kappa_{222}=m_{12}=\alpha_2=0, \\
  \La{sh}_2 \oplus \La{u}(1)_Y &= \spa(X_4,X_Y), \quad \kappa_2=\kappa_{112}=\kappa_{122}=\kappa_{222}=m_{12}=0,\; \alpha_2\ne 0, \\
  \La{u}(1)_Y &= \spa(X_Y), \quad \text{for all other parameter values}.
\end{align}

\paragraph{Leaf 22}
\label{sec:Leaf22}
Now suppose the situation is the same as for Leaf~21, only with
\begin{align}
  m_{22}\ne 0.
\end{align}
Then, if
\begin{align}
  \kappa_{122}=\kappa_{222}=0,
\end{align}
we can set
\begin{align}
  \alpha_2\to 0,
\end{align}
by choosing a specific value for $\gamma_2$, and we arrive at Leaf~22 in
Fig.~\ref{fig:reduction-tree1}. Then there are two possible symmetry algebras,
\begin{align}
  \La{sc}_2 \oplus \La{u}(1)_Y &= \spa(X_6,X_Y), \quad \kappa_2=\kappa_{112}=m_{12}=0, \\
  \La{u}(1)_Y &= \spa(X_Y), \quad \text{for all other parameter values}.
\end{align}

\paragraph{Leaf 23}
\label{sec:Leaf23}
We now assume the same conditions as for Leaf~22, only with
\begin{align}
  \kappa_{122}\ne 0,
\end{align}
and $\alpha_2$ undetermined. We can then choose a $\gamma_2$ such that
\begin{align}
  m_{12}\to 0,
\end{align}
since $\kappa_{122}\ne 0$. Then the only symmetry algebra is $\La{u}(1)_Y$.

\paragraph{Leaf 24}
\label{sec:Leaf24}
Assume the same conditions as for Leaf~23, only with
\begin{align}
  \kappa_{222}\ne 0,
\end{align}
while $m_{12}$ and $\kappa_{112}$ are undetermined. Then we can set
\begin{align}
  m_{22}\to 0,
\end{align}
by some choice of $\gamma_2$, and we end up at Leaf~24. Again, the only possible
symmetry algebra is $\La{u}(1)_Y$.

\paragraph{Leaf 25}
\label{sec:Leaf25}
Suppose, as for Leaf~24 and previous leaves,
\begin{align}
  \lambda_{ijkl}=\lambda_{12}=\kappa_1=0, \quad \lambda_{11}\ne 0,
\end{align}
whereas
\begin{align}
  \lambda_{22}\ne 0,
\end{align}
in contrast to earlier. We can then set
\begin{align}
  \kappa_2\to 0,
\end{align}
by a choice of $\gamma_2$, without altering $\kappa_1$, since $\lambda_{12}=0$,
cf.~\eqref{E:shiftk1}. Then there are two possible symmetry algebras,
\begin{align}
  \La{so}(2) \oplus \La{u}(1)_Y &= \spa\big((X_3-X_5),X_Y\big),  \label{E:Leaf25_so2}\\
  \text{for } \kappa_{ijk}&=m_{12}=\alpha_i=0,\ \lambda_{11}=\lambda_{22},\ m_{11}=m_{22}, \label{E:Leaf25_so2Pars} \\
  \La{u}(1)_Y &= \spa(X_Y), \quad \text{for all other parameter values},
\end{align}
where the conditions $\kappa_{ijk}=0$ and $\alpha_i=0$ are understood to hold for all indices $i,j,k$. The conditions \eqref{E:Leaf25_so2Pars} correspond to the same potential as in \eqref{E:b1so2pot} in Branch~I, but with $\lambda_{1111}=0$.

\paragraph{Leaf 26}
\label{sec:Leaf26}
Assume
\begin{align}
  \lambda_{1ijk}=0, \quad \lambda_{2222}\ne 0,
\end{align}
for all indices $i,j,k$. Then, by a choice of $\gamma_2$, we can set
\begin{align}
  \kappa_{222}\to 0.
\end{align}
 Moreover, suppose
\begin{align}
  \kappa_{ijk}=0,
\end{align}
for all other indices as well, and
\begin{align}
  m_{11}=0,
\end{align}
corresponding to Leaf~26. The possible symmetry algebras are then
\begin{align}
  \La{a}(1)_1\oplus \La{u}(1)_Y &= \spa(X_1,X_2,X_Y), \\
  &\text{for } \lambda_{11}=\lambda_{12}=\kappa_1=m_{12}=\alpha_1=0, \nonumber\\
  \La{sh}_1 \oplus \La{u}(1)_Y &= \spa(X_1,X_Y), \\
  &\text{for } \lambda_{11}=\lambda_{12}=\kappa_1=m_{12}=0,\; \alpha_1\ne 0, \nonumber\\
  \La{u}(1)_Y &= \spa(X_Y), \quad \text{for all other parameter values}.
\end{align}

\paragraph{Leaf 27}
\label{sec:Leaf27}
Assume the same parameter conditions as for Leaf~26, only with
\begin{align}
  m_{11}\ne 0,
\end{align}
and we can set
\begin{align}
  \alpha_1 \to 0,
\end{align}
by a choice of $\gamma_1$ without altering $m_{11}$, as given by Leaf~27. The two possible symmetry algebras are then
\begin{align}
  \La{sc}_1 \oplus \La{u}(1)_Y &= \spa(X_2,X_Y),
  \quad \lambda_{11}=\lambda_{12}=\kappa_{1}=m_{12}=0, \\
  \La{u}(1)_Y &= \spa(X_Y), \quad \text{for all other parameter values}.
\end{align}

\paragraph{Leaf 28}
\label{sec:Leaf28}
Now, let the assumptions be the same as for Leaf~27, only with
\begin{align}
  \kappa_{111}\ne 0,
\end{align}
and $m_{11}$ and $\alpha_1$ undetermined. We can then choose a $\gamma_1$ such that
\begin{align}
  m_{11}\to 0.
\end{align}
The only possible symmetry algebra is then $\La{u}(1)_Y$.

\paragraph{Leaf 29}
\label{sec:Leaf29}
For the next case, we make the same assumption as for Leaf~28, except that
\begin{align}
  \kappa_{112}\ne 0,
\end{align}
and the parameters $\kappa_{111}$ and $m_{11}$ are completely unknown. We can then set
\begin{align}
  m_{12}\to 0,
\end{align}
by some choice of $\gamma_1$. Then the only possible symmetry algebra is again $\La{u}(1)_Y$.

\paragraph{Leaf 30}
\label{sec:Leaf30}
Suppose all parameters are the same as for Leaf~29, only with
\begin{align}
  \kappa_{122}\ne 0,
\end{align}
and with no conditions on $\kappa_{112}$ and $m_{12}$. By picking an appropriate value of $\gamma_1$, we can set
\begin{align}
  m_{22}\to 0,
\end{align}
and we find only the trivial symmetry algebra $\La{u}(1)_Y$.

\paragraph{Leaf 31}
\label{sec:Leaf31}
Finally, assume that
\begin{align}
  \lambda_{11ij}=0, \quad \lambda_{1222}\ne 0,
\end{align}
for all indices $i$ and $j$. We can then set
\begin{align}
  \kappa_{122}\to 0,\quad \kappa_{222}\to 0,
\end{align}
by first choosing an appropriate value of $\gamma_2$ that nullifies $\kappa_{122}$ and might also change
$\kappa_{222}$. Afterwards, we can pick a value of $\gamma_1$ that eliminates $\kappa_{222}$ without changing
$\kappa_{122}$, cf.~\eqref{E:shiftk122} and \eqref{E:shiftk222}. Then the only admitted Lie point symmetry algebra is $\La{u}(1)_Y$.

\section{The $\Lg{SO}(2)$-adapted SM+2S potential}
\label{sec:SymmetryAnalysisWithLgSO2AdaptedSM2SPotential}

A more transparent reduction scheme for the quartic and cubic parameters can be obtained by expressing the SM+2S potential in a polynomial basis with definite transformation properties under $\Lg{SO}(2)$ reparametrizations.

Define
\begin{align}\label{E:defz}
	z=s_1+i s_2,
\end{align}
and note that under
\begin{align}\label{E:complexTrafoz}
	z\to \exp(i\theta)z
\end{align}
the combination $z^{n-k}(z^{\ast})^{k}$ transforms as
\begin{align}\label{E:mModesBasis}
 z^{n-k}(z^{\ast})^{k}\to \exp((n-2k)\theta i)z^{n-k}(z^{\ast})^{k}.
\end{align}
Moreover, for a fixed natural number $n$ and integers $k$ satisfying $0\leq k \leq n/2$, the real and imaginary parts of the monomials $z^{n-k}(z^{\ast})^{k}$, whenever non-zero, span the space of all real homogeneous polynomials in $s_1$ and $s_2$ of degree $n$, since these polynomials are linearly independent and their number matches the dimension of that space. For $n=4$, for example, one finds
\begin{align}\label{E:quarticSO2basis}
	P_{4\Re}^{(4)} &\equiv  \Re(z^4)=s_1^4-6s_1^2 s_2^2 + s_2^4, \nn \\
	P_{4\Im}^{(4)} &\equiv  \Im(z^4)=4s_1 s_2(s_1^2-s_2^2),\nn \\
	P_{2\Re}^{(4)} &\equiv  \Re(z^3z^\ast)= (s_1^2-s_2^2)(s_1^2+s_2^2), \nn \\
	P_{2\Im}^{(4)}&\equiv  \Im(z^3z^\ast)= 2s_1s_2 (s_1^2+s_2^2),\nn \\
	r^4 &\equiv z^2(z^\ast)^2= (s_1^2+s_2^2)^2,
\end{align}
and these five polynomials define a basis for all homogeneous real quartic polynomials in $s_1$ and $s_2$.\footnote{This is related to the harmonic--radial decomposition of homogeneous polynomials; see, e.g., Theorem~5.7 in \cite{AxlerBourdonRamey2001}. Here $r^4$ is radial, $P_{4\Re}^{(4)}$ and $P_{4\Im}^{(4)}$ are harmonic, and $P_{2\Re}^{(4)}$ and $P_{2\Im}^{(4)}$ are equal to $r^2$ times harmonic polynomials.}

The $\Lg{SO}(2)$ reparametrizations
\begin{align}
	s_1 &\to \cos(\theta) s_1- \sin(\theta)s_2, \nn \\
	s_2 &\to \sin(\theta) s_1+ \cos(\theta)s_2,
\end{align}
are equivalent to \eqref{E:complexTrafoz}, with $z$ given by \eqref{E:defz}. Define the doublets
\begin{align}
	P_4^{(4)} \equiv (P_{4\Re}^{(4)}, P_{4\Im}^{(4)})^T, \qquad
	P_2^{(4)} \equiv (P_{2\Re}^{(4)}, P_{2\Im}^{(4)})^T,
\end{align}
which, by \eqref{E:mModesBasis}, transform under these reparametrizations as
\begin{align}
	P_4^{(4)}\to R(4\theta) P_4^{(4)},\qquad P_2^{(4)}\to R(2\theta) P_2^{(4)},
\end{align}
where
\begin{align}\label{E:Rmatrmmodes}
	R(m\theta)=\begin{pmatrix}
\cos(m\theta) & -\sin(m\theta) \\
\sin(m\theta) & \cos(m\theta)
\end{pmatrix}.
\end{align}

Similarly, every homogeneous real cubic polynomial in $s_1$ and $s_2$ may be written as a linear combination of
\begin{align}
	P^{(3)}_{3\Re} &\equiv \Re(z^3)=s_1^3-3 s_1 s_2^2, \nn \\
	P^{(3)}_{3\Im} &\equiv \Im(z^3)=3 s_1^2 s_2-s_2^3, \nn \\
	P^{(3)}_{1\Re} &\equiv \Re(z^2z^\ast)= s_1 \left(s_1^2+s_2^2\right)\nn \\
	P^{(3)}_{1\Im} &\equiv \Im(z^2z^\ast)= s_2 \left(s_1^2+s_2^2\right).
\end{align}
Let
\begin{align}
	P_3^{(3)} \equiv (P_{3\Re}^{(3)}, P_{3\Im}^{(3)})^T, \qquad
P_1^{(3)}	\equiv (P_{1\Re}^{(3)}, P_{1\Im}^{(3)})^T,
\end{align}
so that, again by \eqref{E:mModesBasis},
\begin{align}
	P_3^{(3)}\to R(3\theta) P_3^{(3)},\qquad P_1^{(3)}\to R(\theta) P_1^{(3)}.
\end{align}

Likewise, every homogeneous real quadratic polynomial in $s_1$ and $s_2$ may be written as a sum of
\begin{align}
	P^{(2)}_{2\Re} &\equiv \Re(z^2)= s_1^2-s_2^2, \nn \\
	P^{(2)}_{2\Im} &\equiv \Im(z^2)= 2 s_1 s_2, \nn \\
	r^2 &\equiv z z^\ast= s_1^2+s_2^2,
\end{align}
where the doublet
\begin{align}
	P^{(2)}_2=(P^{(2)}_{2\Re},P^{(2)}_{2\Im})^T
\end{align}
transforms as
\begin{align}
	P^{(2)}_2\to R(2\theta)P^{(2)}_2,
\end{align}
under a $\Lg{SO}(2)$ reparametrization.

The most general SM+2S potential may then be written in $\Lg{SO}(2)$-adapted form as
\begin{align}\label{E:SO2adPot}
	V_\text{SM2S}&=V_\text{SM}
	+ \bar{\alpha}_{12} s  \nn \\
  &+m_1 r^2+\bar{m}_{23}P^{(2)}_2  \nn \\
	&+\bar{\kappa}_{12} s \Phi^\dag \Phi +\bar{\kappa}_{34} P_1^{(3)}+ \bar{\kappa}_{56} P_3^{(3)} \nn \\
	&+\lambda_1r^2\Phi^\dag \Phi+ \bar{\lambda}_{23} P^{(2)}_2\Phi^\dag \Phi \nn \\
	&\lambda_4 r^4+ \bar{\lambda}_{56} P_2^{(4)}+\bar{\lambda}_{78} P_4^{(4)},
\end{align}
with barred quantities defined as parameter doublets by
\begin{align}
	\bar{p}_{ij}=(p_i,p_j).
\end{align}
The unbarred parameters $\mu^2, \lambda, m_1, \lambda_1, \lambda_4$ are invariant under $\Lg{SO}(2)$ reparametrizations, whereas the barred parameters must transform as
\begin{align}\label{E:trafoParDoubletsSO2}
	 \bar{\alpha}_{12}, \bar{\kappa}_{12}, \bar{\kappa}_{34}
	&\to \bar{p}\, R^T(\theta),\nn \\
	\bar{m}_{23}, \bar{\lambda}_{23}, \bar{\lambda}_{56}
	&\to \bar{p}\, R^T(2\theta),\nn \\
	\bar{\kappa}_{56}
	&\to \bar{\kappa}_{56}\, R^T(3\theta),\nn \\
	\bar{\lambda}_{78}
	&\to \bar{\lambda}_{78}\, R^T(4\theta),
\end{align}
in order to keep the potential invariant, where $\bar{p}$ denotes any parameter doublet appearing on the left-hand side.
The relations between the old and new parameters are
\begin{align}
m_{11} &= m_1 + m_2, &
m_{12} &= 2m_3, &
m_{22} &= m_1 - m_2, \nn\\
\kappa_{111} &= \kappa_3 + \kappa_5, &
\kappa_{112} &= \kappa_4 + 3\kappa_6, &
\kappa_{122} &= \kappa_3 - 3\kappa_5, &
\kappa_{222} &= \kappa_4 - \kappa_6, \nn\\
\lambda_{11} &= \lambda_1 + \lambda_2, &
\lambda_{12} &= 2\lambda_3, &
\lambda_{22} &= \lambda_1 - \lambda_2, \nn\\
\lambda_{1111} &= \lambda_4 + \lambda_5 + \lambda_7, &
\lambda_{1112} &= 2(\lambda_6 + 2\lambda_8),\nn \\
\lambda_{1122} &= 2(\lambda_4 - 3\lambda_7), &
\lambda_{1222} &= 2(\lambda_6 - 2\lambda_8), &
\lambda_{2222} &= \lambda_4 - \lambda_5 + \lambda_7.
\end{align}

Let
\begin{align}
M_1=
\begin{pmatrix}
4\lambda_4+3\lambda_5 & 3\lambda_6\\
3\lambda_6 & 4\lambda_4-3\lambda_5
\end{pmatrix},
\qquad
M_3=
\begin{pmatrix}
\lambda_5+4\lambda_7 & -\lambda_6+4\lambda_8\\
\lambda_6+4\lambda_8 & \lambda_5-4\lambda_7
\end{pmatrix}.
\end{align}
Under shifts $\gamma=(\gamma_1,\gamma_2)^T$, with $s\to s-\gamma$, the cubic parameters transform as
\begin{align}\label{E:highDegparsShiftTrafo}
	 \bar{\kappa}_{34} \to
	\bar{\kappa}_{34} -M_1\gamma, \quad
	\bar{\kappa}_{56} \to
	\bar{\kappa}_{56} -M_3 \gamma.
\end{align}
Moreover, let
\begin{align}
	L=\begin{pmatrix}
\lambda_1+\lambda_2 & \lambda_3\\
\lambda_3 & \lambda_1-\lambda_2
\end{pmatrix},
\end{align}
then
\begin{align}\label{E:highDegparsShiftTrafo2}
	\bar{\kappa}_{12}\to \bar{\kappa}_{12} -2L\gamma,\quad
	\mu^2 \to \mu^2 + \bar{\kappa}_{12} \gamma - \gamma^TL \gamma.
\end{align}
Finally, the quadratic and linear parameters transform as
\begin{align}\label{E:lowDegparsTrafo}
	m_1&\to m_1+2 \gamma _1 \left(\kappa _3-3 \gamma _2 \lambda _6\right)+\gamma _2 \left(\gamma _2
   \left(3 \lambda _5-4 \lambda _4\right)+2 \kappa _4\right)-\left(\gamma _1^2 \left(4 \lambda _4+3
   \lambda _5\right)\right), \nn \\
	m_2 &\to m_2+\gamma _1 \left(3 \left(\kappa _5-4 \gamma _2 \lambda _8\right)+\kappa
   _3\right)\nn \\ &+\gamma _2 \left(\gamma _2 \left(2 \lambda _4-3 \lambda _5+6 \lambda _7\right)-\kappa _4+3
   \kappa _6\right)-\left(\gamma _1^2 \left(2 \lambda _4+3 \lambda _5+6 \lambda _7\right)\right), \nn \\
	m_3 &\to m_3+\gamma _2 \left(-4 \gamma _1 \left(\lambda _4-3 \lambda _7\right)+\kappa _3-3 \kappa
   _5\right)\nn \\ &+\gamma _1 \left(-3 \gamma _1 \left(\lambda _6+2 \lambda _8\right)+\kappa _4+3 \kappa
   _6\right)-3 \gamma _2^2 \left(\lambda _6-2 \lambda _8\right),  \\
	\alpha _1 &\to \alpha _1+2 \gamma _1 m_1+2 \gamma _1 m_2+2 \gamma _2 m_3+3 \gamma _1^2 \kappa _3+3
   \gamma _1^2 \kappa _5+2 \gamma _2 \gamma _1 \kappa _4+6 \gamma _2 \gamma _1 \kappa _6+\gamma _2^2
   \kappa _3-3 \gamma _2^2 \kappa _5\nn \\ &-4 \gamma _1^3 \lambda _4-4 \gamma _1^3 \lambda _5-4 \gamma _1^3
   \lambda _7-6 \gamma _2 \gamma _1^2 \lambda _6-12 \gamma _2 \gamma _1^2 \lambda _8-4 \gamma _2^2
   \gamma _1 \lambda _4+12 \gamma _2^2 \gamma _1 \lambda _7-2 \gamma _2^3 \lambda _6+4 \gamma _2^3
   \lambda _8, \nn \\
	\alpha _2 &\to \alpha _2+2 \gamma _1 m_3+2 \gamma _2 m_1-2 \gamma _2 m_2+\gamma _1^2 \kappa _4+3 \gamma
   _1^2 \kappa _6+2 \gamma _2 \gamma _1 \kappa _3-6 \gamma _2 \gamma _1 \kappa _5+3 \gamma _2^2 \kappa
   _4-3 \gamma _2^2 \kappa _6\nn \\ &-2 \gamma _1^3 \lambda _6-4 \gamma _1^3 \lambda _8-4 \gamma _2 \gamma _1^2
   \lambda _4+12 \gamma _2 \gamma _1^2 \lambda _7-6 \gamma _2^2 \gamma _1 \lambda _6+12 \gamma _2^2
   \gamma _1 \lambda _8-4 \gamma _2^3 \lambda _4+4 \gamma _2^3 \lambda _5-4 \gamma _2^3 \lambda _7, \nn
\end{align}
with no simple matrix representation.

When reducing the $\Lg{SO}(2)$-adapted potential \eqref{E:SO2adPot}, one may first ask whether the quartic parameters transforming with the highest $m$-value, cf.~\eqref{E:Rmatrmmodes}, namely $\bar{\lambda}_{78}$, vanish. If so, one may proceed to the quartic parameters with the next-highest $m$-value, namely $\bar{\lambda}_{56}$, and continue in the same way. If instead $\bar{\lambda}_{78}$ is non-zero, one may eliminate, for example, $\lambda_8$ by a suitable choice of $\theta$, cf.~\eqref{E:trafoParDoubletsSO2}. Since $\bar{\lambda}_{78}\neq 0$, this implies $\lambda_7\ne 0$. One may then proceed to eliminate two cubic parameters by shifts, cf.~\eqref{E:highDegparsShiftTrafo} and \eqref{E:highDegparsShiftTrafo2}, provided that at least one of the matrices $M_3$, $M_1$, or $L$ has non-vanishing determinant. In such cases, the determining equations are easily solved, and the corresponding symmetry algebras are readily determined.

However, suppose, for example, that one eliminates $\lambda_8$ by a rotation as described above, and that the transformed parameters then satisfy $\det(M_3)=\det(M_1)=\det(L)=0$. Then non-linear relations arise among the parameters entering these matrices. Solving the determining equations under these conditions yields 30 apparently distinct algebras, of which all but four are exotic; collectively, these exotic algebras involve potential parameters of every degree. Consequently, in order to determine the inequivalent symmetry algebras, one must split into different cases and eliminate parameters using the shift formulas \eqref{E:highDegparsShiftTrafo}, \eqref{E:highDegparsShiftTrafo2}, and \eqref{E:lowDegparsTrafo}, subject to the non-linear constraints $\det(M_3)=\det(M_1)=\det(L)=0$. Such cases complicate the reduction scheme for the \(\Lg{SO}(2)\)-adapted potential. For this reason, we instead work with the original potential
\eqref{E:SM2Spot}, with its more natural parameters.

\bibliographystyle{JHEP}

\bibliography{ref_arXiv_v2}

\end{document}